\documentclass[12pt]{article}
\usepackage[utf8]{inputenc}
\usepackage[left=1in,top=0.8in,right=1in,bottom=0.7in,head=.2in]{geometry}

\usepackage[bookmarksopen=true,bookmarks=true,pdfencoding=auto,psdextra,colorlinks,
]{hyperref}
\hypersetup{ 
    colorlinks,
    linkcolor={blue!80!black},
    citecolor={green!40!black},
    urlcolor={red!50!black}
}
\usepackage{bookmark}
\usepackage{mystyle}
\usepackage{natbib}

\allowdisplaybreaks
\usepackage{xspace}

\usepackage{xcolor}

\usepackage[font=small,labelfont=bf]{caption}
\usepackage{setspace}

\usepackage{tabularx}
\usepackage{array}
\newcolumntype{Y}{>{\centering\arraybackslash}X}
\newcolumntype{P}{>{\raggedleft\arraybackslash}X}
\newcolumntype{C}{ >{\centering\arraybackslash} m{4cm} }
\newcolumntype{D}{ >{\centering\arraybackslash} m{1cm} }

\usepackage{multirow}

\usepackage{algorithm}
\usepackage{algorithmic}

\def\spacingset#1{\renewcommand{\baselinestretch}%
{#1}\small\normalsize} \spacingset{1}

\usepackage{tikz}

\makeatletter
\newcommand*\bigcdot{\mathpalette\bigcdot@{.5}}
\newcommand*\bigcdot@[2]{\mathbin{\vcenter{\hbox{\scalebox{#2}{$\m@th#1\bullet$}}}}}
\makeatother

\newcommand{\footremember}[2]{%
    \footnote{#2}
    \newcounter{#1}
    \setcounter{#1}{\value{footnote}}%
}
\newcommand{\footrecall}[1]{%
    \footnotemark[\value{#1}]%
}
\usepackage{setspace}

\title{Simultaneous inference for generalized linear models\\with unmeasured confounders}
\author{Jin-Hong Du\footremember{cmustats}{Department of Statistics and Data Science, Carnegie Mellon University, Pittsburgh, PA 15213, USA.}\footremember{cmumld}{Machine Learning Department, Carnegie Mellon University, Pittsburgh, PA 15213, USA.} \and 
Larry Wasserman\footrecall{cmustats} \footrecall{cmumld} \and 
Kathryn Roeder\footrecall{cmustats} \footremember{cmucbd}{Computational Biology Department, Carnegie Mellon University, Pittsburgh, PA 15213, USA.}}
\date{\today}

\begin{document}

\maketitle

\begin{abstract}
    Tens of thousands of simultaneous hypothesis tests are routinely performed in genomic studies to identify differentially expressed genes. However, due to unmeasured confounders, many standard statistical approaches may be substantially biased. This paper investigates the large-scale hypothesis testing problem for multivariate generalized linear models in the presence of confounding effects. Under arbitrary confounding mechanisms, we propose a unified statistical estimation and inference framework that harnesses orthogonal structures and integrates linear projections into three key stages. It begins by disentangling marginal and uncorrelated confounding effects to recover the latent coefficients. Subsequently, latent factors and primary effects are jointly estimated through lasso-type optimization. Finally, we incorporate projected and weighted bias-correction steps for hypothesis testing. Theoretically, we establish the identification conditions of various effects and non-asymptotic error bounds. We show effective Type-I error control of asymptotic $z$-tests as sample and response sizes approach infinity. Numerical experiments demonstrate that the proposed method controls the false discovery rate by the Benjamini-Hochberg procedure and is more powerful than alternative methods. By comparing single-cell RNA-seq counts from two groups of samples, we demonstrate the suitability of adjusting confounding effects when significant covariates are absent from the model.
\end{abstract}

\noindent%
\vfill

% \tableofcontents
% \doublespacing

\clearpage
% \begin{bibunit}[apalike]
\section{Introduction}\label{sec:intro}
To discover genes that are differentially expressed under different experimental conditions or across groups of samples, large numbers of simultaneous hypothesis tests must be performed. These tests are made more challenging by the presence of unmeasured covariates that bias the analyses. In 2007, \citet{Leek:2007,leek2008general} presented their pathbreaking ``surrogate variable" approach to control for unmeasured confounding effects in differential expression (DE) studies using microarray data. These confounders go by various names in the literature, including batch effects, surrogate variables, latent effects, or simply unwanted variations \citep{Leek:2010,Gagnon-Bartsch:2012,sun2012multiple}.
Adjusting for confounding effects is crucial because they may distort the correct null distribution of the test statistics, and consequently, standard statistical approaches can be substantially biased \citep{wang2017confounder,mckennan2019accounting}. 
Due to burgeoning developments in the genomics field, DE testing has been dramatically expanded to include a variety of genomic readouts beyond microarray, in which the normality of the observed counts rarely holds.
Inspired by modern-day omic studies, the concerns about confounding are more urgent than ever, and there is a pressing need to adapt statistical approaches to changing data types.

The problem of confounder adjustment has been an important topic in statistics in recent years.
To characterize the confounding effects, the pioneering work in this field assumes a linear model $\bY=\bX\bB^{\top} + \bZ\bGamma^{\top} + \bE$, where $\bY\in\mathbb{R}^{n\times p}$ is the gene expression matrix, $\bX\in\mathbb{R}^{n\times d}$ is the measured covariate matrix, $\bB\in\mathbb{R}^{p\times d}$ is the direct effect to be estimated, $\bZ\in\mathbb{R}^{n\times r}$ is the latent factor matrix, $\bGamma\in\mathbb{R}^{p\times r}$ is the latent factor loading, and $\bE\in\mathbb{R}^{n\times p}$ is the additive noise.
The early investigations study the statistical inference problem under this model by further imposing a linear relationship between $\bX$ and $\bZ$, assuming either $\bX$ causes $\bZ$ as in \Cref{fig:causal-diagram}(a), i.e., $\bZ$ is a hidden mediator \citep{leek2008general,wang2017confounder,gerard2020empirical}, or $\bZ$ causes $\bX$ as in \Cref{fig:causal-diagram}(b), i.e., $\bZ$ is a hidden confounder \citep{guo2022doubly,sun2022decorrelating}.

In the presence of hidden mediators, where the observed covariates are the cause of the hidden variables, \citet{wang2017confounder} and \citet{gerard2020empirical} study the statistical inference problem for multiple outcomes ($p>1$) by assuming a linear relationship between the observed variables and the hidden variables.
In the context of hidden confounders and a single outcome ($p=1$), \citet{guo2022doubly} propose a doubly debiased lasso estimator and establish asymptotic normality; \citet{cevid2020spectral} propose a spectral de-confounding method; and \citet{sun2022decorrelating} analyze non-asymptotic and asymptotic false discovery control with high-dimensional covariates.
    
    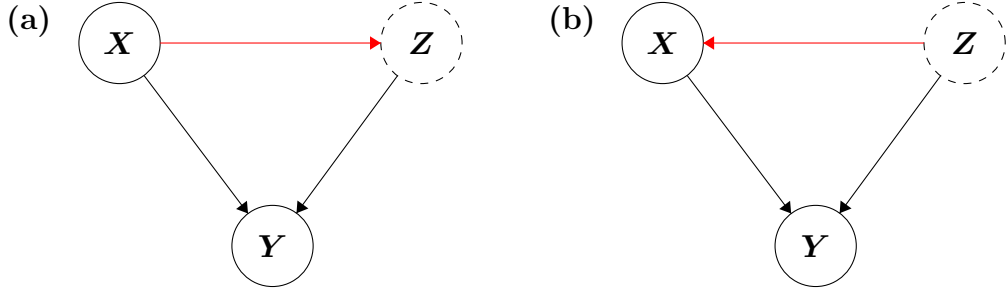
\begin{figure}[!t]
        \centering
        \begin{tikzpicture}[scale=0.18]
        \tikzstyle{every node}+=[inner sep=0pt]
        \node[font=\bf] at (29,-20) {(a)};
        
        \draw [black] (35.7,-21.6) circle (3);
        \draw (35.7,-21.6) node {$\bX$};
        \draw [black,dashed] (58,-21.6) circle (3);
        \draw (58,-21.6) node {$\bZ$};
        \draw [black] (47,-36.6) circle (3);
        \draw (47,-36.6) node {$\bY$};
        \draw [black] (56.23,-24.02) -- (48.77,-34.18);
        \fill [black] (48.77,-34.18) -- (49.65,-33.83) -- (48.84,-33.24);
        \draw [black] (37.51,-24) -- (45.19,-34.2);
        \fill [black] (45.19,-34.2) -- (45.11,-33.26) -- (44.31,-33.87);
        \draw [red] (38.7,-21.6) -- (55,-21.6);
        \fill [red] (55,-21.6) -- (54.2,-21.1) -- (54.2,-22.1);
        \end{tikzpicture}\hspace{1cm}
        \begin{tikzpicture}[scale=0.18]
        \tikzstyle{every node}+=[inner sep=0pt]
        \node[font=\bf] at (29,-20) {(b)};
        \draw [black] (35.7,-21.6) circle (3);
        \draw (35.7,-21.6) node {$\bX$};
        \draw [black,dashed] (58,-21.6) circle (3);
        \draw (58,-21.6) node {$\bZ$};
        \draw [black] (47,-36.6) circle (3);
        \draw (47,-36.6) node {$\bY$};
        \draw [black] (56.23,-24.02) -- (48.77,-34.18);
        \fill [black] (48.77,-34.18) -- (49.65,-33.83) -- (48.84,-33.24);
        \draw [black] (37.51,-24) -- (45.19,-34.2);
        \fill [black] (45.19,-34.2) -- (45.11,-33.26) -- (44.31,-33.87);
        \draw [red] (55,-21.6) -- (38.7,-21.6);
        \fill [red] (38.7,-21.6) -- (39.5,-22.1) -- (39.5,-21.1);
        \end{tikzpicture}
        \caption{Causal diagrams on the generative models illustrating the relationship between the covariate $\bX$, the latent variable $\bZ$, and the response $\bY$.
        \textbf{(a)} $\bZ$ is a hidden mediator when $\bX$ causes $\bZ$.
        \textbf{(b)} hidden confounder when $\bZ$ causes $\bX$.
        Note that we do not require knowledge of the relationship between $\bX$ and $\bZ$ for the analysis in this paper.}\label{fig:causal-diagram}
    \end{figure}

More recently, methods for estimating primary effects extend beyond linear dependence structures between covariates and confounders.
For instance, \citet{jiang2022treatment} model the interaction between the covariates and the confounders, and projection-based methods are employed to estimate the primary effects under arbitrary dependency \citep{lee2017improved,bing2022adaptive,mckennan2022estimating}.
For statistical inference, \citet{mckennan2019accounting} propose an estimator that is asymptotically equivalent to the ordinary least squares estimators obtained when every covariate is observed, and \citet{bing2022inference} establish asymptotic normality, efficiency, and consistency.
    
The applicability of the aforementioned methods to the nonlinear model remains challenging.
Limited research has been done to address adjustments for confounding effects under the setting of \emph{arbitrary confounding mechanisms}, \emph{nonlinear models}, and \emph{multiple outcomes}.
For empirical studies, \citet{salim2022ruv} propose a heuristic algorithm that utilizes a pseudo-replicate design matrix and negative control genes to remove unwanted variations.
For theoretical analysis, to the best of our knowledge, the related literature that explores slightly broader settings is limited to \citet{feng2020causal}, who studies nonlinear factor models concerning treatment effects with a single outcome by PCA-based matching, and \citet{ouyang2022high}, who study the generalized linear models with a single outcome and linear hidden confounders.
However, both of these works assume the covariates are some functions of the unobserved confounders.

Our work is inspired by the rapid developments in the field of genomics, particularly single-cell omics \citep{review:2023}.
For example, CRISPR perturbations with single-cell sequencing readouts have promised extraordinary scientific insight \citep{Kampmann:2020,Hong:2023,Cheng:2023}; due to the sparsity of outcomes and the nature of the molecular readout, these data are not suitable for analysis by linear models under Gaussianity assumptions \citep{sarkar2021separating,barry2023exponential}.
Hence, our development of generalized linear models for confounding is timely.

In this paper, we adopt the term ``confounder'' to encompass a broad category of latent variables, including both mediators and confounders, as defined in the context of causal inference literature.
The purpose of this paper is to derive valid simultaneous inference for multivariate generalized linear models in the presence of unmeasured confounding effects.
Existing methods in this domain typically focus on Gaussian linear models \citep{wang2017confounder,bing2022adaptive,bing2022inference} or necessitate direct modeling of the relationship between covariates and confounders \citep{feng2020causal,ouyang2022high}.
To the best of our knowledge, the proposed method is the first estimation and inference framework capable of (1) accommodating general relationships between observed covariates and unmeasured confounders, allowing for \emph{arbitrary confounding mechanisms}; (2) utilizing generalized linear models, allowing for \emph{nonlinear modeling}; and (3) incorporating information from \emph{multiple outcomes}.
Our approach leverages the orthogonal structures inherent in the problem, incorporating linear projection techniques into both estimation and inference processes to effectively mitigate confounding effects and elucidate primary effects. Notably, it exhibits significant utility in high-dimensional sparse count data, as demonstrated through the analysis of single-cell datasets on systemic lupus erythematosus disease in~\Cref{sec:case-studeis}.

Our proposed procedure \textsc{gcate} (generalized confounder adjustment for testing and estimation) consists of three main steps.
In the first step, we use joint maximum likelihood estimation \citep{chen2019joint,chen2022determining} to obtain the initial estimate of the marginal effects and uncorrelated latent components by projecting the latent factor $\bZ$ to the orthogonal space of $\bX$, from which we recover the column space of $\bGamma$.
In the second step, we use a similar strategy to obtain the estimates of both $\bZ$ and primary effect $\bB$, by constraining the latter to be orthogonal to the estimated latent coefficients $\hat{\bGamma}$ and using $\ell_1$-regularization to encourage sparsity.
Lastly, the valid inference is guaranteed by a bias-corrected estimator of $\hat{\bB}$, which innovates a link-specific weight function, similar to \cite{javanmard2014confidence,cai2021statistical}, while incorporating projection-based score adjustments that combine the information from multivariate responses.

In our theoretical framework of confounded generalized linear models, we establish conditions for identifying the latent coefficients and direct effects. Furthermore, we provide non-asymptotic estimation error bounds for these estimated quantities in high-dimensional scenarios where both the sample size $n$ and response size $p$ tend to be infinity. In particular, we derive element-wise $\ell_2$-norm and $\ell_1$-norm bounds for the estimation error of the primary effects by effectively controlling the column-wise estimation errors of the latent components. Lastly, we demonstrate the asymptotic normality of our proposed bias-corrected estimator and show the proper control of statistical errors, thereby enabling the construction of valid confidence intervals and hypothesis tests.

\paragraph{Organization and Notation.}
In \Cref{sec:background}, we set up our modeling framework, which extends existing results in the literature to the generalized linear model setting. 
In \Cref{sec:estimation}, we describe our strategy for estimation and establish bounds on the estimation error of the parameters of interest.
In \Cref{sec:inference}, we motivate and construct asymptotically valid confidence intervals and hypothesis tests.  
Finally, in \Cref{sec:simulation} and \Cref{sec:case-studeis}, we study the empirical behavior of our estimators in realistic simulations and a study of gene expression in lupus patients.
Technical proof of the results is provided in the supplementary material.

Throughout our exposition, we will use the following notational conventions. 
For any matrix $\bA\in\RR^{n\times p}$, we use $\ba_{i}$, $\bA_{j}$, and $a_{ij}$ to denote its $i$th row, $j$th column, and $(i,j)$-th entry, respectively, for $i=1,\ldots,n$, $j=1,\ldots,p$.
For any matrix $\bA\in\RR^{n\times p}$ with full column rank, let $\cP_{\bA} = \bA(\bA^{\top}\bA)^{-1}\bA^{\top}$ and $\cP_{\bA}^{\perp} = \bI_p - \cP_{\bA}$ be the orthogonal projection matrices on the $\bA$'s column space and its orthogonal space, respectively.
For any square matrix $\bA\in\RR^{n\times n}$, $\lambda_i(\bA)$ denotes its $i$th largest eigenvalue.
The symbol ``$\odot$'' denotes the Hadamard product.
We use ``$o$'' and ``$\cO$'' to denote the little-o and big-O notations and let ``$\op$'' and ``$\Op$'' be their probabilistic counterparts.
For sequences $\{a_n\}$ and $\{b_n\}$, we write $a_n\ll b_n$ or $b_n\gg a_n$ if $a_n=o(b_n)$; $a_n\lesssim b_n$ or $b_n\gtrsim a_n$ if $a_n=\cO(b_n)$; and $a_n\asymp b_n$ if $a_n=\cO(b_n)$ and $b_n=\cO(a_n)$.
Convergence in distribution and probability are denoted by ``$\dto$'' and ``$\pto$''.

\section{Modeling Differential Expression}\label{sec:background}

In the context of DE testing and related applied problems, the outcome variable can be a variety of measures, including gene expression, protein abundance, and open chromatin. For simplicity of exposition, we will describe our methods in the context of tests for differential gene expression. These tests aim to contrast outcomes from case versus control samples, wherein case and control observations may be derived from various study designs, spanning the spectrum from diseased versus healthy subjects to perturbed versus non-targeted cells.

    \subsection{Generalized linear model with hidden confounders}\label{subsec:glm}
        Suppose the gene expression $\by\in\RR^p$ is a $p$-dimensional random vector containing conditional independent entries from a one-dimensional exponential family with density:
        \begin{align*}
            p(y_{j}\mid \theta_j) = h(y_{j} )\;\exp\left(y_{j} \theta_{j} - A(\theta_{j}) \right),
        \end{align*}
        where $\theta_{j}\in\RR$ is the \emph{natural parameter}, and $A(\cdot)$ and $h(\cdot)$ are functions that depend on the member of the exponential family.        
        We restrict ourselves to the regular families whose natural parameter space is a nonempty open set and $A$ is continuously thrice differentiable, which is satisfied by most common exponential families as summarized in \Cref{tab:exp-family}.
        Because the one-parameter exponential family is minimal, the natural parameter space is convex, and the \emph{log-partition function} $A$ is strictly convex.
        If we know the distribution of $\by$, then $\btheta\in\RR^p$ is a unique solution to the equation $\EE[\by\,|\,\btheta] = A'(\btheta)$, where $A'$ is the first derivative of $A$ and applied element-wise to $\btheta$; equivalently, $\btheta = A'^{-1}(\EE[\by\,|\,\btheta])$.
        In other words, we can recover $\btheta$ based on the information of the first moment of $\by$ and the log-partition function~$A$.

        To associate multiple outcomes with both covariates and hidden confounders, one can naturally consider the generalized linear model, where the natural parameters are linear functions of both the observed covariates $\bx\in\RR^d$ and the unmeasured confounder $\bz\in\RR^r$:
        \begin{align*}
            \btheta_{p\times 1} &= \bB_{p\times d}\bx_{d\times 1} + \bGamma_{p\times r}\bz_{r\times 1}.
        \end{align*}
        Here, $\bB$ and $\bGamma$ are the linear coefficients.
        Denote $\bD\bx$ the linear projection of $\bz$ onto $\bx$, where $\bD:=\EE[\bz\bx^{\top}]\EE[\bx\bx^{\top}]^{-1}\in\RR^{r\times d}$ is the projection coefficient and $\bw = \bz -\bD\bx $ is the residual uncorrelated with $\bx$.
        To see how $\bz$ may affect the inference on $\bB$, note that
        \begin{align}
            \btheta &= (\bB + \bGamma\bD)\bx + \bGamma\bw.\label{eq:theta-projz}
        \end{align}        
        When $\by$ is normally distributed, the confounding effects occur even when regressing the mean response $\btheta=\EE[\by\mid\btheta]$ on $\bx$, which yields the \emph{confounded coefficient} $\bB + \bGamma\bD$ while the \emph{direct effect} of interest is $\bB$.
        When $\by$ comes from general exponential families, the confounding effects are more intractable because all moments and cumulants of the response may be affected by the colinearity of $\bx$ and $\bz$.

        In the context of genomic analysis, the problem of confounding is more severe when the number of available covariates is limited.
        In particular, one typically encounters high-dimensional scenarios characterized by a substantial number of genes, often surpassing the available numbers of covariates and hidden confounding factors.
        In this paper, we also consider such a challenging scenario where the number of genes is much larger than the numbers of the observed covariates and the unmeasured confounders, namely, $p\gg d$ and $p\gg r$.
        Under such challenges, the first natural and essential question one may ask is whether there is any hope to disentangle the confounding effects and identify the direct effects.

        The answer to this inquiry is affirmative.
        On the one hand, the column space of $\bGamma$ can be identified up to rotations if $\Cov(\bGamma\bw)=\bGamma\bSigma_w\bGamma^{\top}$ has rank $r$, where $\bSigma_w =\Cov(\bw)$ is the covariance of the uncorrelated latent factors.
        This fact originates from basic principles in linear algebra, frequently employed in factor analysis \citep{bai2003inferential,bai2012statistical}.
        On the other hand, once the column space of $\bGamma$ is known, one can apply the orthogonal projections to remove the confounding effects based on \eqref{eq:theta-projz}:
        \begin{align*}
            \cP_{\bGamma}^{\perp} \btheta &= \cP_{\bGamma}^{\perp}\bB\bx,
        \end{align*}
        where $\cP_{\bGamma}=\bGamma(\bGamma^{\top}\bGamma)^{-1}\bGamma^{\top}$ and $\cP_{\bGamma}^{\perp}=\bI_p-\cP_{\bGamma}$ are the orthogonal projection matrices that project vectors on to the image of $\bGamma$ and the orthogonal complement of $\bGamma$, respectively.
        By regressing $\cP_{\bGamma}^{\perp} \btheta$ on $\bx$, we can further obtain the unconfounded primary effects $\cP_{\bGamma}^{\perp}\bB=\cP_{\bGamma}^{\perp} \EE[\btheta\bx^{\top}]\EE[\bx\bx^{\top}]^{-1}$.
        However, the other component, $\cP_{\bGamma}\bB$, often poses challenges in identifiability unless additional conditions are imposed. 
        Typically, extra assumptions on the spectrum of $\bGamma$ and the sparsity of $\bB$ are necessary, to assert that $\bGamma$ and $\bB$ are asymptotically orthogonal, in the sense that $\cP_{\bGamma}\bB$ is negligible \citep{lee2017improved,wang2017confounder,bing2022inference}.
        In that case, $\bB$ can be well approximated by $\cP_{\bGamma}^{\perp} \bB$. 
        Below, we give one sufficient identification condition.
        \begin{proposition}[Identification of $\bB$]\label{prop:iden-B}
            Suppose there exists a sequence $\{\tau_p\}_{p\in\NN}$ that is uniformly lower bounded away from zero such that the following conditions hold:
            \begin{align}
                \lambda_r(\bGamma\bSigma_w\bGamma^{\top})\geq \tau_p,\qquad \max_{1\leq j\leq p}(\bGamma\bSigma_w\bGamma^{\top})_{jj}=\cO(1),\qquad \max_{1\leq \ell\leq d}\|\bB_{\ell}\|_1=o(\tau_p). \label{eq:cond-iden-B}
            \end{align}
            Then as $p$ tends to infinity, it follows that $\bB = \cP_{\bGamma}^{\perp}\bB + \bo(1)$ and $\|\cP_{\bGamma}\bB\|_{\fro} \lesssim \sqrt{p}\|\bB\|_{1,1}/\tau_p $, where $\|\cdot\|_{1,1}$ is the element-wise $\ell_1$-norm.
            Further, $\cP_{\bGamma }$ and $\bB$ can be identified from the first two moments of $\bx,\by$ asymptotically.
        \end{proposition}

        As hinted above, the lower bound condition of $\bGamma\bSigma_w\bGamma^{\top}$'s spectrum in \eqref{eq:cond-iden-B} ensures that the column space of $\bGamma$ can be identified up to rotations.
        The second condition guarantees that the diagonal entries of $\bGamma\bSigma_w\bGamma^{\top}$ are balanced.
        Finally, the last condition in \eqref{eq:cond-iden-B} can hold when $\bB$ is sparse, and its entry is bounded.
        Compared to \citet[Theorem 1]{bing2022inference} where the response $\by$ is normally distributed, the identifiability condition of \Cref{prop:iden-B} applies for exponential families, which is of much broader generality.
        Furthermore, the smallest eigenvalue of $\bGamma\bSigma_w\bGamma^{\top}$ can grow at a specific rate $\tau_p$ in \Cref{prop:iden-B}.
        When $\tau_p=p$, we can recover the result in \citet[Theorem 1]{bing2022inference}.
        Lastly, we also provide a norm bound for the residual $\cP_{\bGamma}\bB$, which is helpful for later analysis of the estimation errors.

    \subsection{Random samples}\label{subsec:random-samples}
    While the preceding identification outcomes are applicable when population moments are known, practical scenarios involve the observation of independent and identically distributed (i.i.d.) samples.
    Consequently, statistical estimation becomes imperative to disentangle direct effects from the confounding effects based on samples.
    Suppose $(\bx_i,\by_i)$ for $i=1,\ldots,n$ are $n$ i.i.d. samples coming from the same distribution as $(\bx,\by)$, and let $\bX\in\RR^{n\times d}$ and $\bY\in\RR^{n\times p}$ denote the design matrix and the gene expression matrix, respectively.
    The expression $y_{ij}$ of the $i$th observation and the $j$th gene has the density:
    $$p(y_{ij}\mid \theta_{ij}) = h(y_{ij} )\exp\left(y_{ij} \theta_{ij} - A(\theta_{ij})\right),$$        
    where $\theta_{ij}$ is the natural parameter.
    In matrix form, the natural parameters decompose as
    \[\bTheta = \bX\bB^{\top} + \bZ\bGamma^{\top} ,\]
    where $\bB\in\RR^{p\times d}$, $\bZ\in\RR^{n\times r}$, and $\bGamma\in\RR^{p\times r}$ are unknown.
    Note that $y_{ij}$'s are conditionally independent given the natural parameter $\bTheta$.

    One natural way to estimate the unknown variable $\bZ$ and parameters $(\bB,\bGamma)$ is to perform maximum likelihood estimation.
    Ignoring the constant terms, the negative log-likelihood function of $\bY$ is given by
    \begin{align}
        \cL(\bTheta) = \cL(\bB,\bZ,\bGamma) &= -\frac{1}{n}\sum_{i=1}^n\sum_{j=1}^p (y_{ij}\theta_{ij} - A(\theta_{ij})).\label{eq:loglikelihood}
    \end{align}
    The second notation $ \cL(\bB,\bZ,\bGamma)$ reflects the dependence of $\bTheta$ on the model parameters $\bB$ and unknown quantities $\bGamma,\bZ$.
    Addressing the challenges of nonconvexity and high dimensionality requires developing efficient algorithms to estimate these unknown quantities and analyze their statistical properties.

    To overcome the difficulty of estimation, one critical observation comes from the projection-based decomposition and \Cref{prop:iden-B}:
    \begin{align*}
        \bTheta &= \bX\bB^{\top} + \bZ\bGamma^{\top} \\
        &= (\bX\bB^{\top}\cP_{\bGamma}^{\perp} + \bX\bB^{\top}\cP_{\bGamma} )+ (\bX\bD^{\top}\bGamma^{\top} + \bW \bGamma^{\top})\\
        &= \bX\bB^{\top}\cP_{\bGamma}^{\perp} + \cP_{\bX}\bZ\bGamma^{\top} + \cP_{\bX}^{\perp}\bW\bGamma^{\top} + \bo_{\PP}(1),
    \end{align*}
    where we replace the best linear projection $\bX\bD$ with its empirical counterpart $\cP_{\bX}\bZ$ in finite samples, which yield negligible terms that contribute to $\bo_{\PP}(1)$.
    It is worth noting that $\bX\bB^{\top}\cP_{\bGamma}^{\perp}$ and $\cP_{\bX}\bZ\bGamma^{\top}+\bW \bGamma^{\top}$ have orthogonal columns, while $\bX\bB^{\top}\cP_{\bGamma}^{\perp} + \cP_{\bX}\bZ\bGamma^{\top} $ and $\cP_{\bX}^{\perp}\bW\bGamma^{\top}$ have orthogonal rows.
    Our analysis will then take advantage of such two-way structural orthogonality to perform both estimation and inference for the parameters of interest, as detailed in the following sections.

\section{Estimation}\label{sec:estimation}

        From now on, we will use an asterisk on the upper subscript to indicate the population parameters and the true latent factors.
        Specifically, we denote the underlying parameter as
        \[\bTheta^* = \bX\bB^{*\top} + \bZ^*\bGamma^{*\top} =  \bX(\bB^*+\bGamma^*\bD^*)^{\top} + \bW^*\bGamma^{*\top}.\]
        Let $\cR\subseteq\RR$ be an open domain of $\theta$ such that $A(\theta)<\infty$ for all $\theta\in\cR$.
        For a given $C>0$, define $\cR_{C} = \cR\cap[-C,C]$ for Gaussian, Binomial and Poisson distributions and
        $\cR_{C} = \cR\cap[-C,-1/C]$ for Negative Binomial distributions.
        For our theoretical results, we assume the existence of constant $C>1$ such that the following common assumptions hold.
        
        \begin{assumption}[Model parameters]\label{asm:model}
            Assume that $\bTheta^*\in \cR_{C}^{n\times p}$ with probability $\iota_n$ for some deterministic sequence $\iota_n$ tending to one as $n$ tends to infinity.
            The primary coefficient satisfies that $\max_{1\leq \ell\leq d}\|\bB_{\ell}^*\|_0\leq s$ for some $1\leq s\leq p$ and $\max_{1\leq j\leq p}\|\bb_j^*\|_2\leq C$.
        \end{assumption}

        \begin{assumption}[Covariates]\label{asm:covariate}
            Assume that $\bx_{1},\ldots,\bx_{n}$ are i.i.d. $\nu$-sub-Gaussian random vectors with second moment $\bSigma_x:=\EE[\bx_1\bx_1^{\top}]$ such that $C^{-1}\leq \lambda_p(\bSigma_x)\leq \lambda_1(\bSigma_x) \leq C$.     
        \end{assumption}
        
        \begin{assumption}[Latent vectors]\label{asm:latent}
            Assume that the uncorrelated latent factors $\bw_{1}^*,\ldots,\bw_{n}^*$ are i.i.d. $\nu$-sub-Gaussian random vectors with zero means and covariance $\bSigma_w$, such that $C^{-1}\leq \lambda_r(\bSigma_w)\leq \lambda_1(\bSigma_w) \leq C$; and the factor loadings $\bGamma^*$ satisfy that $C^{-1} \leq \lambda_r(p^{-1}\bGamma^*\bGamma^{*\top}) \leq C$ and $\max_{1\leq j\leq p}\|\bgamma_j^*\|_2\leq C$.
        \end{assumption}        

        Like all nonlinear (nonconvex) analyses, the rows of $\bB^*$ and $\bGamma^*$ are assumed to be in a bounded set, as in \Cref{asm:model,asm:latent}.
        The boundedness of the natural parameter $\bTheta^*$ is required to control the tail probability of the response $y$ conditional on observed covariates and latent factors.
        In \Cref{asm:covariate}, the sub-Gaussian assumptions admit the particular case when $x_{i1}=1$ for all $1\leq i\leq n$ so that the intercept can be incorporated into our model.
        In \Cref{asm:latent}, the zero-mean condition on $\bW^*$ is to simplify the theoretical analysis, which can be guaranteed if we include the intercept and project $\bZ^*$ onto the linear span of the columns of $[\one_n,\bX]$.
        Finally, the sparsity and boundedness assumptions of $\bB^*$ in \Cref{asm:model}, and the bounded spectrum assumptions of $\bSigma_w$ and $\bGamma^*$ in \Cref{asm:latent} imply the conditions of \Cref{prop:iden-B} with $\tau_p=p$ therein.
        These assumptions are relatively lenient on the projection coefficient $\bD^*$, provided they ensure that $\bTheta^*$ remains within a bounded set with high probability.

        \begin{remark}[The number of latent factors]\label{rm:JIC}
            For our theoretical results, we assume the number of latent factors $r$ is known in advance.
            Note that the joint-likelihood-based information criterion (JIC) proposed by \citet{chen2022determining} can be utilized to select the number of latent factors.
            The JIC value is the sum of deviance and a penalty on model complexity:
            \begin{align}
                \JIC(\hat{\bTheta}^{(r)}) &= - 2 \sum_{i\in[n],j\in[p]}\log\ p(y_{ij}\mid \hat{\theta}_{ij}^{(r)}) + \nu(n,p,d+r), \label{eq:JIC}
            \end{align}
            where $\hat{\bTheta}^{(r)}$ is the joint maximum likelihood estimator of the natural parameter matrix that minimizes \eqref{eq:loglikelihood} with $r$ latent factors and $d$ observed covariates, and $\nu(n,p,r)=c_{\JIC}\cdot{r\log(n\wedge p)}{(n\wedge p)}^{-1}$ is the complexity measure with penalty level $c_{\JIC}>0$.            
            As shown by \citet{chen2022determining}, minimizing the empirical JIC yields a consistent estimate for the number of factors in generalized linear factor models with an intercept parameter.
            The utility of this metric in our problem setting is also empirically examined for both the simulation in \Cref{sec:simulation} and the real data analysis in \Cref{sec:case-studeis}.
        \end{remark}

        As motivated in \Cref{subsec:glm}, we consider the following optimization problem:
        \begin{align}
            \begin{split}
                \hat{\bB},\hat{\bZ},\hat{\bGamma} &= \argmin_{\bB\in\RR^{p\times d},\bZ\in\RR^{n\times r}, \bGamma\in\RR^{p\times r}}\cL(\bB,\bZ,\bGamma) + \lambda\|\bB\|_{1,1}\\
                \text{s.t.}&\qquad \bX\bB^{\top}+\bZ\bGamma^{\top}\in\cR_C^{n\times p},\qquad\cP_{\bGamma}\bB=\zero.
            \end{split}            \label{opt:3}
        \end{align}
        where the unregularized loss function $\cL(\bB,\bZ,\bGamma)$ is defined in \eqref{eq:loglikelihood} and $\|\cdot\|_{1,1}$ denotes the element-wise $\ell_1$-norm.
        It is worth noting that for any feasible $\hat{\bGamma}$ fixed, \eqref{opt:3} reduces to a convex optimization problem in variables $\bB$ and $\bZ$:
        \begin{align}
            \begin{split}
                \hat{\bB},\hat{\bZ} &= \argmin_{\bB\in\RR^{p\times d},\bZ\in\RR^{n\times r}}  \cL(\bB,\bZ,\hat{\bGamma}) + \lambda\|\bB\|_{1,1}\\
                \text{s.t.}&\qquad \bX\bB^{\top}+\bZ\hat{\bGamma}^{\top}\in\cR_C^{n\times p},\qquad \cP_{\hat{\bGamma}}\bB=\zero.
            \end{split}\label{opt:4}
        \end{align}
        This motivates us to solve optimization problem \eqref{opt:3} in two steps: (1) firstly obtaining a good estimate of $\hat{\bGamma}$ and (2) then based on $\hat{\bGamma}$, obtaining good estimates for $\bB^*$ and $\bZ^*$.
        In \Cref{alg:deconfounder}, we incorporate the two-step procedure by solving two sub-problems \eqref{opt:1} and \eqref{opt:2} consecutively.
        We next analyze the statistical properties of estimators in each step of \Cref{alg:deconfounder}.

    \begin{algorithm}[!t]
    \caption{\textsc{gcate} (generalized confounder adjustment for testing and estimation)}
        \label{alg:deconfounder}
        \begin{algorithmic}[1]
        \REQUIRE
        A data matrix $\bY\in\RR^{n\times p}$,
        a design matrix $\bX\in\RR^{n\times d}$,
        a natural number $r\geq 1$ (the number of latent factors)
        \smallskip
        
        \STATE \textbf{Estimation of uncorrelated latent components $\bW\bGamma^{\top}$:} Solve optimization problem \eqref{opt:1} to obtain $\hat{\bW}_0\hat{\bGamma}_0^{\top}$ and the initial estimate $\hat{\bTheta}_0=\bX\hat{\bF}^{\top}+\hat{\bW}_0\hat{\bGamma}_0^{\top}$ by alternative maximization (\Cref{alg:jmle}) with initialization given in \Cref{app:subsec:opt}:
        \begin{align}
            \begin{split}
                \hat{\bF},\hat{\bW}_0,\hat{\bGamma}_0 &\in \argmin_{\bF\in\RR^{p\times d},\bW\in\RR^{n\times r},\bGamma\in\RR^{p\times r}} \cL(\bX\bF^{\top} + \bW\bGamma^{\top})   \\
            \text{subject to}&\quad \bX\bF^{\top} + \bW\bGamma^{\top}\in\cR_C^{n\times p},\qquad  \cP_{\bX}\bW =\zero.
            \end{split}\label{opt:1} 
        \end{align}

        \STATE \textbf{Estimation of latent coefficients $\bGamma$:} Set $\hat{\bW}:=\sqrt{n}\bU\bSigma^{1/2}$ and $\hat{\bGamma}:=\sqrt{p}\bV\bSigma^{1/2}$, where $\hat{\bW}_0\hat{\bGamma}_0^{\top} =\sqrt{np} \bU\bSigma\bV^{\top}$ is the condensed SVD with $\bU\in\RR^{n\times r}$, $\bSigma\in\RR^{r\times r}$, $\bV\in\RR^{p\times r}$.

        \STATE \textbf{Estimation of direct effects $\bB$ and latent factors $\bZ$:}
        Solve optimization problem \eqref{opt:2} to obtain $(\hat{\bB},\hat{\bZ})$ by \Cref{alg:jmle} with initialization given in \Cref{app:subsec:opt}:
        \begin{align}
            \begin{split}
                \hat{\bB},\hat{\bZ} &= \argmin_{\bB\in\RR^{p\times d},\bZ\in\RR^{p\times r}}\cL(\bX\bB^{\top}+\bZ\hat{\bGamma}^{\top}) + \lambda\|\bB\|_{1,1}\\
            \text{subject to}&\qquad \bX\bB^{\top}+\bZ\hat{\bGamma}^{\top}\in\cR_C^{n\times p},\qquad \cP_{\hat{\bGamma}}\bB=\zero.
            \end{split}            \label{opt:2}
        \end{align}

        \STATE \textbf{Debiasing:} Construct the debiased estimate \eqref{eq:debias} and its estimated variance \eqref{eq:est-var}, based on $(\hat{\bB},\hat{\bZ},\hat{\bGamma})$.
        Compute $p$-values according to the asymptotic distribution \eqref{eq:normality}.
        
        \ENSURE Return the $p$-values.
        \end{algorithmic}
    \end{algorithm}

    \subsection{Estimation of uncorrelated latent components}
    To estimate the marginal effects $\bF^*= \bB^* + \bGamma^*\bD^*$ and the uncorrelated latent components $\bW^*\bGamma^{*\top}$, we first solve optimization problem \eqref{opt:1}.
    This is also known as the joint maximum likelihood estimation 
    \citep{chen2019joint,chen2020structured}, which is statistically optimal in the minimax sense when both the sample size $n$ and the response dimension $p$ grow to infinity.
    From optimization problem \eqref{opt:1}, we obtain the initial estimates of the natural parameter matrix $\hat{\bTheta}_0 = \bX\hat{\bF}^{\top} + \hat{\bW}_0 \hat{\bGamma}_0^{\top}$.
    The following theorem characterizes the estimation error of the initial maximum likelihood estimate $\hat{\bTheta}_0$.

    \begin{theorem}[Estimation error of $\hat{\bTheta}_0$]\label{thm:est-error-bTheta}
        Under \Crefrange{asm:model}{asm:latent}, let $\hat{\bTheta}_0$ be any estimator such that $\cL(\hat{\bTheta}_0)\leq \cL(\bTheta^*)$. For any constant $\delta>1$, when $np\geq 3$, it holds with probability at least $1-(n+p)^{-\delta}-(np)^{-\delta} - \iota_n$ that
        \begin{align*}
            \|\hat{\bTheta}_0 - \bTheta^*\|_{\fro} = \cO\left( \sqrt{(d+r)((n\vee p)\vee \delta^3)} \right),\  \max_{1\leq j\leq p} \|(\hat{\bTheta}_0)_{j} - \bTheta^*_{j}\|_2 = \cO\left( \sqrt{(d+r)(n \vee \delta^3)}\right).
        \end{align*}
    \end{theorem}

    In our specific setting, where the dimensions represented by $d$ and $r$ are orders of magnitude smaller compared to $n$ and $p$, the estimation error is primarily dominated by the scale of the larger dimensions $n$ and $p$. 
    Because the dimensions of natural parameters expand with both $n$ and $p$, the bound implies that the error associated with each entry is approximately on the order of $\sqrt{(n\vee p) / np}$. As demonstrated in the upcoming subsection, these results empower us to attain robust estimates of the confounding effects.

    \subsection{Estimation of latent coefficients}\label{subsec:est-confounding}
        
        From optimization problem \eqref{opt:1}, we obtain the initial estimates $\hat{\bE}=\hat{\bW}_0\hat{\bGamma}_0^{\top}$ of the latent components that are uncorrelated to the observed covariates.
        Though $\bW^*$ and $\bGamma^*$ are only identified up to rotations, we can use the condensed singular value decomposition of the normalized latent components $\hat{\bE}/\sqrt{np} = \bU\bSigma\bV^{\top}$ to obtain the final estimates through the following optimization problem:
        \begin{align}
            \begin{split}
                \hat{\bW},\hat{\bGamma}\in&\argmin_{\bW\in\RR^{n\times r},\bGamma\in\RR^{p\times r}}\frac{1}{np}\|\hat{\bE} - \bW\bGamma^{\top}\|_{\fro}^2\\
            \text{subject to}&\quad \frac{1}{n}\bW^{\top}\bW=\frac{1}{p}\bGamma^{\top}\bGamma\text{ is diagonal.}
            \end{split}\label{eq:W-Gamma}
        \end{align}
        A simple derivation yields the solution $\hat{\bW}=\sqrt{n}\bU\bSigma^{1/2}$ and $\hat{\bGamma}=\sqrt{p}\bV\bSigma^{1/2}$ to the above problem.
        The above procedure is also commonly used in the factor analysis literature for estimating factor loadings from regression residuals \citep{bai2003inferential,bing2022inference}.
        The following theorem guarantees that the above estimate of the latent coefficients is provably accurate in recovering the column space of the true coefficients.
    
    \begin{theorem}[Estimation error of $\cP_{\hat{\bGamma}}$]\label{thm:est-error-proj-bGamma}
        Under \Crefrange{asm:model}{asm:latent}, as $n,p\rightarrow\infty$, it holds that
        \begin{align*}
            \|\cP_{\hat{\bGamma}}-\cP_{\bGamma^*}\|_{\oper}= \Op\left(\sqrt{\frac{1}{n\wedge p}}\right),\qquad \max_{1\leq i,j\leq p}|(\cP_{\hat{\bGamma}}-\cP_{\bGamma^*})_{ij}|=\Op\left(\sqrt{\frac{1}{p^2(n\wedge p)}}\right).
        \end{align*}        
    \end{theorem}

    \Cref{thm:est-error-proj-bGamma} implies that the image of $\bGamma^*$ can be estimated well by $\hat{\bGamma}$.
    Furthermore, the column-wise error decays at a fast rate.
    The precise error control of each column individually enables us to disentangle the intricate relationships within the confounder-adorned high-dimensional dataset.
    However, these do not directly extend to error control for the latent factors $\bZ^*$ or the latent coefficients $\bGamma^*$ themselves.

    Under multivariate linear models, \citet[Theorem 4]{bing2022inference} show the concentration of the latent coefficients $\max_{1\leq j\leq p}\|\bgamma_j^*-\hat{\bgamma}_j\|_2$.
    Their results rely on the special structure of the regression problem $\bY_j= \bX(\bb_j^*+\bD^{*\top}\bgamma_j^*)+\bepsilon_j$, where the regression coefficient can be decomposed into a sparse component $\bb_j^*$ and a dense component $\bD^{*\top}\bgamma_j^*$.
    By using the lava estimator \citep{chernozhukov2017lava}, they can derive the estimation error of the residual $\bepsilon_j$, from whose covariance structure, $\|\bgamma_j^*-\hat{\bgamma}_j\|_2$ can be further bounded.
    In the generalized linear model setting, we don't have the flexibility to utilize the additive noises' covariance structure to directly estimate $\bgamma_j$ well.
    Instead, we need to rely on joint maximum likelihood estimation to estimate them, as we illustrate below.

    \subsection{Estimation of latent factors and direct effects}\label{subsec:est-direct}
        Once the column space of confounding effects becomes distinguishable, the subsequent phase entails retrieving direct effects by mitigating the influence of confounding variables and solving optimization problem \eqref{opt:4}.
        Through this, we can simultaneously obtain the estimates of latent factors $\bZ^*$ and direct effects $\bB^*$.
        In the high-dimensional scenarios when $p$ can be larger than $n$, one natural approach to estimate the sparse coefficient $\bB^*$ is via $\ell_1$-regularization, as employed in the optimization problem \eqref{opt:2}, which aims to obtain a sparse and consistent estimator $\hat{\bB}$ by $\ell_1$-regularization while simultaneously removing the unmeasured effects.
        Next, we analyze the properties of the two estimators $\hat{\bZ}$ and $\hat{\bB}$ in turn.

        \noindent\textbf{Latent factors.}
        As previously alluded to, estimating latent factor $\bZ^*$ demands special consideration.
        To bypass the technical difficulty, we will use the estimation errors of $\hat{\bTheta}_0$ and $\cP_{\hat{\bGamma}}$ provided by \Cref{thm:est-error-bTheta} and \Cref{thm:est-error-proj-bGamma}, respectively, coupled with one extra identifiability condition for the latent factors $\bZ^*$ and their coefficients $\bGamma^*$.
        From \citet[Proposition 5.1]{lin2021exponential}, there exists an invertible matrix $\bR\in\RR^{r\times r}$ with bounded operator norm, such that $\Var(\bR\bz_1^*)$ and $\bR^{-\top}\bgamma_1^*\bgamma_1^{*\top}\bR^{-1}$ are the same diagonal matrix.
        The following assumption from \citet{lin2021exponential} restricts the spacing of $\Var(\bR\bz_1^*)$'s eigenvalues.
        
        \begin{assumption}[Identifiability of latent factors]\label{asm:iden-latent-eigs}
            Assume there exists positive numbers $c_1\leq c_2$ and $1 < k_1\leq k_2$ such that for all $\ell \in\{ 1,\ldots , r \}$, the eigenvalues of $\Var(\bR\bz_1^*)$ satisfy $c_1 \ell^{-k_1} \leq \lambda_{\ell}\leq c_2 \ell^{-k_1}$, and $\lambda_{\ell}-\lambda_{\ell+1} \geq c_1 \ell^{-k_2} $, with the convention that $\lambda_{r + 1}= 0$.
        \end{assumption}
    
        Intuitively, \Cref{asm:iden-latent-eigs} guarantees that $\bZ^*\bR^{\top}$ can be recovered up to sign from the matrix product $\bZ^*\bGamma^{*\top}$.
        This implies that, if one can consistently estimate $\bZ^*\bGamma^{*\top}$ with $\hat{\bZ}\hat{\bGamma}^{\top}$, then $\bZ^*$ and $\bGamma^{*}$ can also be consistently estimated by appropriate transformations of $\hat{\bZ}$ and $\hat{\bGamma}$, respectively.
        A simple consequence from \citet[Proposition 5.2]{lin2021exponential}, coupled with \Cref{thm:est-error-bTheta,thm:est-error-proj-bGamma}, is the following error bound on the columns of latent components.
        
        \begin{corollary}[Estimation of latent components]\label{cor:est-confound-col}
            Let $\hat{\bGamma}$ and $\hat{\bZ}$ be solutions to optimization problems \eqref{eq:W-Gamma} and \eqref{opt:2}, respectively.
            Under \Crefrange{asm:model}{asm:iden-latent-eigs}, suppose $\min_{\{\bB\in\RR^{p\times d}\mid \cP_{\hat{\bGamma}}\bB=\zero\}}$ $\cL(\bX\bB^{\top} + \hat{\bZ}\hat{\bGamma}^{\top}) \leq \cL(\bTheta^*)$ with probability tending to one.
            Then, as $n,p\rightarrow\infty$, it holds that,
            \[\max_{1\leq j\leq p} \frac{1}{\sqrt{n}}\|\hat{\bZ}\hat{\bgamma}_j - \bZ^*\bgamma_j^*\|_2 =  \Op\left(\sqrt{\frac{\log n}{n}} \vee \sqrt{\frac{r^{4k_2-k_1+4} \log n}{n\wedge p }}\right). \]
        \end{corollary}

        It's important to note that, unlike the analysis for linear models in prior work \citep{bing2022adaptive,bing2022inference} that projects the responses onto the orthogonal column space of $\hat{\bGamma}$ and removes the effects of latent factors $\bZ$, estimating $\bZ$ is unavoidable under generalized linear models.
        In \Cref{cor:est-confound-col}, we require the joint maximum likelihood based on the estimated latent components to be higher than the likelihood evaluated at the truth.
        This requires the estimated latent components derived from \eqref{opt:2} to exhibit stability and ensures that the maximum likelihood with the estimated latent factors remains close to the joint maximum likelihood from \eqref{opt:1}.     
        In the presence of nuisance parameters $\bZ^*$ and $\bGamma^*$, the sharp control on estimation error of the column $\bZ^*\bgamma_j^*$ provided by \Cref{cor:est-confound-col} helps control the estimation error of $\hat{\bB}$.

        \noindent\textbf{Direct effects.}
        In high-dimensional scenarios, controlling the estimation error of $\cP_{\bGamma^*}\bB^*$ requires the projection $\cP_{\bGamma^*}$ does not excessively densify the primary effects $\bB^*$.
        To this end, we require the ratio $\|\cP_{\bGamma^*}\bB^*\|_{1,1}/\|\cP_{\bGamma^*}\bB^*\|_{\fro}$ to be of smaller order than $\sqrt{p}$.
        Coupled with the previous assumptions and results, the estimation error of $\hat{\bB}$ returned by problem \eqref{opt:2} can be controlled.

    \begin{theorem}[Estimation error of $\hat{\bB}$]\label{thm:est-err-B}
        Suppose the assumptions in \Cref{cor:est-confound-col} hold and $\|\cP_{\bGamma^*}\bB^*\|_{1,1}= \cO(p^{k/2}\|\cP_{\bGamma^*}\bB^*\|_{\fro})$ for some constant $k\in[0,1)$.
        Then, as $n,p\rightarrow\infty$ such that $\sqrt{n}/\log(nd) = o(p)$ and $\log (p) = o(n)$, the estimate $\hat{\bB}$ of optimization problem \eqref{opt:2} with $\lambda \asymp 8\nu^2\log(nd)n^{-1/2}$ satisfies that
        \begin{align*}
            \|\hat{\bB}-\bB^*\|_{\fro}=\Op( r_{n,p}),\qquad \|\hat{\bB} - \bB^*\|_{1,1} = 
            \Op(r_{n,p}'),
        \end{align*}
        where the sequences $r_{n,p}$ and $r_{n,p}'$ are defined as:
        \begin{align*}
            r_{n,p}&:=
            \sqrt{\frac{(sd\log^2(nd))\vee\log(np)}{n}}  + \sqrt{\frac{n^{1/2}}{(n\wedge p)^{3/2}\log(nd)}} +\sqrt{\frac{sd}{n\wedge  p^{1-k}}}, \\
            r_{n,p}' &:=  \sqrt{sd}\ r_{n,p} + \frac{\sqrt{n}}{(n\wedge p)\log(nd)}.
        \end{align*}
        % \[r_{n,p}:=
        % \sqrt{sd\frac{(sd\log^2(nd))\vee\log(np)}{n}}  + \frac{\sqrt{n}}{(n\wedge p)\log(nd)} +\sqrt{\frac{sd}{n\wedge  p^{1-k}}}, 
        % \sqrt{\frac{(sd\log^2(nd))\vee\log(np)}{n}} + \frac{\sqrt{n}}{(n\wedge p)\log(nd)} + \sqrt{\frac{sd}{n\wedge  p^{1-k}}}.\]
    \end{theorem}
    In \Cref{thm:est-err-B}, the parameter $k$ captures the deviation of the projected $\ell_1$-norm $\|\cP_{\bGamma^*}^{\perp}\bB^*\|_{1,1}$ from $\|\bB^*\|_{1,1}$.
    The smaller $k$, the more information of $\bB^*$ is retained after projection, and the signal-noise-ratio is larger.
    In the high-dimensional scenarios when $n<p$, the $\ell_1$-norm and $\ell_2$-norm of the estimation error scale in $\Op(n^{-1/2}\wedge p^{{(k-1)}/{2}})$ when ignoring lower order factors.
    The appearance of the response dimension $p$ in the denominator reflects the blessing of dimensionality; namely, having more responses than the sample size is not detrimental to consistency, provided that $p$ does not grow exponentially larger than $n$, as we numerically demonstrate in \Cref{fig:est-err}.
    
    To prove \Cref{thm:est-err-B}, we establish the (approximate) optimal condition for ($\hat{\bB},\hat{\bZ},\hat{\bGamma})$ from the two-step procedure to the joint optimization problem \eqref{opt:3}, as shown in \Cref{lem:optimality}.
    This relies on the optimality condition of optimization problem \eqref{opt:2} and the convergence rate of $\cP_{\hat{\bGamma}}^{\perp}$ provided by \Cref{thm:est-error-proj-bGamma}.
    It then allows us to establish the cone condition, obtain the upper and lower bounds of the first-order approximation error of the loss function, and derive the error rate in \Cref{app:est-direct}.
    Compared to double machine learning in the presence of high-dimensional nuisance parameters \citep{chernozhukov2018plug,chernozhukov2018double}, our estimation procedure does not require sample splitting. 
    To establish consistency, we only need the convergence rate of $\max_{1\leq j\leq p}\|\bZ^*\bgamma^*_j-\hat{\bZ}\hat{\bgamma}_j\|_2/\sqrt{n}$ to be the parametric rate $(n\wedge p )^{-1/2}$, as shown in the proof of \Cref{thm:est-err-B}; see also \Cref{remark:neyman-ortho} for discussion on the connection to Neyman orthogonality.
    However, to derive the asymptotic distribution for inference, one may need more stringent conditions or sample splitting, as illustrated next.

\section{Inference}\label{sec:inference}

    \subsection{Projected and weighted bias correction}
    When evaluating uncertainty in high-dimensional inference, confidence intervals and statistical hypothesis tests are required.
    After obtaining the initial estimate $\hat{\bB}$, we need to remove the bias caused by $\ell_1$-regularization to have valid inferences on the estimated coefficients. 
    Without loss of generality, we focus on testing the coefficients of the first covariate $b_{j1}$ for $j=1,\ldots,p$.
    We consider the following debiased estimator for each of them:
    \begin{align}
        \hat{b}_{j1}^{\de} &= \hat{b}_{j1} + {\bu}^{\top}\frac{1}{n}\sum_{i=1}^n  \bx_{i} (\by_{i} - A'(\hat{\btheta}_{i}))^{\top} \bv_i,\label{eq:debias}
    \end{align}
    where $\hat{\bTheta} :=\bX\hat{\bB}^{\top}+\hat{\bZ}\hat{\bGamma}^{\top}$ is the estimated natural parameter matrix, and $\bu\in\RR^d$ and $\bv_i\in\RR^p$ are projection vectors to be specified later, such that the correction term $n^{-1} \bu^{\top}  \sum_{i=1}^n\bx_{i} (\by_{i} - A'(\hat{\btheta}_{i}))^{\top}\bv_{i}$ is a reasonable estimate of the bias $b_{j1}^*-\hat{b}_{j1}$.
    
    By Taylor expansion of $A'(\theta_{ij}^*)$ at $ \hat{\theta}_{ij} := \bx_i^{\top}\hat{\bb}_j + \hat{\bz}_i^{\top}\hat{\bgamma}_j$, we have
    \begin{align*}
        A'(\theta_{ij}^*) &= A'(\hat{\theta}_{ij}) + A''(\hat{\theta}_{ij}) ( \theta_{ij}^*-\hat{\theta}_{ij}) + \frac{1}{2}A'''(\psi_{ij}) (\theta_{ij}^* - \hat{\theta}_{ij})^2,
    \end{align*}
    for some $\psi_{ij}$ between $\hat{\theta}_{ij}$ and $\theta_{ij}^*$.
    Then, the residual of the $i$th sample can be decomposed into three sources of errors:
    \begin{align}
        \by_{i} - A'(\hat{\btheta}_{i})=&\underbrace{\bepsilon_{i}}_{\text{stochastic error}} + \underbrace{\bp_i}_{\text{remaining bias}}   + \underbrace{\bq_{i}}_{\text{ approximation error}} \label{eq:error-decomp}
    \end{align}
    where the three terms of errors are given by
    \begin{align*}
        \bepsilon_{i} &= \by_{i}- A'(\btheta_{i}^*)\\
        \bp_i &=A''(\hat{\btheta}_{i})\odot(\hat{\btheta}_{i}-\btheta_{i}^* )  \\
        \bq_i &= -\frac{1}{2}[A'''(\psi_{ij})(\theta_{ij}^* - \hat{\theta}_{ij})^2]_{1\leq j\leq p}.
    \end{align*}
    If we let $\bv_i={\omega}_{i}\diag(A''(\hat{\btheta}_i))^{-1}\cP_{\hat{\bGamma}}^{\perp}\be_j $ where $\omega_{i}$ is the sample-specific weight and $\be_j\in\RR^{p}$ is the unit vector with $j$th entry being one, then substituting \eqref{eq:error-decomp} into \eqref{eq:debias} yields that
    \begin{align}
        \hat{b}_{j1}^{\de} - b_{j1}^* &=    (\hat{b}_{j1} - b_{j1}^*) + \bu^{\top}\frac{1}{n}\sum_{i=1}^n \bx_{i}\bepsilon_{i}^{\top}\bv_i + \bu^{\top}\frac{1}{n}\sum_{i=1}^n \bx_{i}\bp_{i}^{\top}\bv_i  + \bu^{\top}\frac{1}{n}\sum_{i=1}^n \bx_{i}\bq_{i}^{\top}\bv_i \notag\\
        &= \bu^{\top}\frac{1}{n}\sum_{i=1}^n \bx_{i}\bepsilon_{i}^{\top}\bv_i+ \left(\bu^{\top}\frac{1}{n}\sum_{i=1}^n \omega_{i}\bx_{i} \bx_{i}^{\top} -\be_1^{\top}\right) (\bb_{j}^* - \hat{\bb}_{j}) + \mathrm{Rem}. \label{eq:diff-debias}
    \end{align}
    The estimation error rates provided by \Cref{thm:est-error-proj-bGamma,thm:est-err-B} guarantee that $\|\bb_j^*-\hat{\bb}_j\|_1=\Op(r_{n,p}')$ and the remaining term is $\mathrm{Rem}=\Op(\max_{1\leq i\leq n}|\bu^{\top}\bx_i|^3 r_{n,p}^{2})$ for $r_{n,p}$ and $r_{n,p}'$ defined in \Cref{thm:est-err-B}.
    Based on \eqref{eq:diff-debias}, the idea of debiasing is to choose $\bu$ and $\omega_{i}$'s such that the second term and the remaining term is of order $\op(n^{-1/2})$, while enabling the convergence of the average of primary stochastic errors to a normal distribution by central limit theorem.

    To facilitate our theoretical analysis, suppose we split the dataset into two parts $\cD_1=\{(\bx_{i},\by_{i}),1\leq i\leq n\}$ and $\cD_2=\{(\bx_{i},\by_{i}),n+1\leq i\leq 2n\}$, where $\cD_2$ is used to obtain the estimates $\hat{\bB}$ and $\hat{\bGamma}$, and $\cD_1$ is used to remove the bias for $\hat{\bB}$ induced by $\ell_1$-regularization.
    There are also latent factors $\{\bz_i^*\}_{i=1}^n$ and $\{\bz_i^*\}_{i=n+1}^{2n}$ associated with $\cD_1$ and $\cD_2$, respectively.
    Further, if $\bepsilon_i$'s are independent of the projection vectors $\bu$ and $\bv_i$ (or equivalently ${\omega}_{i}$) conditional on $(\bx_i,\bz_i^*)$'s and $\cD_2$, then we can approximate the scaled conditional variance of the stochastic errors as:
    \begin{align*}
        \sigma_j^2&=\Var\left(\bu^{\top}\frac{1}{\sqrt{n}}\sum_{i=1}^n \bx_{i} \bepsilon_{i}^{\top}\bv_i \,\middle|\, \{(\bx_{i},\bz_i^*)\}_{i=1}^n, \cD_2\right)\\
        &= \bu^{\top}\frac{1}{n}\sum_{i=1}^n \omega_i^2\bx_{i}\be_j^{\top}\cP_{\hat{\bGamma}}^{\perp}\diag(A''(\hat{\btheta}_i))^{-1}\Cov(\bepsilon_i\mid \btheta_{i}^*)\diag(A''(\hat{\btheta}_i))^{-1}\cP_{\hat{\bGamma}}^{\perp}\be_j\bx_{i}\bu\\
        &= \bu^{\top}\frac{1}{n}\sum_{i=1}^n \omega_{i}^2
        (\be_j^{\top}\cP_{\hat{\bGamma}}^{\perp}\diag(A''(\hat{\btheta}_i))^{-1}\diag(A''(\btheta_{i}^*))\diag(A''(\hat{\btheta}_i))^{-1}\cP_{\hat{\bGamma}}^{\perp}\be_j)\bx_{i}\bx_{i}^{\top}\bu\\
        &\approx \bu^{\top}\frac{1}{n}\sum_{i=1}^n \omega_{i}\bx_{i}\bx_{i}^{\top}\bu=:\hat{\sigma}_j^2,
    \end{align*}
    by using a proper data-dependent weight ${\omega}_{i}=\hat{\omega}_{i}$.
    Then, the projection vector $\hat{\bu}$ is constructed by minimizing the variance proxy while controlling the bias and remaining terms in \eqref{eq:diff-debias}:
    \begin{align}
        \begin{split}        
        \hat{\bu}\quad \in\quad &\argmin_{\bu\in\RR^d} \bu^{\top}\frac{1}{n}\sum_{i=1}^n \hat{\omega}_i \bx_{i}\bx_{i}^{\top}\bu\\
        \text{s.t.}\qquad& \left\|\frac{1}{n}\sum_{i=1}^n \hat{\omega}_i\bx_{i} \bx_{i}^{\top}\bu -\be_1\right\|_{\infty} \leq \lambda_n \\
        &\max_{1\leq i\leq n}|\bx_{i}^{\top}\bu|\leq \tau_n,
        \end{split} \label{opt:min-var}
    \end{align}
    where $\lambda_n\asymp \sqrt{\log(nd)/n}$ and $\tau_n\asymp\sqrt{\log n}$.
    Based on $\hat{\omega}_{i}$ and $\hat{\bu}$, the resulting bias-corrected estimator \eqref{eq:debias} is similar to those used by \citet{van2014asymptotically,javanmard2014confidence,cai2021statistical}; however, we need to incorporate information from multiple responses with projection operator $\cP_{\hat{\bGamma}}^{\perp}$ to de-confound, in the spirit of proximal gradient descent.
    Under mild regularity conditions, the following theorem shows that the debiased estimator $\hat{b}_{j1}^{\de}$ is asymptotically normal.

    \begin{theorem}[Asymptotical normality of $\hat{\bB}^{\de}$]\label{thm:normality}
        Under the same conditions in \Cref{thm:est-err-B}, for $j=1,\ldots,p$, additionally assume the following conditions hold: (i) $n/\log(nd)=o(p^{3/2})$ and $n=o(p^{2(1-k)})$; and (ii) $\hat{\omega}_i=\omega(\bx_i,\bz_i^*,\cD_2)$ for some real-valued function $\omega$ that is uniformly bounded away from 0 and $\infty$.
        Then as $n,p\rightarrow\infty$, it holds that
        \begin{align}
            \sqrt{n}\frac{ \hat{b}_{j1}^{\de} - b_{j1}^* }{\sigma_j} \dto \cN(0, 1). \label{eq:normality}
        \end{align}        
    \end{theorem}

    With fewer assumptions on the correlation between the covariate $\bX$ and the confounder $\bZ^*$, removing unmeasured confounders is only possible by utilizing multiple outcomes to disentangle the primary effect $\bB^*$ and the latent coefficient $\bGamma^*$.
    In particular, because the estimation error rates of $\bB^*$ and $\bGamma^*$ are related to $(n\wedge p)^{-1}$, the number of outcomes $p$ is expected to be larger than $n$, so that these errors are primarily affected by the sample size $n$.

    In \Cref{thm:normality}, condition (i) requires that the response dimension $p$ grows faster than $n^{2/3} \vee n^{1/(2(1-k))}$, which ensures the remainder term $\mathrm{Rem}$ in \eqref{eq:diff-debias} vanishes in the limit.
    Specifically, $\mathrm{Rem}$ has a magnitude associated with the convergence rate of $\|\hat{\bB}-\bB^*\|_{\fro}^2$, as provided by \Cref{thm:est-err-B}.
    To derive the asymptotic normality, $\mathrm{Rem}=\op(n^{-1/2})$ is required; however, if $n$ is too large compared to $p$, the convergence rate of $\hat{\bB}$ from the first two steps of the proposed procedure is insufficient to establish the desired asymptotic normality.
    In this case, having a much larger sample size does not help.
    When $k\leq 1/2$, condition (i) is satisfied with $n=o(p)$, which is reasonable in most scientific scenarios of cohort-level differential expression analysis, as we shall see later from the real data example in \Cref{sec:case-studeis}.

    In terms of condition (ii), a proper sample-specific and link-specific weight function is required.
    One can construct such weights $\hat{\omega}_i$ by sample splitting to fulfill this condition.
    For instance, using sample splitting procedure in \Cref{alg:split}, one valid choice is $\hat{\omega}_{i}=A''(\hat{\theta}_{ij})$.
    In \Cref{lem:hsigma}, we show that such a choice of $\hat{\omega}_i$ satisfies the condition (ii) in \Cref{thm:normality} with probability tending to one and the resulting variance estimator
    \begin{align}
        \hat{\sigma}^2_j &=\hat{\bu}^{\top}\frac{1}{n}\sum_{i=1}^n \hat{\omega}_{i}\bx_{i}\bx_{i}^{\top}\hat{\bu}, \label{eq:est-var}
    \end{align}
    is also consistent with $\sigma_j^2$.
    Hence, \Cref{thm:normality} implies that $t_j=\sqrt{n}(\hat{b}_{j1}^{\de} - b_{j1}^* )/\hsigma_j \dto \cN(0, 1).$ We reject the null hypothesis $H_{0j}:b_{j1}^*=0$ at level-$\alpha$ if $|t_j|>z_{\alpha/2}:= \Phi^{-1}(1-\alpha/2)$, where $\Phi$ is the cumulative distribution function of standard normal.    
    Numerically, we show that the efficiency loss of sample splitting is negligible in \Cref{app:subsec:simu-sample-split}; and the proposed method performs well without sample splitting and is statistically more efficient than the alternative methods in \Cref{sec:simulation}.

    \begin{remark}[Inference without unmeasured confounders]
        In the special case when there are no unmeasured confounders, the matrix $\cP_{\hat{\bGamma}}^{\perp}$ reduces to the identity matrix.
        Also, the projection vector $\hat{\bu}$ is the $j$th column of $(\bX^{\top}\diag(A''(\bX\hat{\bB}_1))\bX)^{-1}$, and $\hsigma_j$ is the asymptotic variance of $\hat{b}_{j1}$ under well-specified generalized linear models.
        In this case, \eqref{eq:debias} is simply a one-step adjustment based on the score function.
        For Bernoulli distributed binary outcomes without unmeasured confounders, the above choice of the weight function $\omega(\theta) = A''(\theta)$ for optimization problem \eqref{opt:min-var} coincide with $f'(\theta)^2/(f(\theta)(1-f(\theta)))$, the one used in \citet{cai2021statistical} with $f=A'$ being the link function.
        For Gaussian outcomes when $A'$ is the identity link with a choice of weight $\hat{\omega}_i\equiv 1$, the procedure above reduces to the debias method by \citet{javanmard2014confidence}.

        When there are unmeasured confounders, the main difficulty lies in taking account of the rates of convergence of $\|\hat{\bB} - \bB^*\|_{1}$, $\|(\hat{\bGamma} - \bGamma^*)\be_j\|_{2}$, and $\|{\bB^*}^{\top}\cP_{\bGamma^*}\be_j\|_2$ such that the remainder term in \eqref{eq:diff-debias} is $\op(n^{-1/2})$, which is essentially the idea of the proof for \Cref{thm:normality}, as we have alluded to after \eqref{eq:diff-debias}.
    \end{remark}

    \begin{remark}[Incorporate information from latent factors]
        In \eqref{eq:debias} and \eqref{opt:min-var}, we only use the covariate $\bX$ to adjust for the estimation bias.
        However, including the estimated latent factors $\bZ$ to construct a projection vector $\bu$ of dimension $d+r$ is also feasible.
        The validity of this extension is also guaranteed by the sample splitting procedure in \Cref{alg:split}.
    \end{remark}

    \begin{remark}[Estimation and inference with non-canonical links]
        Through \Cref{sec:estimation,sec:inference}, we discuss the methodology to conduct inference on confounded generalized linear models (GLM) with canonical link functions, as outlined in \Cref{tab:exp-family}.
        However, in practical scenarios, non-canonical link functions may also be employed.
        For instance, the log link function is commonly used with Negative Binomial GLMs. Fortunately, our method extends its applicability to GLMs with non-canonical link functions, as exemplified in the case of the Negative Binomial GLMs in \Cref{app:subsec-nb-log-link}. Establishing theoretical guarantees for these scenarios may follow a similar framework with suitable assumptions to address the non-convexity of the objective functions, as elaborated in \Cref{app:subsec-nb-log-link}.
    \end{remark}

    \subsection{Simultaneous inference}
        
    For $j=1,\ldots,p$, the asymptotic normality provided in \Cref{thm:normality} provides Type-I error controls for individual hypothesis tests $H_{0j}:b_{j1}^*=0$.
    The following proposition shows that we can also control the overall Type-I error and family-wise error rate (FWER) using the statistic $t_j=\sqrt{n}(\hat{b}_{j1}^{\de} - b_{j1}^* )/\hsigma_j$. 

    \begin{proposition}[Simultaneous inference]\label{prop:simul-inference}
        Let $\cN_p=\{j\mid b_{j1}^*=0,j=1,\ldots,p\}$ be the true null hypotheses.
        Under the assumptions of \Cref{thm:normality}, as $n,p,|\cN_p|\rightarrow\infty$, it holds that
        \begin{align*}
            \frac{1}{|\cN_p|}\sum_{j\in\cN_p} \ind\{|t_j|>z_{\frac{\alpha}{2}}\} \pto \alpha,\qquad \text{and}\qquad\limsup\  \PP\left(\sum_{j\in\cN_p} \ind\{|t_j|>z_{\frac{\alpha}{2p}}\} \geq 1\right) \leq \alpha.
        \end{align*}
    \end{proposition}

    When $p$ is large, controlling for the false discovery rate (FDR) is more desirable when performing simultaneous testing.
    In that regard, \citet[Section 2.3]{cai2021statistical} provides insights on FDR controls using different techniques.
    From simulations in \Cref{sec:simulation}, we also show that FDR is usually well controlled by the Benjamini–Hochberg procedure empirically.

\section{Numerical experiments}\label{sec:simulation}
    DE and related tests are frequently performed in two distinct settings in the genomic field.
    One relies on counts of gene expression to contrast the expression of each gene in case versus control observations.
    Typically, observations are either samples from RNA sequencing (RNA-seq) \citep{love2014moderated} or pseudo-bulk cells obtained from single-cell sequencing by aggregating the expressions of single cells in the same homogeneous groups \citep{Squair:2021}.
    Another setting is single-cell RNA-sequence (scRNA-seq) CRISPR screening \citep{Dixit2016,barry2023exponential}, where the fundamental task is to test for association between a designed genetic perturbation and gene expression \citep{Dixit2016}.
    In both settings, the measured gene expression is often assumed to approximately follow a Poisson or Negative Binomial (NB) distribution \citep{sarkar2021separating}.
    However, in the former, the mean expression per sample is much larger due to molecular design, and the distribution is often approximated by a normal distribution with an appropriate transformation.
    In the latter case, the observational unit is a single cell. 
    Hence, the mean of the gene expression is near zero, and the data is not well approximated with a normal distribution.

    Before we turn to the simulation details, we present a simulated bulk-cell dataset and a simulated single-cell dataset corresponding to the above two distinct scenarios, respectively (\Cref{fig:data_dist}).
    The Poisson distribution can often model the former scenario, while the NB distribution is a better option for the latter because the counts are sparser and typically exhibit strong overdispersion (\Cref{fig:data_dist}(a)-(b)).
    Furthermore, for single-cell data, the lower-expressed genes are typically more dispersed, and this feature is captured in our simulated data set (\Cref{fig:data_dist}(c)). 
    In practice, both Poisson and NB models are available for analysis of either type of experiment; however, to simplify exposition, we use a Poisson distribution for bulk samples in \Cref{subsec:simu-poisson} and a NB distribution for single-cell samples in \Cref{subsec:simu-splatter}.
    In the subsequent experiments, we adhere to the protocol described in \Cref{app:subsec:hyperparam} for selecting both the hyperparameters and the number of factors pertinent to the proposed methods.

    \begin{figure}[!t]
        \centering
        \includegraphics[width=0.8\textwidth]{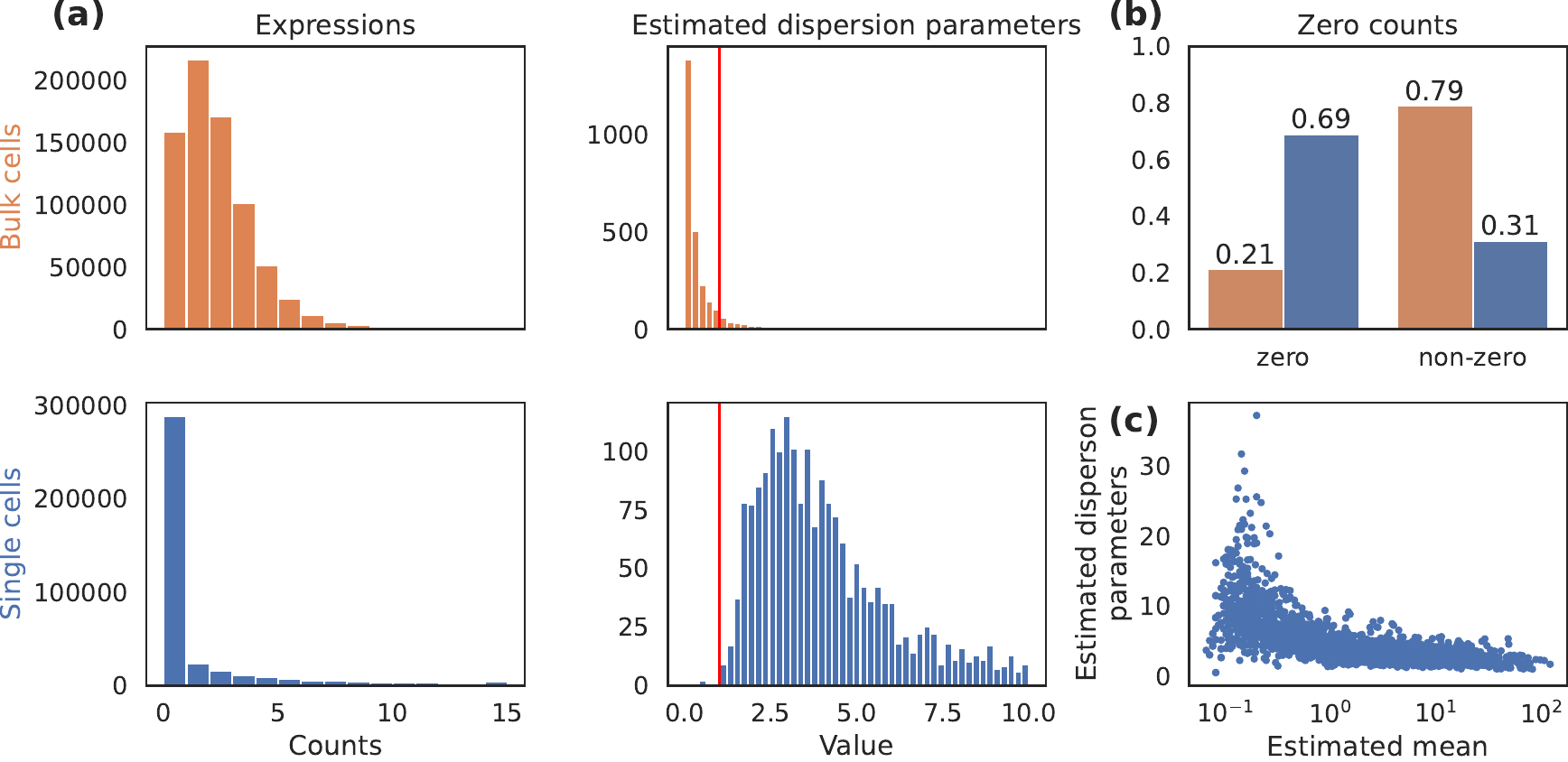}
        \caption{Overview of the simulated data.
        \textbf{(a)} The first and second rows show the summary of one simulated dataset for bulk cells (Poisson) in \Cref{subsec:simu-poisson} and single cells (Negative Binomial) by Splatter in \Cref{subsec:simu-splatter}, respectively.
        The first column shows the overall distribution of the generated counts; the second column shows the estimated dispersion parameters by methods of moments using the mean estimates from GLM with Poisson likelihood.
        \textbf{(b)} The proportions of zero and non-zero counts in the two datasets, colored in orange and blue, respectively.
        \textbf{(c)} The estimated dispersion parameter versus the estimated mean for the simulated single-cell dataset.}
        \label{fig:data_dist}
    \end{figure}

    \subsection{Well-specified simulated datasets}\label{subsec:simu-poisson}

    We simulate expression data $\bY$ that consists of $n\in\{100,250\}$ cells and $p=3,000$ genes based on the Poisson likelihood with natural parameter $\bTheta$.
    More specifically, we generate the covariate $x_1$ to be a centered binary variable, i.e., $(x_1+1)/2\sim \text{Bernoulli}(0.5)$.
    We also include an intercept $x_2=1$, so that the covariate vector $\bx=[x_1,x_2]^{\top}$ has dimension $d=2$.
    To allow for the most general confounding scenarios without assuming causal relationships as in \Cref{fig:causal-diagram}, we directly generate the latent factor matrix using $\bZ = \bX\bD^{\top} + \bW \in\RR^{n\times r}$ with the number of latent factors being $r\in\{2,10\}$.
    Here, to generate $\bD$ and $\bW$, we first sample their entries independently from $\cN(0,1)$ and further modify the singular values to be $s_1,\ldots,s_r$ where $s_k=a\cdot(2 - (k-1)/(r-1))$, with $a=n^{-3/2}$ for $\bD$ and $a=(n/2)^{1/2}$ for $\bW$.
    For the latent loading matrix $\bGamma$, we follow \citet{wang2017confounder} to take $\bGamma=\tilde{\bGamma}\bLambda$ where $\tilde{\bGamma}$ is a $p\times r$ orthogonal matrix sampled uniformly from the set of all $p\times r$ orthogonal matrix and $\bLambda=(p/2)^{1/2}\diag(\lambda_1,\ldots,\lambda_r)$ where $\lambda_k=2-(k-1)/(r-1)$.
    The primary effect of $x_1$ on gene $j$ is sampled from $(b_{j1}+0.2)/0.4\sim \text{Bernoulli}(0.5)$ with probability 0.05 and set to be zero with probability 0.95.
    The coefficient for the intercept is set to be $b_{j2}=0.5$.

    \begin{figure}[!t]
        \centering
        \includegraphics[width=\textwidth]{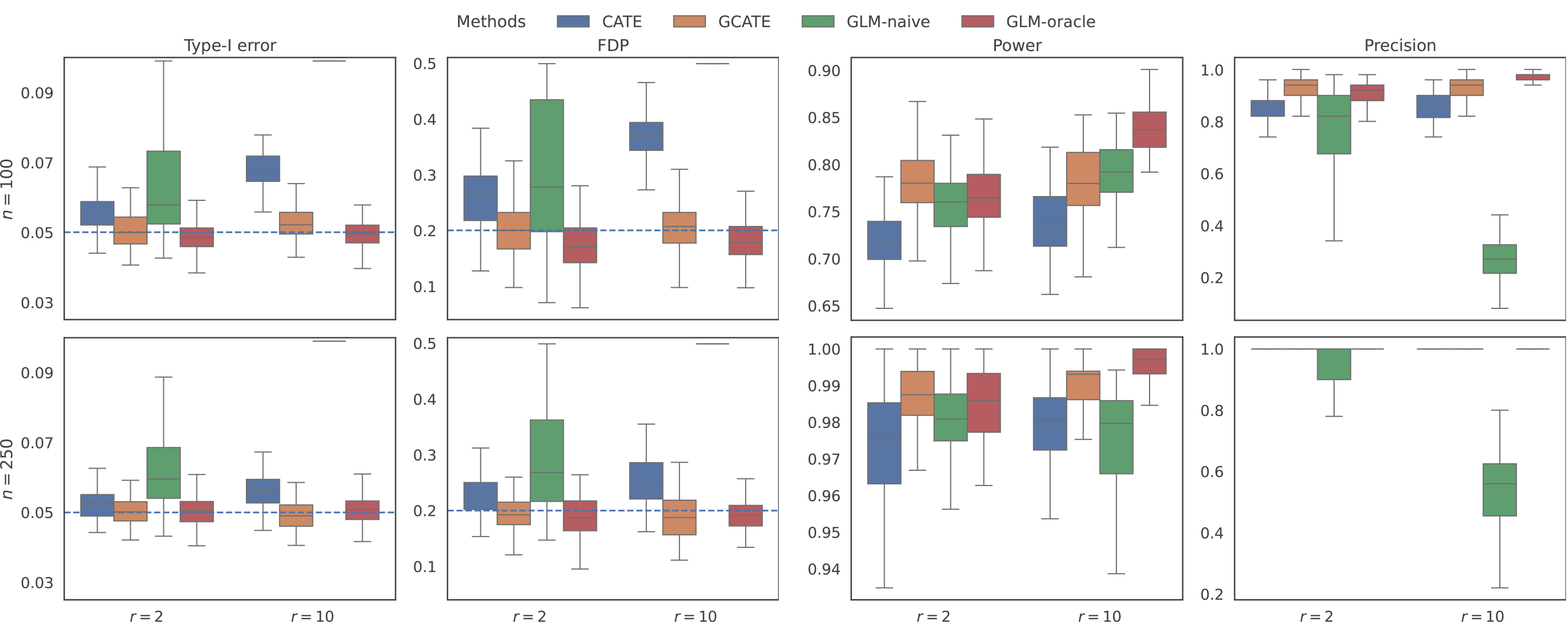}
        \caption{The Type-I errors, false discovery proportions (FDPs), powers, and precision of different methods on the simulated datasets over 100 runs, with varying numbers of samples $n\in\{100,250\}$ and numbers of latent factors $r\in\{2,10\}$.
        For \textsc{glm}, the maximum values of Type-I errors and FDPs are clipped at 0.1 and 0.5, respectively.
        The blue dashed lines indicate the desired cutoffs.}
        \label{fig:ex2-err}
    \end{figure}

     \begin{figure}
        \centering
        \includegraphics[width=0.7\textwidth]{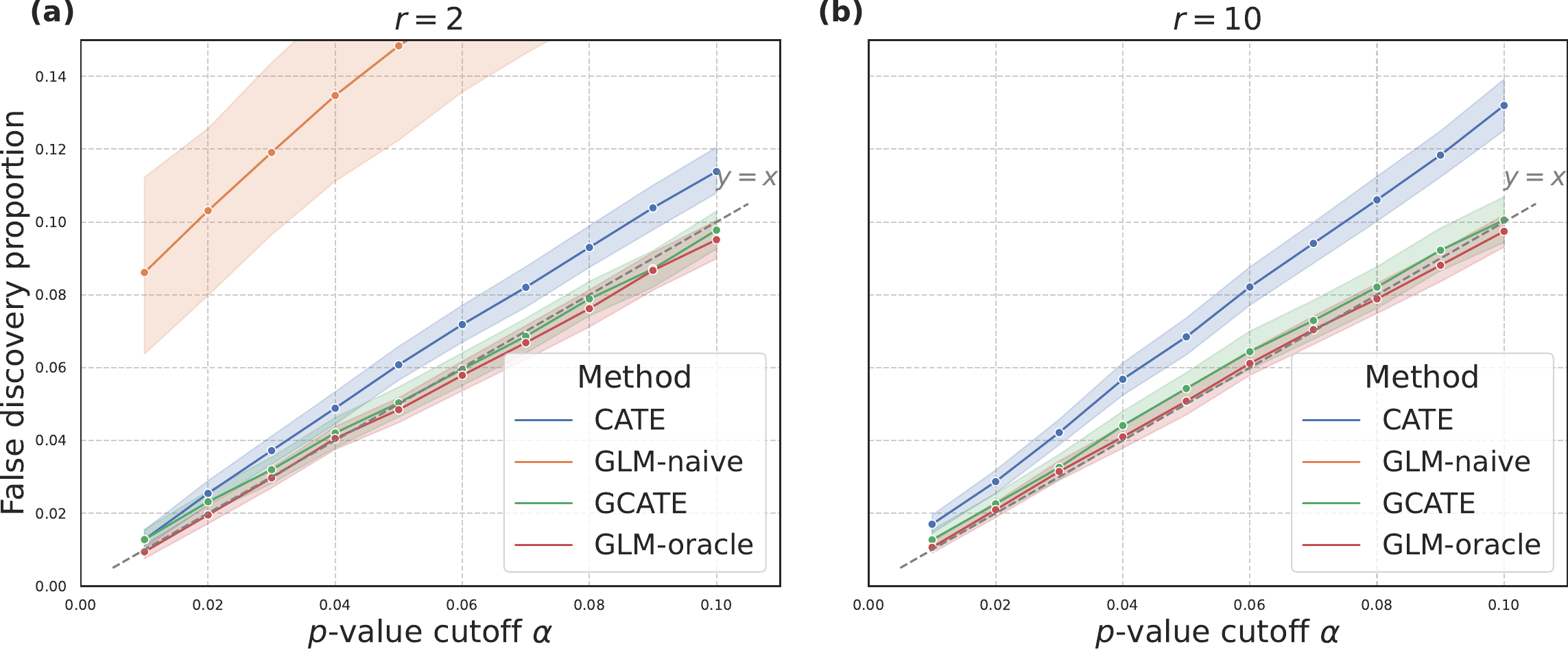}
        
        \caption{False discovery proportion at different $\alpha$ levels for $p$-values adjusted by the Benjamini-Hochberg procedure on 100 simulated datasets when $n=250$.
        The left and right panels show the results for different numbers of latent factors, \textbf{(a)} $r=2$ and \textbf{(b)} $r=10$, respectively.
         When $r=10$, the FDP of \textsc{glm}-naive is above 0.15; hence it is not shown in the figure.}
        \label{fig:ex2-fdp-vs-alpha}
    \end{figure}

    Four methods are applied to the simulated datasets:
    (1) \textsc{cate} (confounder adjustment for testing and estimation), which is a unified approach for surrogate variable analysis under linear models \citep{wang2017confounder} and operates on the log-normalized data.
    (2) \textsc{glm}-naive, which fits generalized linear models with the Poisson likelihood method but only uses the measured covariates $\bX$ without adjusting for unmeasured confounders;
    (3) \textsc{glm}-oracle, which fits generalized linear models with the Poisson likelihood method and uses both observed and unobserved covariates $(\bX,\bZ)$ for estimation and testing;
    (4) \textsc{gcate}, our proposed method with the Poisson likelihood.
    For \textsc{cate}, we use bi-cross-validation (BCV) \citep{owen2016bi} to select the number of factors, as suggested in their original paper \citep{wang2017confounder}.
    For \textsc{gcate}, we use JIC described in \Cref{rm:JIC} to select the number of factors.

    To evaluate different methods, we summarize the type-I error and FDP (false discovery proportion) after the Benjamini-Hochberg procedure in \Cref{fig:ex2-err}, where the desirable thresholds for the two are set to be 5\% and 20\%, respectively.
    From \Cref{fig:ex2-err}, we see that when the sample size $n$ is small, or the latent factor dimension $r$ is large, the performance of all methods gets slightly worse, especially those that are misspecified, which is expected.
    In all setups of $(n,r)$, because the multivariate-Gaussian assumptions of \textsc{cate} are violated, it does not provide proper Type-I error control and FDP control. 
    This suggests that \textsc{cate} may inflate test statistics and cause anti-conservative inference.
    Similarly, \textsc{glm}-naive also fails to control the FDPs because it cannot account for dependencies induced by the latent factors.
    On the other hand, \textsc{gcate} performs as well as \textsc{glm}-oracle that has knowledge of the latent factors $\bZ$.
    This indicates that our modeling helps to accurately remove unwarranted sources of confounding effects. Note that variations of \textsc{cate} may yield improved performance using empirical nulls or negative controls, but \textsc{gcate} requires no such tuning.
    
    We further inspect the FDP control of different methods with varying thresholds.
    In the ideal scenarios, FDP aligns closely with the specified $\alpha$ cutoffs.
    From \Cref{fig:ex2-fdp-vs-alpha}, \textsc{glm}-oracle has FDP aligning closely with the specified $\alpha$ cutoffs and consistently performs admirably across different levels of confounding effects.
    Conversely, the \textsc{glm}-naive approach struggles to control the FDP effectively, and this discrepancy becomes increasingly pronounced as the number of latent factors grows. However, in a commendable contrast to \textsc{cate}, our method \textsc{gcate} consistently outperforms in terms of FDP control at various alpha cutoffs. This superiority can be attributed to our method's ability to model the data distribution accurately and eliminate unwarranted variations.
    
    Lastly, we also evaluate the statistical power and precision of different methods.
    Here, the power is evaluated when the Type-I error threshold is 5\%.
    We anticipate that both \textsc{cate} and the \textsc{glm}-naive approach would yield higher power compared to other methods because they tend to allow more discoveries without adequately controlling the Type-I errors (\Cref{fig:ex2-err}). 
    In \Cref{fig:ex2-err}, we observe that \textsc{cate} exhibits the lowest power among the considered methods. In contrast, the \textsc{glm}-naive approach concurrently registers the most insufficient precision among all the methods. 
    As anticipated, the \textsc{glm}-oracle approach boasts the highest power and precision because it operates in an ideal scenario without confounding effects.
    In contrast, our proposed method, \textsc{gcate}, demonstrates a balanced and robust performance concerning power and precision. It achieves a competitive power level while maintaining a significantly higher precision than the \textsc{glm}-naive method. Moreover, \textsc{gcate} outperforms \textsc{cate} regarding both power and precision.
    This suggests that correct modeling of confounding effects boosts the statistical power and precision in high-dimensional datasets.

    \subsection{Misspecified simulated datasets using scRNA simulators}\label{subsec:simu-splatter}
        To better evaluate the performance of various methods, we use the single-cell RNA sequencing data simulator Splatter \citep{zappia2017splatter} to generate simulated count datasets.
        Splatter explicitly models the hierarchical Gamma-Poisson processes that give rise to data observed in scRNA-seq experiments and can model the multiple-faceted variability.
        Thus, the simulated datasets generated by Splatter are similar to real-world datasets and suitable for benchmarking differential expression testing methods.

        \begin{figure}[!t]
            \centering
            \includegraphics[width=0.6\textwidth]{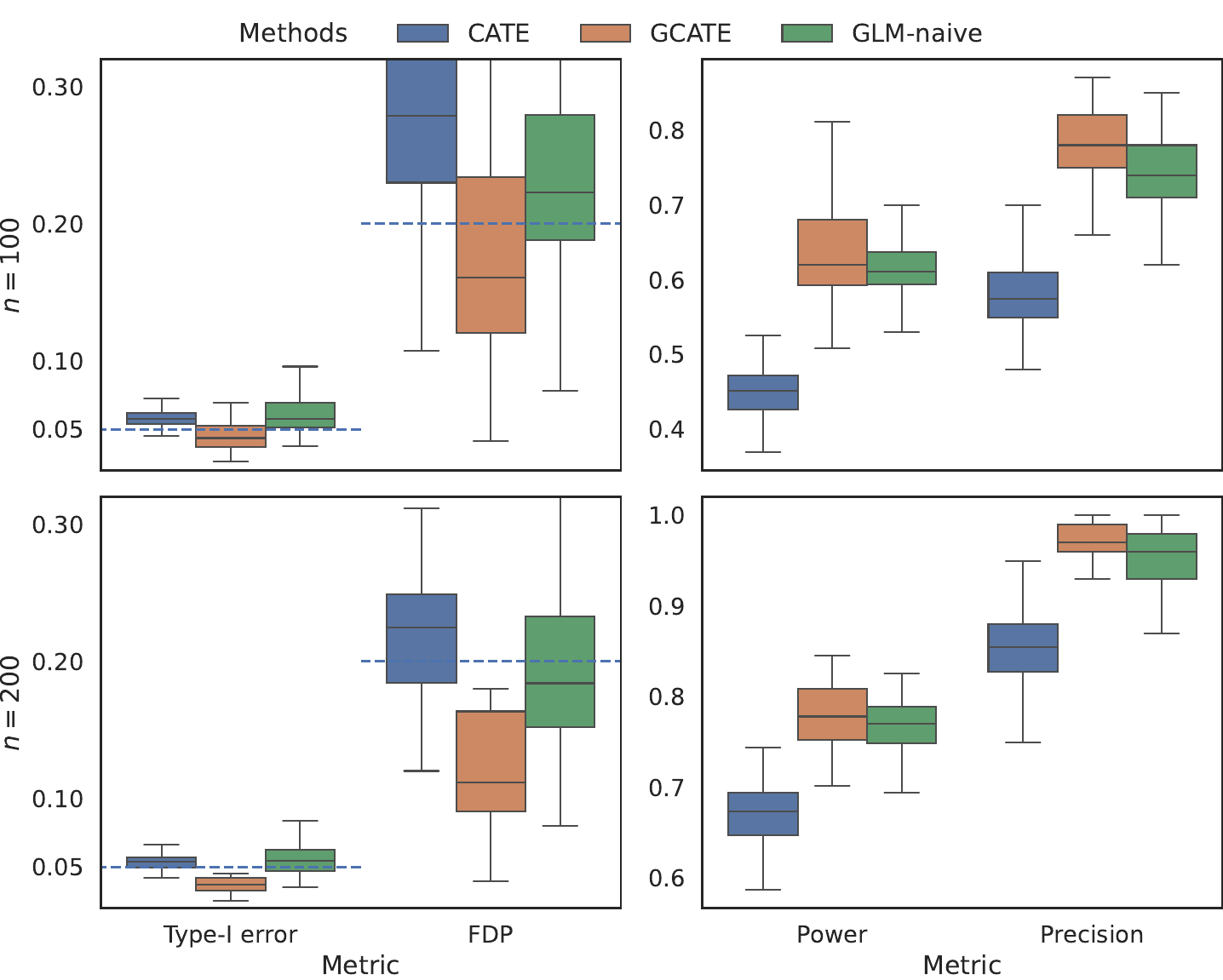}
            \caption{Simulation results on 100 simulated scRNA-seq datasets generated by Splatter with varying numbers of samples $n\in\{100,200\}$.
            The four metrics are shown in four columns respectively.
            The blue dashed lines indicate the desired cutoffs for the statistical errors.
        }
            \label{fig:simu2}
        \end{figure}

    Using Splatter, $n$ cells are sampled from two groups with equal probability for $n\in\{100,200\}$, containing $p=10,000$ genes.
    Because of the sparse nature of the simulated single-cell datasets, about 80\% of the genes are only expressed in 10 cells.
    Hence, we exclude these lowly-expressed genes and evaluate the methods for the remaining genes.    
    We include $d=3$ covariates for each cell: the intercept, the group indicator ($\{\pm 1\}$), and the logarithm of the library sizes, which is the sum of expression across all genes.    
    When simulating the datasets, we use Splatter to generate four batches, introducing three major confounders.
    Because the data is not generated from well-specified GLMs, the oracle model is unknown and hence not included.
    For \textsc{glm}-naive and \textsc{gcate}, we use the NB likelihood with log links to directly model the count data, where the gene-level dispersion parameters are estimated by the method of moments based on the estimated mean returned by using the Poisson likelihood; see \Cref{app:subsec:opt} for more details.
    For \textsc{cate}, we normalize the counts in each cell by its library size, then multiply them by a scale factor of $10^4$ and shift them by one, and finally, apply the logarithm transform, following the standard preprocessing approach of single-cell data.

    Compared to the previous bulk-cell simulation in \Cref{subsec:simu-poisson}, the simulated data from Splatter is sparser and more noisy.
    From \Cref{fig:simu2}, both \textsc{cate} and \textsc{glm} fail to control the Type-I error at level 5\% and have lower power than \textsc{gcate} in this more challenging setting. 
    The primary reason lies in the assumption underlying \textsc{cate} is significantly violated, while the \textsc{glm} approach fails to account for confounding effects. 
    Though \textsc{glm} may have reasonable control over the false discoveries, its power and precision are highly affected by the confounders.
    On the contrary, \textsc{gcate} obtain valid Type-I error and FDP controls and higher power and precision with small sample sizes because of proper distributional modeling.
    The result of \textsc{gcate} is slightly conservative because of model misspecification and zero inflation induced by the Splatter simulator, which could bias the estimates of coefficients $\bB$ towards zero.
    Additionally, the NB distribution involves the additional challenge of estimating the overdispersion parameters.

\section{Lupus data example}\label{sec:case-studeis}
    
    \subsection{The dataset}
    Systemic lupus erythematosus (SLE) is an autoimmune disease predominantly affecting women and individuals of Asian, African, and Hispanic descent.
    \citet{perez2022single} developed multiplexed single-cell RNA sequencing (mux-seq) to capture the complexity of immune cell populations and systematically profile the composition and transcriptional states of immune cells in a large multiethnic cohort. The dataset contains 1.2 million peripheral blood mononuclear cells from 8 major cell types and 261 individuals, including 162 SLE cases and 99 healthy controls of either Asian or European ancestry.
    The cell-type-specific DE analysis aims to provide insights into the diagnosis and treatment of~SLE.

    To remove the genes with small variations, we use the Python package \texttt{scanpy} \citep{wolf2018scanpy} to pre-process the single-cell data and select the top 2,000 highly variable genes (HVGs) within each cell type.
    For each cell type, we aggregate expression across cells from the same subject to obtain gene-level pseudo-bulk counts and then remove genes expressed in less than 10 subjects.
    The basic information of the preprocessed datasets is provided in \Cref{app:subsec-lupus}.
    For each subject, the recorded variables are SLE status (condition), the logarithm of the library size, sex, population, and processing cohorts (4 levels).
    The latter 3 variables, which account for $r=5$ degrees of freedom, are considered to be the measured confounders.

    \subsection{Confounder adjustment}

    For each cell type, we compare four approaches \textsc{glm}, \textsc{gcate}, \textsc{cate}, and \textsc{cate}-mad.    
    The first two approaches are based on the NB GLM model, while the latter two are designed for linear models.
    The last method uses an estimated empirical null \citep{wang2017confounder} based on median absolute deviation (MAD).
    For each method, we consider two variants, using a ``subset'' of covariates and a ``full'' set of covariates, without and with measured confounders included, respectively.
    Only 5 cell types (T4, cM, B cell, T8, NK) contain more than 50,000 single cells and have sufficient power to obtain significant findings using the \textsc{glm}-full approach, so we restrict our comparisons to those types. 
    In particular, we display our results for the largest T4 cell type in this section, and similar results for other cell types are included in \Cref{app:sec:extra-ex-results}.
    To estimate the number of latent factors $r$, we analyze the JIC values according to \Cref{rm:JIC}.
    For a subset of covariates, as shown in \Cref{fig:JIC-r-lupus-T4}, the scree plot reveals a diminishing negative log-likelihood with increasing $r$, which plateaus for $r=4$ to $r=7$, and the decrement becomes marginal beyond $r=7$. Consequently, we recommend selecting $r=7$ for \textsc{gcate}-subset analysis, and similarly, $r=2$ for \textsc{gcate}-full analysis.
    We also conduct the sensitivity analysis for the number of latent factors in \Cref{subsubsec:sensitivity}.

    Our first analysis is to treat \textsc{glm}-full as \textsc{glm}-oracle and inspect the performance of all four methods without measured confounders included.
    The majority of the test statistics obtained for \textsc{glm}-full are well approximated by a standard normal distribution, which suggests that the experiment conducted by \cite{perez2022single} was well controlled, and the impact of unmeasured confounders was negligible (\Cref{fig:lupus-zscores}).
    However, when we excluded the measured confounders, the \textsc{glm}-subset statistics were poorly calibrated, indicating that controlling for these variables is essential to proper analysis, either directly or indirectly. The \textsc{cate} statistics are even more poorly calibrated than \textsc{glm}-subset, suggesting that these sparse data cannot be modeled using a linear model, though restricting the test to the top 250 HVGs yields test statistics closer to the expected distribution (\Cref{fig:lupus-cate-250}).
    With the empirical null adjustment, \textsc{cate}-mad performed somewhat better, but this adaptation is insufficient, suggesting that \textsc{cate} cannot remove the confounding effects when the data are unsuitable for a linear model.
    Finally, the performance of \textsc{gcate} is ideal: the majority of the statistics are well approximated by the standard normal, and a few signals can be captured on the right tail.    
    Similar results were obtained for each of the 5 biggest cell types, as shown in \Cref{fig:lupus-zscores-others}.

    \begin{figure}[!t]
        \centering
        \includegraphics[width=0.9\textwidth]{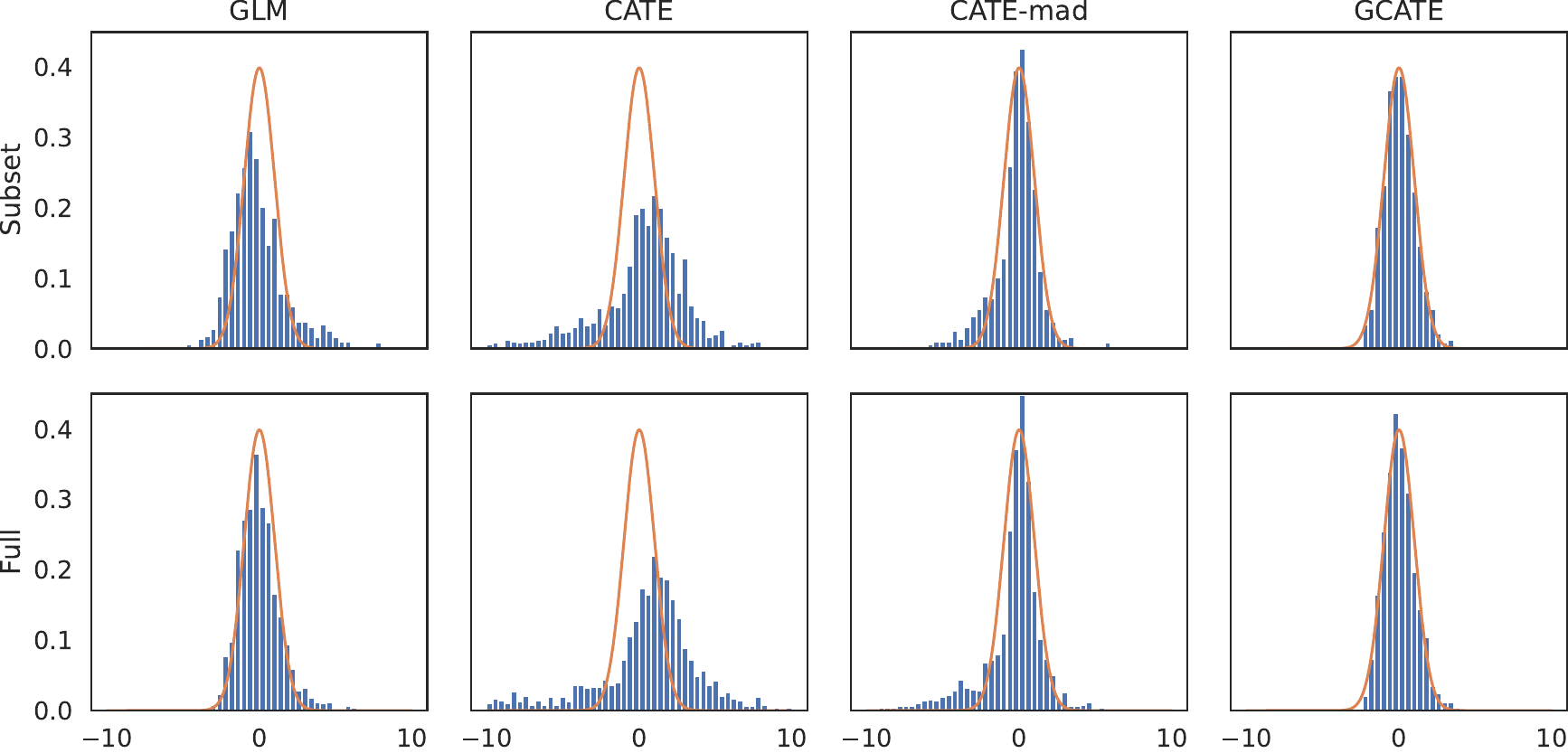}
        \caption{Results on the lupus datasets. 
        Histograms of lupus $z$-statistics of different methods on T4 cell type.
        The first row uses only a subset of the covariates, while the second row uses the full set of covariates for all the methods.
        The orange curves represent the standard normal density.
        }\label{fig:lupus-zscores}
    \end{figure}

    For comparison, we label genes based on the \textsc{glm}-full analysis with FDR control at cutoff 0.2 as ``true positives", resulting in 72 significant genes for the T4 cell type.
    With FDR control at cutoff 0.2, 15 of the 16 \textsc{gcate}'s statistics overlap with the true positives, indicating that the test loses power when the impact of confounders has to be removed using factor analysis.
    Still, the test appears to control the error rate.
    To illustrate the performance of the 4 competing analysis methods across all 5 large cell types, we calculate the precision and specificity using 0.2 as a cutoff for false discovery rate control.
    As shown in \Cref{fig:lupus-violin}, only \textsc{gcate} achieves uniformly high precision and specificity.

    To compare different methods in the biological significance of the discoveries, we conduct gene ontology over-representation analysis to identify the related biological processes.
    As shown in \Cref{fig:treemap}, both \textsc{glm}-full and \textsc{gcate}-subset discover genes that are pertinent to the immune-response-related pathways, which also appear in prior studies on lupus \citep[Fig. 3]{perez2022single}.
    On the other hand, though hundreds of significant genes are claimed by \textsc{cate}-mad, they are not associated with meaningful biological pathways.
    The results indicate that \textsc{gcate} identifies scientifically more relevant genes than \textsc{cate} under unmeasured confounders.

    Our second analysis is to compare the four methods when all the measured covariates are included.
    As shown in the second row of \Cref{fig:lupus-zscores}, we observe similar performance for each of the three methods (\textsc{cate}, \textsc{cate}-mad, and \textsc{gcate}) remains similar whether partial or all covariates are included.
    In particular, we see that the test results of \textsc{cate} and \textsc{cate}-mad get more anti-conservative.
    On the other hand, as shown in \Cref{fig:lupus-T4-upset}, the results of \textsc{gcate}  are consistent and more powerful with added covariates, although they exhibit lower power than \textsc{glm}.
    This is expected as, in general, the estimated latent factors may remove some signals for confounder adjustment methods.
    Furthermore, the GO analysis of the biological processes given in \Cref{fig:treemap-full} aligns closely with the results in the first analysis using a subset of the covariates, suggesting the biological relevance of the findings from \textsc{gcate}. 
    Overall, \textsc{gcate} demonstrates robust performance and consistency across various levels of confounding.

\section{Discussion}
    We presented novel estimation and inference procedures for multivariate generalized linear models with unmeasured confounders in the high-dimensional scenarios when both the sample size $n$ and response size $p$ tend to infinity.
    Our approach consists of three main steps. 
    In the first step, we disentangle the marginal effects from the uncorrelated confounding effects, recovering the column space of latent coefficients $\hat{\bGamma}$ from the latter. 
    We provide non-asymptotic estimation error bounds for both the estimated natural parameter matrix $\hat{\bTheta}_0$ and the projection onto the column space of $\hat{\bGamma}$. 
    In the second step, we estimate both latent factors $\bZ$ and primary effects $\bB$ by solving a constrained lasso-type problem that confines $\bB$ to the orthogonal space of $\hat{\bGamma}$. 
    From the column-wise estimation error of the latent components, we obtain the estimation error for the primary effects in the presence of nuisance parameters. 
    In the third step, we design an inferential procedure to correct the bias introduced by $\ell_1$-regularization and establish Type-I error and family-wise error rate controls.

    Numerically, we demonstrate the usage of the proposed method with Poisson and Negative Binomial likelihoods for bulk-cell and single-cell simulations, respectively.
    Compared to alternative methods, the proposed method effectively controls the Type-I error and false discovery proportion while delivering enhanced statistical power and precision as the count data get sparser and more over-dispersed.
    Furthermore, our analysis of real single-cell datasets underscores the importance of accounting for confounding effects when major covariates are unobserved. 
    Notably, our proposed method consistently outperforms alternative techniques, demonstrating superior precision and specificity, thus establishing its suitability for high-dimensional sparse count data.
 
    The present study, while offering valuable insights, is not without its limitations and opportunities for future exploration. 
    Some of these include the development of hypothesis testing for confounding effects, the theoretical guarantee of the FDR, and more robust criteria for selecting the optimal number of latent factors.
    Recent works by \citet{dai2023scale} and \citet{chen2022determining} offer promising insights that may contribute to resolving some of these challenges.
    Although we have briefly touched upon the applicability of our proposed method under non-canonical link functions in \Cref{app:subsec-nb-log-link}, comprehensive theoretical guarantees remain an area deserving of further research and investigation.

\section*{Acknowledgement}
    We sincerely thank Kevin Lin for generously sharing his carefully prepared manuscript on cohort eSVD, which has proven to be an indispensable asset in the advancement of our current research endeavor.
    We are indebted to Pratik Patil, Catherine Wang, F. William Townes, Zach Branson, Eli Ben Michael, and all members of the CMU GenStats Lab Group and CMU Causal Reading Group for suggestions on the earlier draft and many insightful conversations.
    This work used the Bridges-2 system at the Pittsburgh Supercomputing Center (PSC) through allocations BIO220140 and MTH230023 from the Advanced Cyberinfrastructure Coordination Ecosystem: Services \& Support (ACCESS) program. This project was funded by the National Institute of Mental Health (NIMH) grant R01MH123184.
    We thank the editor, the associate editor, and two anonymous reviewers for their valuable and constructive comments, which significantly improved this paper.

% {
% \spacingset{1.1}\small
% \putbib[ref]
% % \bibliographystyle{apalike}
% % \bibliography{ref}
% }
% \end{bibunit}            

% \begin{bibunit}[apalike]                
\appendix

\counterwithin{theorem}{section}
\renewcommand{\thesection}{\Alph{section}}
\setcounter{table}{0}
\renewcommand{\thetable}{\thesection\arabic{table}}
\setcounter{figure}{0} 
\renewcommand\thefigure{\thesection\arabic{figure}}
\renewcommand{\thealgorithm}{\thesection.\arabic{algorithm}}

\clearpage
\begin{center}
    {\LARGE\bf Appendix}
\end{center}

\bigskip

The appendix includes the proof for all the theorems, computational details, and extra experiment results.
The structure of the appendix is listed below:

\bigskip

\begin{table}[!ht]
\centering
\begin{tabularx}{0.95\textwidth}{l l l }
    \toprule
    \multicolumn{2}{c}{\textbf{Appendix}} & \textbf{Content} \\
    \midrule \addlinespace[0.5ex] 
    \Cref{app:sec:prop:iden-B} & \cellcolor{lightgray!25} & Proof of \Cref{prop:iden-B}\\ \addlinespace[0.5ex] \cmidrule(l){1-3}\addlinespace[0.5ex] 
    \multirow{2}{*}{\Cref{app:sec:thm:est-error-bTheta}} & \ref{app:subsec:est-Theta} & Proof of \Cref{thm:est-error-bTheta}. \\
    & \ref{thm:est-error-bTheta-tech-lems} & \Crefrange{lem:upper-likelihood-diff}{lem:oper-norm-subexp}, which are used in the proof of \Cref{thm:est-error-bTheta}.\\ \addlinespace[0.5ex] \cmidrule(l){1-3}\addlinespace[0.5ex] 
    \multirow{3}{*}{\Cref{app:est-confound}} & \ref{app:est-confound-pre} & Preparatory definitions. \\
    & \ref{app:subsec:thm:est-error-proj-bGamma} & Proof of \Cref{thm:est-error-proj-bGamma}.\\
    &\ref{app:subsec:thm:est-error-proj-bGamma-tech-lems} & \Crefrange{lem:est-error-bGamma}{lem:bound-latent}, which are used in the proof of \Cref{thm:est-error-proj-bGamma}.\\ \cmidrule(l){1-3}
    \multirow{4}{*}{\Cref{app:est-direct}} & \ref{app:est-direct-pre} & Preparatory definitions. \\
    & \ref{app:subsec:cor:est-confound-col} & Proof of \Cref{cor:est-confound-col}.\\
    &\ref{app:subsec:thm:est-err-B} & Proof of \Cref{thm:est-err-B}\\
    &\ref{app:subsec:est-latent-direct-tech-lems} & \Crefrange{lem:optimality}{lem:boundedness}, which are used in the proof of \Cref{thm:est-err-B}.\\ \addlinespace[0.5ex] \cmidrule(l){1-3}\addlinespace[0.5ex]    
    \multirow{3}{*}{\Cref{app:sec:thm:normality}} & \ref{app:subsec:thm:normality} & Proof of \Cref{thm:normality}.\\
    & \ref{app:subsec:prop:simul-inference} & Proof of \Cref{prop:simul-inference}.\\
    & \ref{app:subsec:inference-tech-lems} & \multicolumn{1}{p{0.68\textwidth}}{\Crefrange{lem:sol-u}{lem:hsigma}, which are used in the proof of \Cref{thm:normality} and \Cref{prop:simul-inference}.}\\ \addlinespace[0.5ex] \cmidrule(l){1-3}\addlinespace[0.5ex] 
    \multirow{4}{*}{\Cref{app:sec:comp}} & \ref{app:subsec:exp-family} & Summary of commonly used exponential families.\\
    &\ref{app:subsec:opt} & \multicolumn{1}{p{0.68\textwidth}}{Initialization procedure, alternative maximization algorithm, and the estimation of dispersion parameters.}\\
    & \ref{app:subsec:hyperparam} &  Choosing hyperparameters in practice.\\
    &\ref{app:subsec-nb-log-link} & Discussion about non-canonical links.\\ \addlinespace[0.5ex] \cmidrule(l){1-3}\addlinespace[0.5ex] 
    \multirow{4}{*}{\Cref{app:sec:extra-ex-results}} & \ref{app:subsec:simu-sample-split} & Efficiency loss of sample splitting. \\
    &\ref{subsec:est-err}&The blessing of dimensionality.\\
     & \ref{app:subsec-lupus} & Information about lupus data. \\
     & \ref{app:subsec-extra} & Extra results on lupus datasets. \\
    \addlinespace[0.5ex] \arrayrulecolor{black}
    \bottomrule
\end{tabularx}
\end{table}
\addcontentsline{toc}{part}{\appendixname}

\clearpage
\section{Proof of \Cref{prop:iden-B}}\label{app:sec:prop:iden-B}

 \begin{proof}[Proof of \Cref{prop:iden-B}]
    Because the one-parameter exponential family is minimal, the natural parameter space is convex, and the log-partition function $A$ is strictly convex.
    Based on the information of the first moment of $\by$ and the log-partition function $A$, we can identify $ \bB\bx + \bGamma\bz = \btheta = A'^{-1}(\EE[\by])$.
    Because $\bGamma\bw$ has zero mean and is uncorrelated to $\bx$, $\Cov(\bGamma\bw)=\bGamma\bSigma_w\bGamma^{\top}$ can be identified as the residual covariance of regression of $\btheta$ on $\bx$.
                
    Because $\lambda_r(\bGamma\bSigma_w\bGamma^{\top})\geq \tau_p$, $\bGamma$ and $\bSigma_w$ have full rank.
    Let $\bU_r\bLambda_r\bU_r^{\top}$ be the reduced eigenvalue decomposition of $\bGamma\bSigma_w\bGamma^{\top}$ where $\bU_r\in\RR^{p\times r}$.
    Note that 
    \begin{align*}
        \cP_{\bGamma} &= \bGamma\bSigma_w^{1/2}(\bSigma_w^{1/2}\bGamma^{\top}\bGamma\bSigma_w^{1/2})^{-1}\bSigma_w^{1/2}\bGamma^{\top} \\
        &= \bU_r\bLambda_r^{\frac{1}{2}} (\bLambda_r^{\frac{1}{2}}\bU_r^{\top}\bU_r\bLambda_r^{\frac{1}{2}})^{-1} \bLambda_r^{\frac{1}{2}}\bU_r^{\top}\\
        &= \bU_r\bLambda_r^{\frac{1}{2}} (\bLambda_r^{\frac{1}{2}}\bLambda_r^{\frac{1}{2}})^{-1} \bLambda_r^{\frac{1}{2}}\bU_r^{\top}\\
        &=\bU_r\bU_r^{\top}.
    \end{align*}
    Thus, $\cP_{\bGamma}$ can be recovered.

    By the orthogonal decomposition, we have $\bB = \cP_{\bGamma}^{\perp}\bB + \cP_{\bGamma}\bB $.
    Let $\be_{p,i}=(\delta_{i\ell})_{1\leq \ell\leq p}$ and $\be_{d,j}=(\delta_{j\ell})_{1\leq \ell\leq d}$.
    We consider the $(i,j)$-th entry of $\cP_{\bGamma}\bB$:
    \begin{align}
        |\be_{p,i}^{\top} \cP_{\bGamma}\bB \be_{d,j}| &= |\be_{p,i}^{\top} \bGamma\bSigma_w^{1/2}(\bSigma_w^{1/2}\bGamma^{\top}\bGamma\bSigma_w^{1/2})^{-1}\bSigma_w^{1/2}\bGamma^{\top}\bB \be_j| \notag\\
        &\leq \| \bGamma\bSigma_w^{1/2}(\bSigma_w^{1/2}\bGamma^{\top}\bGamma\bSigma_w^{1/2})^{-1}\bSigma_w^{1/2}\bGamma^{\top}\be_{p,i}\|_{\infty}\cdot\|\bB \be_{d,j}\|_1 \notag\\
        &= \max_{\ell\in[p]}| \be_{\ell}^{\top}\bGamma\bSigma_w^{1/2}(\bSigma_w^{1/2}\bGamma^{\top}\bGamma\bSigma_w^{1/2})^{-1}\bSigma_w^{1/2}\bGamma^{\top}\be_{p,i}|\cdot\|\bB \be_{d,j}\|_1 \notag\\
        &\leq  \max_{\ell\in[p]}\|\bSigma_w^{1/2}\bGamma^{\top}\be_{p,\ell}\|_{2} \cdot \|(\bSigma_w^{1/2}\bGamma^{\top}\bGamma\bSigma_w^{1/2})^{-1}\bSigma_w^{1/2}\bGamma^{\top} \be_{p,i}\|_2 \cdot \|\bB \be_{d,j}\|_1 \notag\\
        &\leq \max_{\ell\in[p]}\|\bSigma_w^{1/2}\bGamma^{\top}\be_{p,\ell}\|_{2}\cdot
        \|(\bSigma_w^{1/2}\bGamma^{\top}\bGamma\bSigma_w^{1/2})^{-1}\|_{\oper} \cdot\|\bSigma_w^{1/2}\bGamma^{\top} \be_{p,i}\|_2 \cdot\|\bB_j\|_1 \notag\\
        &\leq \max_{\ell\in[p]}\|\bSigma_w^{1/2}\bGamma^{\top}\be_{p,\ell}\|_{2}\cdot
        \lambda_r(\bGamma\bSigma_w\bGamma^{\top})^{-1}\cdot\|\bSigma_w^{1/2}\bGamma^{\top} \be_{p,i}\|_2 \cdot\|\bB_j\|_1 \notag\\
        &=\cO\left(\frac{\|\bB_j\|_1}{\tau_p}\right) \notag\\
        &=o(1), \label{eq:proof-prop-norm}
    \end{align}
    where the first two inequalities are from Holder's inequality;
    the third inequality holds because of the sub-multiplicativity of the operator norm; and the last inequality holds because $\bSigma_w^{1/2}\bGamma^{\top}\bGamma\bSigma_w^{1/2}$ and $\bGamma\bSigma_w\bGamma^{\top}$ have the same non-zero eigenvalues.
    Then we have
    \begin{align*}
        \|\cP_{\bGamma}\bB\|_{\fro} &\leq \sqrt{p} \max_{1\leq i\leq p}\| \bB^{\top}\cP_{\bGamma} \be_{p,i}\|_2 \lesssim \frac{\sqrt{p}\|\bB\|_{1,1}}{\tau_p}. 
    \end{align*}
    Thus, the conclusion follows.
\end{proof}

\section{Estimation error of natural parameters by alternative maximization}\label{app:sec:thm:est-error-bTheta}

In this section, we gather useful results to bound the estimation error for the natural parameter matrix.
Let $E_C = \{\bTheta^*\in \cR_{C}^{n\times p}\}$ be the event that all the natural parameters are bounded.
From \Cref{asm:model}, we know that $\PP(E_C)=\iota_n \rightarrow 1$ as $n\rightarrow \infty$.
Under event $E_C$, because $A$ is strictly convex and trice continuously differentiable, we have that 
\begin{align}
    \kappa_1:=\inf_{\theta\in\cR_C} A''(\theta)>0\text{ and }\kappa_2:=\sup_{\theta\in\cR_C} A''(\theta)<\infty. \label{eq:kappa}
\end{align}
These facts enable us to derive \Cref{thm:est-error-bTheta}, which will be used in \Cref{app:est-confound} for proving \Cref{thm:est-error-proj-bGamma} and in \Cref{app:est-direct} for proving \Cref{thm:est-err-B}.

\subsection{Estimation error of natural parameters}\label{app:subsec:est-Theta}

\begin{proof}[Proof of \Cref{thm:est-error-bTheta}]
    We split the proof into two parts under event $E_C$.
    
    \paragraph{Part (1) Bounding $\|\hat{\bTheta}_0 - \bTheta^*\|_{\fro}$.} 
    From the assumption of \Cref{thm:est-error-bTheta}, we have    
    \begin{align*}
        \cL(\bTheta^*) - \cL(\hat{\bTheta}_0)  \geq 0,
    \end{align*}
    which also holds when $\hat{\bTheta}_0$ is the maximum likelihood estimator.
    From \Cref{lem:upper-likelihood-diff} it further follows that
    \begin{align*}
        0 \leq \sqrt{2(d+r)}\|\bY-A'(\bTheta^*)\|_{\oper} \|\hat{\bTheta}_0-\bTheta^*\|_{\fro} - \frac{\kappa_2}{2} \|\hat{\bTheta}_0-\bTheta^*\|_{\fro}^2.
    \end{align*}
    Thus, we have
    \begin{align*}
        \|\hat{\bTheta}_0-\bTheta^*\|_{\fro} &\leq \frac{2\sqrt{2(d+r)}}{\kappa_2} \|\bY-A'(\bTheta^*)\|_{\oper}.
    \end{align*}
    Next, we bound the operator norm of $\bY-A'(\bTheta^*)$.
    Conditional on $\bX$ and $\bZ^*$, observe that $e_{ij}:=y_{ij}-A'(\theta_{ij}^*)$ ($i\in[n]$ and $j\in[p]$) are independent, zero-mean, and sub-exponential with parameters $\nu=\sqrt{\kappa_2}$ and $\alpha=1/C^2$.
    To see this, note that its moment generating function is $\EE[\exp(te_{ij})]=\exp(A(\theta_{ij}^* +t ) - A(\theta_{ij}^*  )  - t A'(\theta_{ij}^*)) = \exp(A''(\theta_{ij}^* +t' )t^2/2)$ for some $|t'|<|t|$.
    By \Cref{asm:model}, we have $\EE[\exp(te_{ij})]\leq \kappa_2t^2/2$ for all $|t|<C^2$, which shows that $e_{ij}$ is sub-exponential.
    By \Cref{lem:oper-norm-subexp}, for any $\delta>0$, with probability at least $1-(n+p)^{-\delta}-(np)^{-\delta}$, it follows that
    \begin{align*}
        \|\hat{\bTheta}_0-\bTheta^*\|_{\fro} &\leq \frac{2\sqrt{2(d+r)}}{\kappa_2}(4 \nu \sqrt{n\vee p} + 2\delta^{\frac{3}{2}}\sqrt{c}(\alpha\vee \nu) \log(np) \sqrt{\log(n+p)} ) \\
        &\lesssim \sqrt{(d+r)(n\vee p)}.
    \end{align*}

    \paragraph{Part (2) Bounding $\max_{1\leq j\leq p}\|(\hat{\bTheta}_0)_{j} - \bTheta^*_{j}\|_2$.}
    Similarly to Part (1), by union bound, we have
    \begin{align*}
        \max_{1\leq j\leq p} \|(\hat{\bTheta}_0)_j-\bTheta^*_j\|_{2} &\leq \max_{1\leq j\leq p}\frac{2\sqrt{2(d+r)}}{\kappa_2} \|\bY_j-A'(\bTheta^*_j)\|_{2}\\
        &\leq \frac{2\sqrt{2(d+r)}}{\kappa_2}(4 \nu \sqrt{n} + 2(\delta+1)^{\frac{3}{2}}\sqrt{c}(\alpha\vee \nu) \log(n) \sqrt{\log(n+1)} ) ,
    \end{align*}
    with probability at least $1-p(n+p)^{-\delta-1}-p(np)^{-\delta-1}\geq 1-(n+p)^{-\delta}-(np)^{-\delta}$, for any $\delta>0$.
    
    For $\delta>1$, taking union bound over the above two events and $E_C$ finishes the proof.
\end{proof}

\subsection{Technical lemmas}\label{thm:est-error-bTheta-tech-lems}

    \begin{lemma}[Upper bound of likelihood difference]\label{lem:upper-likelihood-diff}
    Suppose that $\bTheta_1\in\cR_{r_1},\bTheta_2\in\cR_{r_2}$ with $r_j=\rank(\bTheta_j)$ for $j=1,2$.
    Define $\kappa_1:=\inf_{\theta\in\cR} A''(\theta)$. Then it holds that
    \begin{align*}
        \cL(\bTheta_2) - \cL(\bTheta_1) \leq \frac{\sqrt{r_1+r_2}}{n}\|\bY-A'(\bTheta_2)\|_{\oper} \|\bTheta_1-\bTheta_2\|_{\fro} - \frac{\kappa_1}{2n} \|\bTheta_1-\bTheta_2\|_{\fro}^2,
    \end{align*}
    and
    \begin{align*}
        \cL(\bTheta_2) -\cL(\bTheta_1) \leq \frac{\sqrt{r_1+r_2}}{n}\|\bY-A'(\bTheta_1)\|_{\oper} \|\bTheta_1-\bTheta_2\|_{\fro} + \frac{\kappa_2}{2n} \|\bTheta_1-\bTheta_2\|_{\fro}^2.
    \end{align*}
\end{lemma}
\begin{proof}[Proof of \Cref{lem:upper-likelihood-diff}]
    Recall that $\cL(\bTheta)=n^{-1}[- \tr(\bY^{\top}\bTheta) + \tr(\one_{p\times n} A(\bTheta))]$.
    Then we have
    \begin{align}
        \cL(\bTheta_2) - \cL(\bTheta_1)  &= \frac{1}{n}\tr((\bY-A'(\bTheta_2))^{\top}(\bTheta_1-\bTheta_2)) \notag\\
        &\qquad - \frac{1}{n}\tr(\one_{p\times n}(A(\bTheta_1) - A(\bTheta_2)) - A'(\bTheta_2)^{\top}(\bTheta_1-\bTheta_2) ).\label{eq:lem:upper-likelihood-diff-eq-1}
    \end{align}
    Next, we analyze the two terms separately.

    For the first term, we have
    \begin{align}
        &\tr((\bY-A'(\bTheta_2))^{\top}(\bTheta_1-\bTheta_2)) \\
        \leq&  \sqrt{\rank(\bTheta_1-\bTheta_2)}\|\bY-A'(\bTheta_2)\|_{\oper}\|\bTheta_1-\bTheta_2\|_{\fro}\notag\\
        \leq& \sqrt{\rank(\bTheta_1)+\rank(\bTheta_2)}\|\bY-A'(\bTheta_2)\|_{\oper}\|\bTheta_1-\bTheta_2\|_{\fro}, \label{eq:lem:upper-likelihood-diff-eq-2}
    \end{align}
    where the first inequality is from the matrix norm inequality $|\tr(\bA^{\top}\bB)|\leq \sqrt{\rank(\bB)}\|\bA\|_{\oper}\|\bB\|_{\fro}$ and the last inequality is due to the fact that $\rank(\bA+\bB)\leq \rank(\bA)+\rank(\bB)$.

    For the second term, note that each entry inside the trace takes the form
    \begin{align*}
        A((\bTheta_{1})_{ij}) - A((\bTheta_{2})_{ij}) - A'((\bTheta_{2})_{ij})((\bTheta_{1})_{ij}-(\bTheta_{2})_{ij}) &= \frac{1}{2}A''(\theta)((\bTheta_{1})_{ij}-(\bTheta_{2})_{ij})^2 \\
        &\geq \frac{\kappa_1}{2}((\bTheta_{1})_{ij}-(\bTheta_{2})_{ij})^2
    \end{align*}
    where the first equality is from Taylor expansion with $\theta$ lies between $(\bTheta_{1})_{ij}$ and $(\bTheta_{2})_{ij}$ and the definition of $\kappa_1:=\inf_{\theta\in\cR} A''(\theta)\geq 0$.
    Thus we have
    \begin{align}
        \tr(\one_{p\times n}(A(\bTheta_1) - A(\bTheta_2)) - A'(\bTheta_2)^{\top}(\bTheta_1-\bTheta_2) ) &\geq \frac{\kappa_1}{2}\|\bTheta_1-\bTheta_2\|_{\fro}^2. \label{eq:lem:upper-likelihood-diff-eq-3}
    \end{align}

    Combining \eqref{eq:lem:upper-likelihood-diff-eq-1}, \eqref{eq:lem:upper-likelihood-diff-eq-2}, and \eqref{eq:lem:upper-likelihood-diff-eq-3} finishes the proof of the first inequality.
    Similarly, we have
    \begin{align*}
        \cL(\bTheta_2) - \cL(\bTheta_1)  &= \frac{1}{n}\tr((\bY-A'(\bTheta_1))^{\top}(\bTheta_1-\bTheta_2)) \notag\\
        &\qquad + \frac{1}{n}\tr(\one_{p\times n}(A(\bTheta_2) - A(\bTheta_1)) - A'(\bTheta_1)^{\top}(\bTheta_2-\bTheta_1) )\\
        &\leq \frac{\sqrt{r_1+r_2}}{n}\|\bY-A'(\bTheta_1)\|_{\oper} \|\bTheta_1-\bTheta_2\|_{\fro} + \frac{\kappa_2}{2n}\|\bTheta_1-\bTheta_2\|_{\fro}^2,
    \end{align*}
    which completes the proof of the second inequality.
\end{proof}

\begin{lemma}[Operator norm of matrices with sub-exponential entries]\label{lem:oper-norm-subexp}
    Let $\bX=(x_{ij})_{i\in[n],j\in[p]}$ be a matrix with independent and centered entries such that $x_{ij}$'s are $(\nu,\alpha)$-sub-exponential random variables\footnote{Here we adopt the definition from \citet[Definition 2.7]{wainwright2019high}. In some literature, the term `sub-gamma' is used interchangeably with `sub-exponential' to refer to this definition.} with parameters $\nu,\alpha>0$:
    \[\EE[\exp(t x_{ij})] \leq \exp(t^2 \nu^2/2),\qquad \forall\ |t|<\frac{1}{\alpha}.\]
    Then for all $\delta>0$, there exists a universal constant $c>0$ such that, with probability at least $1-(n+p)^{-\delta} - (np)^{-\delta}$, when $n,p$ are large enough, it holds that
    \begin{align*}
        \|\bX\|_{\oper} &\leq 4 \nu \sqrt{n\vee p} + 2\delta^{3/2}\sqrt{c}(\alpha\vee \nu) \log(np) \sqrt{\log(n+p)}.
    \end{align*}
\end{lemma}
\begin{proof}[Proof of \Cref{lem:oper-norm-subexp}]
    We define a symmetric matrix
    \[\bZ=\left(z_{ij}\right)=\left(\begin{array}{cc}0 & \tilde{\bX} \\ \tilde{\bX}^{\top} & 0\end{array}\right) \in \RR^{(n+p)\times(n+p)},\]
    where $\tilde{\bX}:=(\tilde{x}_{ij})=\bX-\bX^{\prime}$ and $\bX^{\prime}=\left(x_{i j}^{\prime}\right)$ is an independent copy of $\bX$.
    Because $\tilde{x}_{ij}$'s have symmetric distribution and are independent, it follows that $z_{ij}$'s are also independent and symmetric random variables, and $\|\bZ\|_{\oper} = \|\tilde{\bX}\|_{\oper}$.
    Because the tail event $\ind\{\|\bX-\bX'\|_{\oper} \geq t\}$ is a convex function on $\bX'$, by Jensen's inequality we have 
    \[\ind\{\|\bX\|_{\oper} \geq t\} = \ind\{\|\bX-\EE_{\bX'}[\bX']\|_{\oper} \geq t\} \leq   \EE_{\bX'}[\ind\{\|\bX-\bX'\|_{\oper} \geq t\}] = \EE_{\bX'}[\ind\{\|\tilde{\bX}\|_{\oper} \geq t\}].\]
    By Fubini's theorem, it follows that
    \[\PP(\|\bX\|_{\oper} \geq t) \leq \EE_{\bX,\bX'}[\ind\{\|\tilde{\bX}\|_{\oper} \geq t\}] = \PP(\|\tilde{\bX}\|_{\oper} \geq t) = \PP(\|\bZ\|_{\oper} \geq t).\]
    Then, it suffices to bound the tail probability of $\|\bZ\|_{\oper}$.

    We define a truncated random matrix $\bZ(\lambda)$ of $\bZ$,
    \[
    \bZ(\lambda)=\left(z_{ij}(\lambda)\right)_{1 \leq i \leq n, 1 \leq j \leq p}=\left(z_{ij} \ind\left(\left|z_{ij}\right| \leq \lambda\right)\right)_{1 \leq i \leq n, 1 \leq j \leq p}
    \]
    whose entries are independent, symmetric random variables bounded by $\lambda$. 
    By \citet[Corollary 3.12]{bandeira2016sharp}, there exists a universal constant $c>0$ such that
    \[
    \PP\left(\|\bZ(\lambda)\|_{\oper} \geq 2^{\frac{3}{2}} \max _{1 \leq i \leq n+p}\left(\sum_{j=1}^{n+p} \EE[z_{ij}^2(\lambda)]\right)^{\frac{1}{2}}+t\right) \leq(n+p) \exp\left(- \frac{t^2}{c \lambda^2}\right).
    \]
    Note that
    \begin{align*}
        \max _{1 \leq i \leq n+p}\left(\sum_{j=1}^{n+p} \EE[z_{ij}^2(\lambda)]\right)^{\frac{1}{2}} & \leq \max_{1 \leq i \leq n+p}\left(\sum_{j=1}^{n+p} \EE[z_{ij}^2]\right)^{1 / 2} \\
        & = \max \left\{\max_{1 \leq i \leq n}\left(\sum_{j=1}^p \EE[\tilde{x}_{ij}^2]\right)^{\frac{1}{2}}, \max _{1 \leq j \leq p}\left(\sum_{i=1}^N \EE[\tilde{x}_{ij}^2]\right)^{\frac{1}{2}}\right] \\
        &\leq \max\{\sqrt{p},\sqrt{n}\} \max_{i,j}\EE[\tilde{x}_{ij}^2]^{\frac{1}{2}} \\
        & \leq \sqrt{2(n\vee p)}\max_{i,j}\EE[x_{ij}^2]^{\frac{1}{2}}\\
        &\leq \nu\sqrt{2(n\vee p)} ,
    \end{align*}
    where the first inequality is from the definition of the truncated variable, the second and the third inequality is by Cauchy-Schwartz inequality, and the last inequality is because $x_{ij}$ is $(\nu,\alpha)$-sub-exponential.
    Thus, the above two inequality yields that
    \begin{align*}
        \PP\left(\|\bZ(\lambda)\|_{\oper} \geq 4\nu\sqrt{n\vee p} +t\right) &\leq (n+p) \exp\left(-\frac{t^2}{c \lambda^2}\right).
    \end{align*}
    Then we have
    \begin{align*}
        \PP\left(\|\bZ\|_{\oper} \geq 4\nu\sqrt{n\vee p} +t\right) &\leq \PP\left(\|\bZ(\lambda)\|_2 \geq 4\nu\sqrt{n\vee p} +t\right)+\PP\left(\max _{1 \leq i, j \leq n+p}\left|z_{ij}\right|>\lambda\right)    \\
        &\leq  (n+p) \exp\left(-\frac{t^2}{c \lambda^2}\right) + \sum_{1\leq i\leq n}\sum_{1\leq j\leq p} \PP(|\tilde{x}_{ij}|>\lambda)\\
        &\leq (n+p) \exp\left(-\frac{t^2}{c \lambda^2}\right) + np \left(\exp\left(-\frac{\lambda^2}{4\nu^2}\right) \vee \exp\left(-\frac{\lambda}{2\alpha}\right)\right).
    \end{align*}
    where the last inequality follows because $\tilde{x}_{ij} = x_{ij}-x_{ij}'$ is $(2\nu^2,\alpha)$-sub-exponential.
    For all $\delta>0$, let $\lambda  = 2(\delta+1)(\alpha\vee \nu) \log (np)$ and $np\geq 3$, the second term is bounded by $(np)^{-\delta}$.
    Let $t = \lambda ((\delta+1) c\log (n+p))^{1/2}$, the first term is bounded by $(n+p)^{-\delta}$.
    Combining these yields that
    \begin{align*}
        \PP\left(\|\bZ\|_{\oper} \geq 4\nu\sqrt{n\vee p} + 2\delta^{3/2}\sqrt{c}(\alpha\vee \nu) \log (np) \sqrt{\log (n+p)} \right) &\leq (n+p)^{-\delta} + (np)^{-\delta},
    \end{align*}
    which completes the proof.
\end{proof}

\section{Estimation of latent coefficients}
\label{app:est-confound}

\subsection{Preparatory definitions}\label{app:est-confound-pre}
Recall that $\hat{\bW}$ and $\hat{\bGamma}$ are derived from the SVD of $\hat{\bW}_0\hat{\bGamma}_0$ from the first-stage optimization and we have $\hat{\bW}\hat{\bGamma}=\hat{\bW}_0\hat{\bGamma}_0$.
Analogous to \citet{bing2022inference}, we define $\bH_0 := (np)^{-1} \bW^{*\top}\bW^*\bGamma^{*\top}\hat{\bGamma}\bSigma^{-3/2}$ and 
\begin{align}
    \tilde{\bGamma} :=\bGamma^*\bH_0 = (np)^{-1} \bGamma^*\bW^{*\top}\bW^*\bGamma^{*\top}\hat{\bGamma}\bSigma^{-3/2}, \label{eq:tGamma}    
\end{align}
which is identifiable because it depends on both the data $\hat{\bGamma}\bSigma^{-3/2}$ and the identifiable quantity $\bGamma^*\bW^{*\top}\bW^*\bGamma^{*\top}$.
Note that 
\begin{align}
    \cP_{\tilde{\bGamma}} &= \tilde{\bGamma}(\tilde{\bGamma}^{\top}\tilde{\bGamma})^{-1}\tilde{\bGamma}^{\top} \notag\\
    &= \bGamma^*(\bW^{*\top}\bW^*\bGamma^{*\top}\bV)(\bV^{\top}\bGamma^*\bW^{*\top}\bW^*\bGamma^{*\top}\bGamma^*\bW^{*\top}\bW^*\bGamma^{*\top} \bV )^{-1} (\bV^{\top}\bGamma^*\bW^{*\top}\bW^*)\bGamma^{*\top}\notag\\
    &= \bGamma^*(\bGamma^{*\top} \bGamma^* )^{-1} \bGamma^{*\top}\notag\\
    &= \cP_{\bGamma^*} \label{eq:P-tGamma}
\end{align}
because both $\bGamma^{*\top}\bGamma^*$ and $\bW^{*\top}\bW^*\bGamma^{*\top}\bV\in\RR^{r\times r}$ have full rank.
Thus, to quantify the error between $\cP_{\hat{\bGamma}}$ and $\cP_{\bGamma^*}$, we can first analyze the error between $\hat{\bGamma}$ and $\tilde{\bGamma}$.

\subsection{Proof of \Cref{thm:est-error-proj-bGamma}}\label{app:subsec:thm:est-error-proj-bGamma}
\begin{proof}[Proof of \Cref{thm:est-error-proj-bGamma}]
    We split the proof into three parts by bounding the operator norm, column-wise $\ell_2$-norm, and the sup norm consecutively.
    
    \paragraph{Part (1) Bounding the operator norm.}
    From \eqref{eq:P-tGamma}, we have that $\cP_{\tilde{\bGamma}}=\cP_{\bGamma^*}$ for $\tilde{\bGamma}$ defined in \eqref{eq:tGamma}.
    Then we have
    \begin{align*}
        &\|\cP_{\hat{\bGamma}} - \cP_{\bGamma^*}\|_{\oper}\\
        =& \|\cP_{\hat{\bGamma}} - \cP_{\tilde{\bGamma}}\|_{\oper} \\
        =& \|\hat{\bGamma}(\hat{\bGamma}^{\top}\hat{\bGamma})^{-1}\hat{\bGamma}^{\top} - \tilde{\bGamma}(\tilde{\bGamma}^{\top}\tilde{\bGamma})^{-1}\tilde{\bGamma}^{\top}\|_{\oper} \\
        \leq& \| (\hat{\bGamma} - \tilde{\bGamma})(\tilde{\bGamma}^{\top}\tilde{\bGamma})^{-1}\tilde{\bGamma}^{\top}\|_{\oper} + \| \hat{\bGamma}( (\hat{\bGamma}^{\top}\hat{\bGamma})^{-1} - (\tilde{\bGamma}^{\top}\tilde{\bGamma})^{-1} ) \tilde{\bGamma}^{\top}\|_{\oper} + \|\hat{\bGamma}(\hat{\bGamma}^{\top}\hat{\bGamma})^{-1} (\hat{\bGamma} - \tilde{\bGamma})^{\top}\|_{\oper},
    \end{align*}
    where the last inequality is from the triangle inequality.
    Next, we bound the three terms separately. 
    Recall $\tilde{\bGamma}$ is defined in \eqref{eq:tGamma} with $\|\tilde{\bGamma}\|_{\oper} \asymp \|\hat{\bGamma}\|_{\oper} \asymp \sqrt{p}$ by \Cref{asm:latent} and \Cref{lem:spec-D}.
    Then by \Cref{lem:est-error-bGamma} and \Cref{asm:latent}, for all $\delta>0$, the first term in the above display is bounded
    \begin{align*}
        \| (\hat{\bGamma} - \tilde{\bGamma})(\tilde{\bGamma}^{\top}\tilde{\bGamma})^{-1}\tilde{\bGamma}^{\top} \|_{\oper}
        &\leq \| \hat{\bGamma} - \tilde{\bGamma} \|_{\oper} \|(\tilde{\bGamma}^{\top}\tilde{\bGamma})^{-1}\|_{\oper}\|\tilde{\bGamma}\|_{\oper}
        \leq C' (n\wedge p)^{-\frac{1}{2}},
    \end{align*}
    with probability at least $1-2(n+p)^{-\delta}-2(np)^{-\delta} - \exp(-n)$, for some constant $C'>0$ and $\delta>0$.
    Similarly, the third term is also bounded by $\Op\left((n\vee p)^{-1/2}\right)$.
    It remains to bound the second term:
    \begin{align*}
        &\| \hat{\bGamma}( (\hat{\bGamma}^{\top}\hat{\bGamma})^{-1} - (\tilde{\bGamma}^{\top}\tilde{\bGamma})^{-1} ) \tilde{\bGamma}^{\top}\|_{\oper} \\
        =& \| \hat{\bGamma}(\hat{\bGamma}^{\top}\hat{\bGamma})^{-1}( \hat{\bGamma}^{\top}\hat{\bGamma} - \tilde{\bGamma}^{\top}\tilde{\bGamma} ) (\tilde{\bGamma}^{\top}\tilde{\bGamma})^{-1} \tilde{\bGamma}^{\top}\|_{\oper}\\
        \leq& \|\hat{\bGamma}(\hat{\bGamma}^{\top}\hat{\bGamma})^{-1}\|_{\oper} \cdot\| ( \hat{\bGamma}^{\top}\hat{\bGamma} - \tilde{\bGamma}^{\top}\tilde{\bGamma} ) (\tilde{\bGamma}^{\top}\tilde{\bGamma})^{-1} \tilde{\bGamma}^{\top}\|_{\oper}\\
        \lesssim& \frac{1}{\sqrt{p}}\| ( \hat{\bGamma}^{\top}\hat{\bGamma} - \tilde{\bGamma}^{\top}\tilde{\bGamma} ) (\tilde{\bGamma}^{\top}\tilde{\bGamma})^{-1} \tilde{\bGamma}^{\top}\|_{\oper}\\
        \leq & \frac{1}{\sqrt{p}}\| \hat{\bGamma}^{\top}( \hat{\bGamma}  - \tilde{\bGamma}) (\tilde{\bGamma}^{\top}\tilde{\bGamma})^{-1} \tilde{\bGamma}^{\top}\|_{\oper} + \frac{1}{\sqrt{p}}\| ( \hat{\bGamma} - \tilde{\bGamma})^{\top}\tilde{\bGamma}  (\tilde{\bGamma}^{\top}\tilde{\bGamma})^{-1} \tilde{\bGamma}^{\top}\|_{\oper} \\
        \leq & 
        \frac{1}{\sqrt{p}}\|\hat{\bGamma}\|_{\oper}   \| (\hat{\bGamma} - \tilde{\bGamma})(\tilde{\bGamma}^{\top}\tilde{\bGamma})^{-1}\tilde{\bGamma}^{\top} \|_{\oper} + 
        \frac{1}{\sqrt{p}}\|  \hat{\bGamma} - \tilde{\bGamma}\|_{\oper} \|\cP_{\tilde{\bGamma}}\|_{\oper}\\
        \lesssim&  \cO\left((n\wedge p)^{-\frac{1}{2}}\right),
    \end{align*}
    with probability at least $1-(n+p)^{-\delta}-(np)^{-\delta}-\exp(-n)$.
    The proof for the operator norm is completed by combing the above inequality.

    \paragraph{Part (2) Bounding the column-wise $\ell_2$-norm.}
    Let $\be_j\in\RR^p$ be the unit vector such that its $i$-th entry is one if $i=j$ and zero otherwise.
    Similar to Part (1), note that
    \begin{align*}
        &\|(\cP_{\hat{\bGamma}} - \cP_{\bGamma^*})\be_j\|_{2}\\
        \leq& \| (\hat{\bGamma} - \tilde{\bGamma})(\tilde{\bGamma}^{\top}\tilde{\bGamma})^{-1}\tilde{\bGamma}^{\top}\be_j\|_{2} + \| \hat{\bGamma}( (\hat{\bGamma}^{\top}\hat{\bGamma})^{-1} - (\tilde{\bGamma}^{\top}\tilde{\bGamma})^{-1} ) \tilde{\bGamma}^{\top}\be_j\|_{2} + \|\hat{\bGamma}(\hat{\bGamma}^{\top}\hat{\bGamma})^{-1} (\hat{\bGamma} - \tilde{\bGamma})^{\top}\be_j\|_{2}.
    \end{align*}
    The first term can be bounded analogously as
    \begin{align*}
        \max_{1\leq j\leq p}\| (\hat{\bGamma} - \tilde{\bGamma})(\tilde{\bGamma}^{\top}\tilde{\bGamma})^{-1}\tilde{\bGamma}^{\top}\be_j \|_{2}
        &\leq \| \hat{\bGamma} - \tilde{\bGamma} \|_{\oper} \|(\tilde{\bGamma}^{\top}\tilde{\bGamma})^{-1}\|_{\oper} \max_{1\leq j\leq p} \|\tilde{\bGamma}^{\top}\be_j\|_{2}
        \leq C'[ p(n\wedge p)]^{-\frac{1}{2}},
    \end{align*}
    for some constant $C'>0$, by noting that $\|\tilde{\bGamma}^{\top}\be_j\|_{2}=\Op(1)$. 
    The rest of the terms follow a similar argument as in Part (1), under the same probabilistic events therein.

    \paragraph{Part (3) Bounding the sup norm.} The sup norm $\|\cdot \|_{\max}$ can be upper bounded analogously:
    \begin{align*}
        &\|\cP_{\hat{\bGamma}} - \cP_{\bGamma^*}\|_{\max} \\
        =& \max_{i,j\in[p]} |\be_i^{\top}(\cP_{\hat{\bGamma}} - \cP_{\bGamma^*})\be_j| \\
        \leq& \max_{i,j\in[p]} (| \be_i^{\top}(\hat{\bGamma} - \tilde{\bGamma})(\tilde{\bGamma}^{\top}\tilde{\bGamma})^{-1}\tilde{\bGamma}^{\top}\be_j| + | \be_i^{\top}\hat{\bGamma}( (\hat{\bGamma}^{\top}\hat{\bGamma})^{-1} - (\tilde{\bGamma}^{\top}\tilde{\bGamma})^{-1} ) \tilde{\bGamma}^{\top}\be_j| + |\be_i^{\top}\hat{\bGamma}(\hat{\bGamma}^{\top}\hat{\bGamma})^{-1} (\hat{\bGamma} - \tilde{\bGamma})^{\top}\be_j|)
    \end{align*}
    The first term can be bounded as
    \begin{align*}
        \max_{1\leq j\leq p}|\be_i^{\top} (\hat{\bGamma} - \tilde{\bGamma})(\tilde{\bGamma}^{\top}\tilde{\bGamma})^{-1}\tilde{\bGamma}^{\top}\be_j |
        &\leq \| (\hat{\bGamma} - \tilde{\bGamma})\be_i \|_{\oper} \|(\tilde{\bGamma}^{\top}\tilde{\bGamma})^{-1}\|_{\oper} \max_{1\leq j\leq p} \|\tilde{\bGamma}^{\top}\be_j\|_{2}
        \leq C'[ p^2(n\wedge p)]^{-\frac{1}{2}},
    \end{align*}
    by involving both Part (1) and (2). 
    The rest of the terms follow a similar argument as in Part (1), under the same probabilistic events therein.
    This completes the proof.
\end{proof}

\subsection{Technical lemmas}\label{app:subsec:thm:est-error-proj-bGamma-tech-lems}

\begin{lemma}[Estimation error of $\hat{\bGamma}$]\label{lem:est-error-bGamma}
    Under \Crefrange{asm:model}{asm:latent} and event $E_C$, for all $\delta>0$ and sufficiently large $n$ and $p$, there exists a absolute constant $C'>0$ such that
    \begin{align*}
        \max_{1\leq j\leq p}\|\hat{\bgamma}_{j} - \tilde{\bgamma}_{j}\|_2 &\leq C', \qquad\|\hat{\bGamma} - \tilde{\bGamma}\|_{\oper} \leq C'\sqrt{\frac{n\vee p}{n}},
    \end{align*}
    with probability at least $1-2(n+p)^{-\delta}-2(np)^{-\delta} - \exp(-n)$.
\end{lemma}
\begin{proof}[Proof of \Cref{lem:est-error-bGamma}]
    Define $\bE=\bW^*\bGamma^{*\top}$ and $\hat{\bE}=\hat{\bW}\hat{\bGamma}^{\top}$\footnote{Throughout the manuscript, the notation $\be_i$ is reserved for the unit vector and is not the $i$-th row of $\bE$. We will only use the notations of $\bE$ and $\bE_j$ to denote the matrix of latent components and its $j$-th column.}.
    Then we have $\hat{\bE} = \bW^*\bGamma^{*\top} +\bDelta$ where $\bDelta = (\hat{\bTheta}_0 - \bTheta^*) - (\bX\hat{\bF}^{\top} - \bX\bF^{*\top})$.
    By the definition of $\hat{\bE}$ we have
    \begin{align*}
        \frac{1}{np}\hat{\bE}^{\top}\hat{\bE} &= \bV\bSigma^2\bV^{\top}.
    \end{align*}
    Note that $\hat{\bGamma}=\sqrt{p}\bV\bSigma^{1/2}$, we further have
    \begin{align*}
        \frac{1}{np}\hat{\bE}^{\top}\hat{\bE}\hat{\bGamma} &= \hat{\bGamma}\bSigma^2 
    \end{align*}
    and
    \begin{align*}
        \frac{1}{n\sqrt{p}}\hat{\bE}^{\top}\hat{\bE}\bV\bSigma^{-3/2} &= \hat{\bGamma}.
    \end{align*}
    It follows that
    \begin{align*}
        \hat{\bGamma} - \tilde{\bGamma} &= \frac{1}{n\sqrt{p}}(\bE^{\top}\bDelta + \bDelta^{\top}\bE+ \bDelta^{\top}\bDelta)\bV\bSigma^{-3/2}.
    \end{align*}
    Because the operator norm is sub-multiplicative, the $\ell_2$-norm of the $j$th row of $\hat{\bGamma} - \tilde{\bGamma}$ is bounded by
    \begin{align}
        \|\hat{\bgamma}_{j} - \tilde{\bgamma}_{j}\|_2 &\leq \frac{1}{n\sqrt{p}}(
        \|\bE^{\top}\bDelta_{j}\|_2 + \|\bDelta^{\top}\bE_{j}\|_2 + \|\bDelta^{\top}\bDelta_{j}\|_2) \|\bV\|_{\oper}\|\bSigma^{-3/2}\|_{\oper} \notag\\
        &\leq \frac{1}{n\sqrt{p}}(
        \|\bE\|_{\oper}\|\bDelta_{j}\|_2 + \|\bDelta\|_{\oper}\|\bE_{j}\|_2 + \|\bDelta\|_{\oper}\|\bDelta_{j}\|_2) \|\bSigma^{-3/2}\|_{\oper} ,\label{eq:lem:est-error-bGamma-1}
    \end{align}
    and
    \begin{align}
        \|\hat{\bGamma} - \tilde{\bGamma}\|_{\oper} 
        &\leq \frac{1}{n\sqrt{p}}(
        2\|\bE\|_{\oper}\|\bDelta\|_{\oper} + \|\bDelta\|_{\oper}^2) \|\bSigma^{-3/2}\|_{\oper} .\label{eq:lem:est-error-bGamma-2}
    \end{align}
    To proceed, we split the proof into three parts.

    \paragraph{Part (1) Bounding operator noms of $\bSigma$, $\bE$, and $\bDelta$.}
    Note that $\bE=\bW^*\bGamma^{*\top}$ and $\bw_1,\ldots,\bw_n$ are mean-zero sub-Gaussian random vectors from \Cref{asm:latent}.
    From \Cref{lem:spec-D}, for any $\delta>0$, there exists $C_{\Sigma}>0$, such that
    \begin{align*}
        \|\bE\|_{\oper} \leq  2C_{\Sigma}\sqrt{ np},\qquad \frac{1}{C_{\Sigma}}\leq \lambda_r(\bSigma) &\leq \lambda_1(\bSigma) \leq C_{\Sigma},
    \end{align*}
    with probability at least $1-(n+p)^{-\delta}-(np)^{-\delta} - \exp(-n)$.
    
    Because
    \begin{align*}
        \bDelta = \bE - \hat{\bE} = \cP_{\bX}^{\perp}(\bTheta^* - \hat{\bTheta}_0) + \cP_{\bX}\bE,
    \end{align*}
    we have
    \begin{align*}
        \|\bDelta\|_{\oper} &= \|\cP_{\bX}^{\perp}(\bTheta^* - \hat{\bTheta}_0) + \cP_{\bX}\bE\|_{\oper}  \\
        &\leq \|\cP_{\bX}^{\perp}\|_{\oper}\|\bTheta^* - \hat{\bTheta}_0\|_{\oper} + \|\cP_{\bX}\bE\|_{\oper} \\
        &\leq \|\bTheta^* - \hat{\bTheta}_0\|_{\oper} + \|\cP_{\bX}\bE\|_{\oper}.
    \end{align*}
    On the one hand, from \Cref{thm:est-error-bTheta}, it follows that when $n,p$ are large enough,
    \begin{align*}
        \|\bTheta^* - \hat{\bTheta}_0\|_{\oper} \leq \sqrt{c (d+r) (n\vee p)}
    \end{align*}
    with probability at least $1-(n+p)^{-\delta}-(np)^{-\delta}$ for some constant $c>0$.
    On the other hand, notice that
    \begin{align*}
        \cP_{\bX}\bE &= \cP_{\bX}\bW^*\bGamma^{*\top} 
        = \bX\left(\frac{\bX^{\top}\bX}{n}\right)^{-1}\frac{\bX^{\top}\bW^*}{n} \bGamma^{*\top}                
    \end{align*}
    where $\bX^{\top}\bW^* = \sum_{i=1}^n\bx_i\bw_i^{*\top}$ is the sum of $n$ i.i.d. sub-exponential random matrices with zero means.
    By the matrix Bernstein's inequality, $\|\bX^{\top}\bW^*/n \|_{\oper}\lesssim \sqrt{\log (nd)/n}$.
    Thus, the second term can be bounded as 
    \[ \|\cP_{\bX}\bE\|_{\oper} \lesssim \sqrt{n} \cdot 1 \cdot \sqrt{\frac{\log (nd)}{n}} \cdot \sqrt{p}\]
    with probability at least $1-n^{-\delta}$.
    The above results suggest that
    \begin{align*}
        \|\bDelta\|_{\oper}  \leq \sqrt{c (d+r) (n\vee p)}
    \end{align*}    
    Below, we condition on the two events above, which hold with probability at least $1-2(n+p)^{-\delta}-2(np)^{-\delta} - \exp(-n)$ by union bound.

    \paragraph{Part (2) Bounding $\|\bE_j\|_2$ and $\|\bDelta_j\|_2$}
    From \Cref{lem:eigvals-W} and \Cref{asm:latent}, we have
    \begin{align*}
        \max_{1\leq j\leq p}\|\bE_j\|_2=\max_{1\leq j\leq p}\|\bW^*\bgamma_j^*\|_2\leq \|\bW\|_{\oper}\max_{1\leq j\leq p}\|\bgamma_j^*\|_2\leq 2C_{\Sigma}\sqrt{n}.
    \end{align*}
    On the other hand, because $\hat{\bW}^{\top}\hat{\bW}=n\bSigma$ and $\|\hat{\bgamma}_j\|_2\leq C^2C_{\Sigma}^{1/2}$ for all $j\in[p]$ from \Cref{lem:bound-latent}, we also have
    \begin{align*}
        \max_{1\leq j\leq p}\|\hat{\bE}_j\|_2=\max_{1\leq j\leq p}\|\hat{\bW}\hat{\bgamma}_j\|_2=\max_{1\leq j\leq p}\sqrt{n}\|\hat{\bgamma}_j\|_2\lesssim C^2C_{\Sigma}^{3/2}\sqrt{n}.
    \end{align*}
    Thus, by triangle inequality, we have 
    \begin{align*}
        \max_{1\leq j\leq p}\|\bDelta_j\|_2 &\leq \max_{1\leq j\leq p}(\|\bE_j\|_2 + \|\hat{\bE}_j\|_2) \lesssim (2+C^2)C_{\Sigma}\sqrt{n}.
    \end{align*}

    \paragraph{Part (3) Combining the previous results.}
    From \eqref{eq:lem:est-error-bGamma-1}, \eqref{eq:lem:est-error-bGamma-2} and the previous two parts, we have
    \begin{align*}
        &\max_{1\leq j\leq p}\|\hat{\bgamma}_j - \tilde{\bgamma}_j\|_2 \\
        \lesssim& \frac{1}{n\sqrt{p}C_{\Sigma}}( 2C_{\Sigma}\sqrt{ np}(2+C^2)C_{\Sigma}\sqrt{n} + \sqrt{c(d+r)(n\vee p)} C\sqrt{n} + \sqrt{c(d+r)(n\vee p)}(2C_{\Sigma}+C)\sqrt{n})\\
        \lesssim& 2(2+C^2)C_{\Sigma}
    \end{align*}
    and
    \begin{align*}
        \|\hat{\bGamma} - \tilde{\bGamma}\|_2 &\leq \frac{1}{n\sqrt{p}C_{\Sigma}}(4C_{\Sigma}\sqrt{ np}\sqrt{c(d+r)(n\vee p)} + c(d+r)(n\vee p)) \lesssim 4\sqrt{c},
    \end{align*}
    which finishes the proof.
\end{proof}

\begin{lemma}[Spectrum of $\bSigma$]\label{lem:spec-D}
    Under \Crefrange{asm:model}{asm:latent}  and event $E_C$, for all $\delta>0$ and sufficiently large $n$ and $p$, there exists a absolute constant $C_{\Sigma}>1$ such that
    \begin{align*}
        \frac{1}{C_{\Sigma}} \leq \lambda_r(\bSigma) &\leq \lambda_1(\bSigma) \leq C_{\Sigma},
    \end{align*}
    with probability at least $1-(n+p)^{-\delta}-(np)^{-\delta}-\exp(-n)$.
\end{lemma}
\begin{proof}[Proof of \Cref{lem:spec-D}]
    By Weyl's inequality, we have that for all $k\in[r]$,
    \begin{align}
        \left|\lambda_k(\bSigma) - \frac{1}{\sqrt{np}}\lambda_{k}(\bW^*\bGamma^{*\top})\right| & = \frac{1}{\sqrt{np}} \left|\lambda_k(\hat{\bW}\hat{\bGamma}^{\top}) - \lambda_{k}(\bW^*\bGamma^{*\top})\right| \notag\\
        &\leq \frac{1}{\sqrt{np}} \|\hat{\bW}\hat{\bGamma}^{\top} - \bW^*\bGamma^{*\top} \|_{\oper} \notag\\
        &= \frac{1}{\sqrt{np}} \|\cP_{\bX}^{\perp} (\hat{\bTheta}_0 - \bTheta^*) \|_{\oper}\notag\\
        &\leq \frac{1}{\sqrt{np}} \| \hat{\bTheta}_0 - \bTheta^*\|_{\oper}. \label{eq:lem:spec-D-eq-1}
    \end{align}
    We next bound $\lambda_{k}(\bW^*\bGamma^{*\top})$ and $\|\hat{\bTheta}_0 - \bTheta^*\|_{\oper}$ separately.
    Applying \Cref{lem:eigvals-W} under \Cref{asm:latent} yields that
    \begin{align*}
        \frac{1}{C_0'} \leq \lambda_r\left(\frac{1}{n}\bW^{*\top}\bW^*\right) \leq \lambda_1\left(\frac{1}{n}\bW^{*\top}\bW^*\right) \leq C_0'.
    \end{align*}
    with probablity at least $1-\exp(-n)$ for some constant $C_0'>1$.
    From \Cref{asm:latent} we further have
    \begin{align}
        \sqrt{\frac{1}{CC_0'}} \leq \frac{1}{\sqrt{np}}\lambda_{r}(\bW^*\bGamma^{*\top}) \leq \frac{1}{\sqrt{np}}\lambda_{1}(\bW^*\bGamma^{*\top}) \leq  \sqrt{CC_0'}. \label{eq:lem:spec-D-eq-2}
    \end{align}
    On the other hand, from \Cref{thm:est-error-bTheta} we have for all $\delta>0$, there exists $C>0$ such that
    \begin{align}
        \frac{1}{\sqrt{np}} \| \hat{\bTheta}_0 - \bTheta^*\|_{\oper} &\leq \frac{1}{\sqrt{np}} \| \hat{\bTheta}_0 - \bTheta^*\|_{\fro}\leq C\sqrt{\frac{r (n\vee p) }{np}} =: C_{n,p},\label{eq:lem:spec-D-eq-3}
    \end{align}
    with probability at least $1-(n+p)^{-\delta}-(np)^{-\delta}$.

    Condition on the above two events, applying triangle inequality on \eqref{eq:lem:spec-D-eq-1} and combining \eqref{eq:lem:spec-D-eq-2} and \eqref{eq:lem:spec-D-eq-3}, we have
    \begin{align*}
         \sqrt{\frac{1 }{CC_0'}} - C_{n,p}  \leq \lambda_r(\bSigma) &\leq \lambda_1(\bSigma) \leq \sqrt{CC_0'}  + C_{n,p}.
    \end{align*}
    Note that $C_{n,p}=o(1)$ as both $n$ and $p$ tend to infinity.
    When $n$ and $p$ is such that $C_{n,p}<1/2\sqrt{CC_0'}$, setting $C_{\Sigma}=3\sqrt{CC_0'}/2$ gives the desired bound.
    By union bound, this holds with probability at least $1-(n+p)^{-\delta}-(np)^{-\delta}-\exp(-n)$, which finishes the proof.
\end{proof}

\begin{lemma}[Spectrum of $\bW^*$]\label{lem:eigvals-W}
    Under \Cref{asm:latent}, for sufficiently large $n$, there exists a absolute constant $C_0'>0$ such that
    \begin{align*}
        \frac{1}{C_0'} \leq \lambda_r\left(\frac{1}{n}\bW^{*\top}\bW^*\right) \leq \lambda_1\left(\frac{1}{n}\bW^{*\top}\bW^*\right) \leq C_0',
    \end{align*}
    with probability at least $1-\exp(-n)$.
\end{lemma}
\begin{proof}[Proof of \Cref{lem:eigvals-W}]
    Note that $\bW^*\in\RR^{n\times r}$ is a random matrix whose rows $\bw_1^*,\ldots,\bw_n^*$ are i.i.d. sub-Gaussian random vectors.
    From the concentration inequality of the operator norm of the random matrices \citep[Theorem 4.6.1, Exercise 4.7.3]{vershynin2018high}, for any $u>0$ we have
    \begin{align*}
        \left\|\frac{1}{n}\bW^{*\top}\bW^* - \bSigma_w\right\|_{\oper} &\leq CK^2\left(\sqrt{\frac{r+u}{n}} + \frac{r+u}{n}\right) \|\bSigma_w\|_{\oper},
    \end{align*}
    with probability at least $1-2\exp(-u)$, where $C>0$ is some absolute constant and $K=\max_{i\in[n]}\|\bw_i^*\|_{\psi_2}$.
    When $n$ is sufficiently large such that $2CK^2n^{-1/4}< 1$, setting $u=n^{1/2}$ and $C'=2CK^2n^{-1/4}C_0$ yields that
    \begin{align*}
        \left\|\frac{1}{n}\bW^{*\top}\bW^* - \bSigma_w\right\|_{\oper} &\leq CK^2\left( n^{-\frac{1}{4}} + n^{-\frac{1}{2}}\right)  \|\bSigma_w\|_{\oper} \leq  C' < C_0.
    \end{align*}
    By Weyl's inequality, we have that,
    \begin{align*}
        \max_{k\in[r]}\left|\lambda_k\left(\frac{1}{n}\bW^{*\top}\bW^*\right) - \lambda_k\left(\bSigma_w\right)\right| &\leq \left\|\frac{1}{n}\bW^{*\top}\bW^* - \bSigma_w\right\|_{\oper} \leq  C'.
    \end{align*}
    By triangle inequality and the boundedness of $\bSigma_w$'s spectrum from \Cref{asm:latent}, it follows that
    \begin{align*}
        \frac{1}{C_0}-C' \leq \min_{k\in[r]}\lambda_k\left(\frac{1}{n}\bW^{*\top}\bW^*\right) \leq \max_{k\in[r]}\lambda_k\left(\frac{1}{n}\bW^{*\top}\bW^*\right) \leq C_0 + C',
    \end{align*}
    with probability at least $1-\exp(-n)$.
    Setting $C_0'= \min\{C_0+C', (C_0^{-1}-C')^{-1}\}$ finishes the proof.
\end{proof}

\begin{lemma}[Boundedness of latent factors and loadings]\label{lem:bound-latent}
     Under \Crefrange{asm:model}{asm:latent} and event $E_C$, for all $\delta>0$ and sufficiently large $n$ and $p$, there exists a absolute constant $C_{\Sigma}>1$ such that $\|\hat{\bgamma}\|_2 \leq C^2C_{\Sigma}^{1/2}$, with probability at least $1-(n+p)^{-\delta}-(np)^{-\delta}-\exp(-n)$.
\end{lemma}
\begin{proof}[Proof of \Cref{lem:bound-latent}]
    Recall that $\hat{\bW}_0$ and $\hat{\bGamma}_0$ are the solutions to the alternative maximization problems which satisfy that $\|\hat{\bw}_{0,i}\|_{2}\leq C$ and $\|\hat{\bgamma}_{0,j}\|_{2}\leq C$ for $i=1,\ldots,n$ and $1,\ldots,p$.
    Let $\hat{\bW}_0\hat{\bGamma}_0^{\top}=\sqrt{np}\bU\bSigma\bV^{\top}$ be the condensed SVD.
    Then $\hat{\bGamma}$ is defined to be $\sqrt{p}\bV\bSigma^{1/2}= \hat{\bGamma}_0\hat{\bW}_0^{\top} \bU \bSigma^{1/2} / \sqrt{n} $.
    From \Cref{lem:spec-D}, with probability at least $1-(n+p)^{-\delta}-(np)^{-\delta}-\exp(-n)$, it holds that
    \begin{align*}
        \|\hat{\bgamma}_j\|_2 &\leq  \|\hat{\bW}_0/\sqrt{n}\|_{\oper}\|\bU\|_{\oper}\|\bSigma^{1/2}\|_{\oper}\|\hat{\bgamma}_{0,j}\|_2\\
        &\leq \max_{1\leq i\leq n}\|\hat{\bW}_{0,i}\|_{2} \cdot1\cdot C_{\Sigma}^{1/2}\cdot C\\
        &\leq C_{\Sigma}^{1/2}C^2,
    \end{align*}
    which finishes the proof.    
\end{proof}

\section{Estimation of latent factors and direct effects}\label{app:est-direct}

\subsection{Preparatory definitions}\label{app:est-direct-pre}
Towards proving \Cref{thm:est-err-B}, we first introduce the following notations.
Recall that optimization problem \eqref{opt:3} is the multivariate lasso with nuisance parameter.
Define the response vector $\tilde{\by}=\vec(\bY)\in\RR^{np}$, the design matrices $\tilde{\bX} = (\bI_p \otimes \bX) \in\RR^{np\times pd}$, $\bbeta=\vec(\bB^*)$, $\bzeta=\vec(\bZ^*\bGamma^{*\top})$, and the projection matrix $\tilde{\cP}_{\bGamma^*} = (\cP_{\bGamma^*}\otimes\bI_d) \in\RR^{pd\times pd}$.
Here, symbol `$\otimes$' denotes the Kronecker product.
With slight abuse of notations, we use $\cL(\bbeta,\bzeta)$ to denote the unregularized loss function of \eqref{opt:3}.
Let $\cF=\{(\bbeta,\bzeta)\in\RR^{pd}\times\RR^{np}\mid \bbeta=\vec(\bB), \bzeta=\vec(\bZ{\bGamma}^{\top})\text{ for }\bB\in\RR^{p\times d},\bZ\in\RR^{n\times r},\bGamma\in\RR^{p\times r} \text{ such that } \cP_{{\bGamma}}\bB=\zero,\tilde{\bX}\bbeta+\bzeta\in\cR_C^{np}\}$ be the feasible set of $\bbeta$ and $\bzeta$.
Then the joint optimization problem \eqref{opt:3} is equivalent to
\begin{align}
    \hat{\bbeta},\hat{\bzeta}&\in \argmin_{(\bbeta,\bzeta)\in\cF} \cL(\bbeta,\bzeta) + \lambda\|\bbeta\|_1.
\end{align}
Let $(\tilde{\bbeta}^*,\tilde{\bzeta}^*)=(\vec(\cP_{\bGamma^*}^{\perp}\bB^*),\vec(\bX\bB^*\cP_{\bGamma^*}+\bZ^*\bGamma^{*\top}))$, $({\bbeta}^*,{\bzeta}^*)=(\vec(\bB^*),\vec(\bZ^*\bGamma^{*\top}))$ denote the tuples of target coefficients and note that $(\tilde{\bbeta}^*,\tilde{\bzeta}^*)\in\cF$.

The organization of the following subsections is summarized as below:
\begin{itemize}
    \item \Cref{app:subsec:cor:est-confound-col} proves \Cref{cor:est-confound-col}, which controls the column-wise $\ell_2$-norm of the estimation error of the latent component.

    \item \Cref{app:subsec:thm:est-err-B} proves \Cref{thm:est-err-B}.
    The proof of \Cref{thm:est-err-B} consists of three main steps:
    \begin{enumerate}[(1)]

        \item Establish cone condition: 
        We involve \Cref{lem:optimality} to show that the estimation from the sequential optimization problems \eqref{opt:1}-\eqref{opt:2} obtain approximately optimality condition to the joint optimization problem \eqref{opt:3}:
        \[
        \cL(\hat{\bB},\hat{\bGamma},\hat{\bZ}) + \lambda\|\hat{\bB}\|_{1,1} \leq \cL(\bB^*,{\bGamma}^*,{\bZ}^*) + \lambda\|{\bB}^*\|_{1,1} + \tau_{n,p},\]
        for some small order term $\tau_{n,p}$ with high probability.
        This enables us to derive the cone condition.

        \item Obtain upper and lower bound of the first-order approximation error: We involve \Cref{lem:est-err-B-grad-inf-norm} to derive the upper bound, and \Cref{lem:rsc} to establish the locally strong convexity and hence the lower bound.

        \item Derive the estimation errors: We compute the $\ell_2$-norm and $\ell_1$-norm estimation error based on the previous two steps.
    \end{enumerate}

    \item \Cref{app:subsec:est-latent-direct-tech-lems} gathers helper lemmas used in the current section.
\end{itemize}

\subsection{Proof of \Cref{cor:est-confound-col}}\label{app:subsec:cor:est-confound-col}
    \begin{proof}[Proof of \Cref{cor:est-confound-col}]
        Define $\hat{\bE}=\hat{\bZ}\hat{\bGamma}^{\top}$ and $\bE^*=\bZ^*\bGamma^{*\top}$.
        Let $\hat{\bB} = \argmin_{\{\bB\in\RR^{p\times d}\mid \cP_{\hat{\bGamma}}\bB=\zero\}}\cL(\bX\bB^{\top} + \hat{\bE})$ and $\hat{\bTheta} = \bX\hat{\bB}^{\top} + \hat{\bE}$.
        Since $\cL(\hat{\bTheta}) \leq \cL(\bTheta^*)$ as assumed in \Cref{cor:est-confound-col}, from \Cref{thm:est-error-bTheta} we have $\|\hat{\bTheta}-\bTheta^*\| \lesssim \Op(\sqrt{n\vee p})$.
        From \Cref{thm:est-error-proj-bGamma}, we further have
        \[\|\hat{\bE}-\bE^*\|_{\fro}\leq \|(\hat{\bTheta}-\bTheta^*)\cP_{\hat{\bGamma}}\|_{\fro} + \|\bTheta^*(\cP_{\hat{\bGamma}}-\cP_{\bGamma^*})\|_{\fro}\lesssim \Op(\sqrt{n\vee p}) .\]
        Then from \citet[Proposition 5.2]{lin2021exponential} we have that, up to sign,
        \[\frac{1}{n}\|\hat{\bZ}-\bZ^*\bR^{\top}\|_{\fro}^2 \lesssim \frac{r^{4k_2-k_1+4}}{n\wedge p } =:\eta_n, \]
        with probability at least $1-n^{-c}- p^{-c}$.
    
        Recall the invertible matrix $\bR$ with $\|\bR\|_{\oper}=\Op(1)$ from \Cref{asm:iden-latent-eigs} and define the transformed parameters $\tilde{\bZ}^*=\bZ^*\bR^{\top}$ and $\tilde{\bGamma}^*=\bGamma^*\bR^{-1}$.
        Because $\bE_j - \bE_j^* = \tilde{\bZ}^*(\bgamma_j - \tilde{\bgamma}_j^*) + (\bZ - \tilde{\bZ}^*)\bgamma_j $, with probability tending to one, it follows that
        \begin{align*}
            \max_{1\leq j\leq p} \frac{1}{\sqrt{n}}\|\hat{\bE}_j - \bE_j^*\|_2 &= \max_{1\leq j\leq p} \frac{1}{\sqrt{n}}\|\tilde{\bZ}^*(\hat{\bgamma}_j - \tilde{\bgamma}_j^*) + (\hat{\bZ} - \tilde{\bZ}^*)\hat{\bgamma}_j \|_2\\
            &\leq \max_{1\leq j\leq p} \frac{1}{\sqrt{n}}\|\tilde{\bZ}^*(\hat{\bgamma}_j - \tilde{\bgamma}_j^*) \|_2 + \max_{1\leq j\leq p} \frac{1}{\sqrt{n}}\|(\hat{\bZ} - \tilde{\bZ}^*)\hat{\bgamma}_j \|_2\\
            &\leq  \max_{1\leq j\leq p,1\leq i\leq n} \frac{1}{\sqrt{n}}\|\tilde{\bz}_i^*\|_{\infty}\|\hat{\bgamma}_j - \tilde{\bgamma}_j^* \|_1 + \max_{1\leq j\leq p} \frac{1}{\sqrt{n}}\|\hat{\bZ} - \tilde{\bZ}^*\|_{\oper}\|\hat{\bgamma}_j \|_2\\
            &\lesssim \frac{\log n}{\sqrt{n}} + \sqrt{\eta_n} \\
            &\lesssim\sqrt{n^{-1}\log n \vee \eta_n} 
        \end{align*}
        where the second last inequality is because that $\|\tilde{\bgamma}_j^*\|_2\leq C$ from \Cref{asm:latent}, $\|\hat{\bgamma}_j - \tilde{\bgamma}_j^*\|_1\leq \sqrt{r} \|\hat{\bgamma}_j - \tilde{\bgamma}_j^*\|_2\leq 2\sqrt{r}C$ from \Cref{lem:est-error-bGamma}, and $\tilde{\bz}^{*}_i$'s are independent $r$-dimensional sub-Gaussian random vectors from \Cref{lem:subgau-Z} so that $\max_{1\leq i\leq n}\|\tilde{\bz}^{*}_i\|_{\infty}$ scales in $\log{(nr)}$.
    \end{proof}

\subsection{Proof of \Cref{thm:est-err-B}}\label{app:subsec:thm:est-err-B}

\begin{proof}[Proof of \Cref{thm:est-err-B}]    
    Define $\bDelta_{\beta}=\hat{\bbeta} - \bbeta^*$, $\bDelta_{\zeta}=\hat{\bzeta} - \bzeta^*$, and $\cS=\supp(\bbeta^*)$.
    Let $\bg_{\beta} = \nabla_{\bbeta}\cL(\hbbeta,\hat{\bzeta})$, and $\bg_{\zeta} = \nabla_{\bzeta}\cL(\hbbeta,\hat{\bzeta}) $ and analogously define $\bg_{\beta}^*,\bg_{\zeta}^*$. 
    \paragraph{(1) Cone condition.} From \Cref{lem:optimality}, we have the optimality condition
    \begin{align*}
        \cL(\hat{\bbeta},\hat{\bzeta}) + \lambda\|\hat{\bbeta}\|_1 &\leq \cL({\bbeta}^*,{\bzeta}^*) + \lambda\|{\bbeta}^*\|_1 + \tau_{n,p},
    \end{align*}
    with $\tau_{n,p}$ defined in \Cref{lem:optimality}.
    Rearranging the above display, it follows that
    \begin{align}
        \lambda\|\hat{\bbeta}\|_1 &\leq \cL(\bbeta^*,\bzeta^*) - \cL(\hat{\bbeta},\hat{\bzeta}) + \lambda\|\bbeta^*\|_1 + \tau_{n,p}\notag\\
        &\leq \bDelta_{\beta}^{\top}\bg_{\beta}^* + \bDelta_{\zeta}^{\top}\bg_{\zeta}^* + \lambda\|\bbeta^*\|_1 + \tau_{n,p}\notag\\
        &\leq \|\bDelta_{\beta}\|_{1}\|\bg_{\beta}^*\|_{\infty} + \bDelta_{\zeta}^{\top}\bg_{\zeta}^* + \lambda\|\bbeta^*\|_1 + \tau_{n,p}, \label{eq:thm:est-err-B-eq-1}
    \end{align}
    where the second inequality is from the convexity of $\cL$, and the last is from Holder's inequality.
    The term involves nuisance parameters can be further bounded as
    \begin{align}
        \bDelta_{\zeta}^{\top}\bg_{\zeta}^* 
        &= \frac{1}{n}\sum_{\ell=1}^{np} [\tilde{y}_{\ell}-A'(\tilde{\bx}_{\ell}^{\top}\bbeta^*+\zeta_{\ell}^*)]\delta_{\zeta,\ell} 
        \leq 2C \|\bg_{\zeta}^*\|_{\infty} \leq  \frac{4cC^2\alpha \log(2np)}{n}, \label{eq:thm:est-err-B-eq-2}
    \end{align}
    where the last inequality is from \Cref{lem:est-err-B-grad-inf-norm}~\eqref{lem:est-err-B-grad-inf-norm-item-2}, which holds with probability at least at least $1-(n+p)^{-c}-(np)^{-c}-\exp(-n)$ for some $c>0$.
    On the other hand, the left hand side of \eqref{eq:thm:est-err-B-eq-1} is lower bounded as
    \begin{align}
        \|\hat{\bbeta}\|_1 &= \|\bbeta^* + \bDelta_{\beta}\|_1 \notag \\
        & = \|\bbeta_{\cS}^*+\bDelta_{\beta,\cS}\|_1 + \|\bDelta_{\beta,\cS^c}\|_1 \notag \\
        & \geq \|\bbeta_{\cS}^*\|_1 - \|\bDelta_{\beta,\cS}\|_1 + \|\bDelta_{\beta,\cS^c}\|_1,\label{eq:thm:est-err-B-eq-3}
    \end{align}
    where $\cS=\{j \in [pd] \mid \beta_j^* \neq 0\}$ is the active set of the true coefficients.
    Combining \eqref{eq:thm:est-err-B-eq-1}, \eqref{eq:thm:est-err-B-eq-2}, and \eqref{eq:thm:est-err-B-eq-3} yields that
    \begin{align*}
        (\lambda - \|\bg_{\beta}^*\|_{\infty}) \|\bDelta_{\beta,\cS^c}\|_1 \leq (\lambda + \|\bg_{\beta}^*\|_{\infty})\|\bDelta_{\beta,\cS}\|_1 + \left(\frac{4cC^2\alpha \log(2np)}{n} + \tau_{n,p}\right).
    \end{align*}
    Note that under the same probabilistic event above, from \Cref{lem:est-err-B-grad-inf-norm}~\eqref{lem:est-err-B-grad-inf-norm-item-1}, we have 
    \begin{align}
        \|\bg_{\beta}^*\|_{\infty} \leq 4\nu^2\sqrt{c\log^2(2nd)/n} . \label{eq:thm:est-err-B-eq-4}
    \end{align}
    When $\lambda^*\asymp 8\nu^2\sqrt{c\log^2(2nd)/n}\geq 2\|\bg_{\beta}^*\|_{\infty}$, this implies the approximate cone condition $\bDelta_{\beta} \in \cC(3, \cS)$, where
    \begin{align}
        \cC(\xi, \cS) &:= \left\{\bDelta_{\beta}\in\RR^{pd}\,\Bigg|\, \|\bDelta_{\beta,\cS^c}\|_1 \leq \xi \|\bDelta_{\beta,\cS}\|_1 +  \tau_{n,p}^*\right\},\label{eq:thm:est-err-B-eq-cone}
    \end{align}
    and
    \begin{align*}
        \tau_{n,p}^* &=\frac{C^2\alpha}{\nu^2}\sqrt{\frac{c\log^2(2np)}{n \log^2(2nd)}}+{\frac{\sqrt{n}}{(n\wedge p)\log(2nd)}} +  \sqrt{\frac{(sd)^2}{n\wedge p^{1-k}}}.
    \end{align*}
    From the cone condition \eqref{eq:thm:est-err-B-eq-cone}, the $\ell_1$-norm bound follows by observing that
    \begin{align}
        \|\bDelta_{\beta}\|_1 &\leq \|\bDelta_{\beta,\cS}\|_1 + \|\bDelta_{\beta,\cS^c}\|_1 \notag\\
        &\leq 4\|\bDelta_{\beta,\cS}\|_1 + \tau_{n,p}^* \notag\\
        &\leq 4\sqrt{sd}\|\bDelta_{\beta,\cS}\|_2 +   \tau_{n,p}^* \notag\\
        &\leq 4\sqrt{sd}\|\bDelta_{\beta}\|_2 +\tau_{n,p}^*.\label{eq:thm:est-err-B-eq-5}
    \end{align}

    \paragraph{(2) Upper and lower bound of the first-order approximation error.}
    To quantify the estimation of the coefficient, we next analyze the first-order approximation error of the normalized likelihood: $\cE(\hbbeta, \bbeta^*;\hat{\bzeta},\bzeta^*) = \bDelta_{\beta}^{\top}(\bg_{\beta} - \bg_{\beta}^*)$.
    By the first-order optimality condition of convex optimization problem \eqref{opt:4}, we have
    \begin{align}
        \cE(\hbbeta, \bbeta^*;\hat{\bzeta},\bzeta^*) &\leq -\bDelta_{\beta}^{\top}\bg_{\beta}^*\notag\\
        &\leq  \|\bDelta_{\beta}\|_{1}  \|\bg_{\beta}^*\|_{\infty}\notag \\
        &\leq \|\bDelta_{\beta}\|_2 
        16\nu^2\sqrt{\frac{csd\log^2(2nd)}{n}} + 4C^2\alpha\sqrt{c}\frac{\log(2np)}{n}, \label{eq:thm:est-err-B-eq-upper}
    \end{align}
    where the second inequality is from Holder's inequality, and the last inequality is from \eqref{eq:thm:est-err-B-eq-4} and \eqref{eq:thm:est-err-B-eq-5}.
    This establishes the upper bound for $\cE(\hbbeta, \bbeta^*;\hat{\bzeta},\bzeta^*)$.

    On the other hand, \Cref{lem:rsc} implies that $\cE(\hbbeta, \bbeta^*;\hat{\bzeta},\bzeta^*)$ is locally restricted strongly convex over an augmented cone $\cC_a(3,\cS,\eta_n)$, defined in \eqref{eq:lem:rsc:eq-aug-cone}:
    \begin{align}
        \inf_{(\bDelta_\beta,\bDelta_\zeta)\in \cC_a(3,\cS,\eta_n)}\cE(\bbeta^* + \bDelta_\beta, \bbeta^*;\bzeta^* + \bDelta_\zeta ,\bzeta^*) &\geq\frac{\kappa_1C'}{2} \|\bDelta_{\beta}\|_2^2-C''\sqrt{\frac{\log n}{n}\vee \eta_n} \|\bDelta_{\beta}\|_1,  \label{eq:thm:est-err-B-eq-lower}
    \end{align}
    for some constant $C'>0$, with probability at least $1 -2 \exp(-3 \log n)$.

    \paragraph{(3) Estimation error.}
    From Part (2), \eqref{eq:thm:est-err-B-eq-upper} and \eqref{eq:thm:est-err-B-eq-lower} imply that
    \begin{align*}
        \frac{\kappa_1C'}{2} \|\bDelta_{\beta}\|_2^2 &\leq \|\bDelta_{\beta}\|_2 
        \left(16\nu^2\sqrt{\frac{csd\log^2(2nd)}{n}} +4C''\sqrt{\frac{sd}{n\wedge p}}\right) \\
        &\qquad + 4C^2\alpha\sqrt{c}\frac{\log(2np)}{n} +\frac{C''r^{4k_2-k_1+4}}{\sqrt{n\wedge p}}\tau_{n,p}^*.
    \end{align*}
    over $\cC_a(3,\cS,\eta_n)$.
    This implies that, with probability at least $1-(n+p)^{-c'}-(np)^{-c'}-\exp(-n)$ for some $c'>0$,
    \begin{align}
        \|\bDelta_{\beta}\|_2 &\leq \frac{2}{\kappa_1C'}\left(16\nu^2\sqrt{\frac{csd\log^2(2nd)}{n}} +C''\sqrt{\frac{sd}{n\wedge p}}\right) \notag\\
        &\qquad 
        + \sqrt{\frac{2}{\kappa_1C'}} \sqrt{4C^2\alpha\sqrt{c}\frac{\log(2np)}{n} + \frac{C''r^{4k_2-k_1+4}\tau_{n,p}^*}{\sqrt{n\wedge p}}} \notag\\
        &\lesssim \sqrt{\frac{(sd\log^2(nd))\vee\log(np)}{n}} + \frac{n^{1/4}}{(n\wedge p)^{3/2}\log^{1/2}(nd)}  + \sqrt{\frac{sd}{n\wedge p^{1-k}}}.\label{eq:thm:est-err-B-eq-l2-bound}
    \end{align}
    To establish the $\ell_1$-norm bound, from \eqref{eq:thm:est-err-B-eq-5} and \eqref{eq:thm:est-err-B-eq-l2-bound}, we have
    \begin{align}
        \|\bDelta_{\beta}\|_1& \leq 4\sqrt{sd}\|\bDelta_{\beta}\|_2 + \tau_{n,p}^* \notag\\
        &\lesssim \sqrt{sd\frac{(sd\log^2(nd))\vee\log(np)}{n}}  +\sqrt{\frac{(sd)^2}{n\wedge  p^{1-k}}} + \frac{\sqrt{n}}{(n\wedge p)\log(nd)}, \label{eq:err-Delta-B}
    \end{align}
    which completes the proof.
\end{proof}

\subsection{Technical lemmas}\label{app:subsec:est-latent-direct-tech-lems}

\begin{lemma}[Sequential and joint optimization]\label{lem:optimality}
    Under the same conditions in \Cref{thm:est-err-B}, for any constant $\delta>0$ it holds with probability at least $1-(n+p)^{-\delta}-(np)^{-\delta}-\exp(-n)$ that 
    \begin{align*}
        \cL(\hat{\bbeta},\hat{\bzeta}) + \lambda\|\hat{\bbeta}\|_1 &\leq \cL({\bbeta}^*,{\bzeta}^*) + \lambda\|{\bbeta}^*\|_1 + \tau_{n,p},
    \end{align*}
    where \begin{align*}
        \tau_{n,p} = \cO\left( {\frac{1}{n\wedge p}} + \lambda \sqrt{\frac{(sd)^2}{n\wedge p^{1-k}}}\right).
    \end{align*}
\end{lemma}
\begin{proof}[Proof of \Cref{lem:optimality}]
    From \Cref{asm:model}, the entry of $\bB^*$ is bounded because $\|\bB^*\|_{\max}=\max_{1\leq i\leq p,1\leq j\leq d}|b_{ij}|\leq \max_{1\leq i\leq p}\|\bb_{i}\|_{\infty}\leq \max_{1\leq i\leq p}\|\bb_{i}\|_{2}\leq C$.
    From \Cref{prop:iden-B} and \Cref{asm:model}, we further have that 
    \begin{align}
        \|\cP_{\bGamma^*}\bB^*\|_{\fro} \lesssim \sqrt{sd/p}\|\bB^*\|_{\max} \lesssim \sqrt{sd/p}. \label{eq:PGamma-B}
    \end{align}
    Thus, from the assumption of \Cref{thm:est-err-B}, we have that
    \begin{align}
        \|\cP_{\bGamma^*}\bB^*\|_{1,1}&= \cO(p^{k/2}\|\cP_{\bGamma^*}\bB^*\|_{\fro})\lesssim \sqrt{sd} p^{(k-1)/2}.\label{eq:PGamma-B-l1-norm}
    \end{align}
    This ensures that $\|\tilde{\bbeta}^*\|_{1,1}$ is close to $\|{\bbeta}^*\|_{1,1}$.

    From optimality condition of optimization problem \eqref{opt:2}, we have
    \begin{align*}
        \cL(\hat{\bbeta},\hat{\bzeta}) + \lambda\|\hat{\bbeta}\|_1 &\leq \cL(\tilde{\bbeta}',\tilde{\bzeta}') + \lambda\|\tilde{\bbeta}'\|_1,
    \end{align*}
    where $\tilde{\bbeta}'= \vec(\cP_{\hat{\bGamma}}^{\perp}\bB^{*})$ and $\tilde{\bzeta}' = \vec(\bTheta^*\cP_{\hat{\bGamma}})$.
    It follows that
    \begin{align}
        \cL(\hat{\bbeta},\hat{\bzeta}) + \lambda\|\hat{\bbeta}\|_1 &\leq \cL(\bTheta^* - \bZ^*\bGamma^{*\top}\cP_{\hat{\bGamma}}^{\perp}) + \lambda\|\bB^{*\top}\cP_{\hat{\bGamma}}^{\perp}\|_{1,1}.\label{eq:lem:optimality-1}
    \end{align}

    We next bound the first term in \eqref{eq:lem:optimality-1}. From the proof of \Cref{lem:upper-likelihood-diff}, we have
    \begin{align*}
        \cL(\bTheta^* - \bZ^*\bGamma^{*\top}\cP_{\hat{\bGamma}}^{\perp}) - \cL(\bTheta^*) &\leq \underbrace{\frac{1}{n}\tr((\bY-A'(\bTheta^*))^{\top}\bZ^*\bGamma^{*\top}\cP_{\hat{\bGamma}}^{\perp})}_{T_1} + \underbrace{\frac{\kappa_2}{2n} \|\bZ^*\bGamma^{*\top}\cP_{\hat{\bGamma}}^{\perp}\|_{\fro}^2}_{T_2}.
    \end{align*}
    Let $\bA = n^{-1}\bZ^*\bGamma^{*\top}\cP_{\hat{\bGamma}}^{\perp}$.
    From \Cref{lem:proj-B-PGamma,lem:subgau-Z}, the second term $T_2$ can be upper bounded as
    \begin{align*}
        T_2 = \frac{\kappa_2}{2}n\|\bA\|_{\fro}^2 = \frac{\kappa_2}{2n}\|\bZ^*\bGamma^{*\top}\cP_{\hat{\bGamma}}^{\perp}\|_{\fro}^2 &\leq \frac{\kappa_2}{2n}\|\bZ^*\|_{\oper}^2\|\bGamma^{*\top}\cP_{\hat{\bGamma}}^{\perp}\|_{\fro}^2
        = \cO\left({\frac{r}{n\wedge p}}\right),
    \end{align*}
    with probability at least $1-(n+p)^{-\delta}-(np)^{-\delta}-\exp(-n)$.
    In the equality, we use the fact that $\bZ^*$ is a matrix with independent sub-Gaussian rows to obtain similar results as \Cref{lem:spec-D}.

    For the first term $T_1$, note that $y_{ij}-A'(\theta^*_{ij})$ is mean-zero $(\nu,\alpha)$-sub-exponential random variable when conditioned on $(\bX^*,\bZ^*)$, where $\nu = \sqrt{\kappa_2}$ and $\alpha = 1/C^2$, as shown in the proof of \Cref{thm:est-error-bTheta}.
    To bound the first term $T_1$, we apply Bernstein’s inequality \citep[Theorem 2.8.2]{vershynin2018high} to obtain
    \begin{align*}
        \PP(|T_1| \geq t \mid \bX^*,\bZ^*) &\leq 2 \exp\left(-\min\left\{\frac{t^2}{2\nu^2 \|\bA\|_{\fro}^2}, 
        \frac{t}{2\alpha \|\bA\|_{\max}}\right\}\right)
    \end{align*}
    From the proof above, we have that $\|\bA\|_{\max} \leq \|\bA\|_{\fro} =\Op( (nr(n\wedge p))^{-1/2})$.
    Therefore, by choosing $t \asymp ({{nr(n\wedge p)}})^{-1/2} \log(np)$, we further have
    \begin{align*}
        |T_1| =\cO\left( \sqrt{\frac{r}{n(n\wedge p)}}\right).
    \end{align*}
    with probability at least $1-(np)^{-\delta}$.
    The above results imply that
    \begin{align}
        \cL(\bTheta^* - \bZ^*\bGamma^{*\top}\cP_{\hat{\bGamma}}^{\perp})  &\leq  \cL(\bTheta^*) + \cO\left({\frac{r}{n\wedge p}}\right),\label{eq:lem:optimality-2-0}
    \end{align}
    with probability at least $1-(n+p)^{-\delta}-(np)^{-\delta}-\exp(-n)$.
    
    Consider the second term of \eqref{eq:lem:optimality-1}, from \Cref{thm:est-error-proj-bGamma} we also have
    \begin{align}
        \|\bB^{*\top}\cP_{\hat{\bGamma}}^{\perp}\|_{1,1} - \|\bB^{*\top}\cP_{{\bGamma}}^{*\perp}\|_{1,1}  & \leq  \|\bB^{*\top}(\cP_{\hat{\bGamma}}^{\perp} - \cP_{{\bGamma}}^{*\perp})\|_{1,1} \notag\\
        &\leq \sum_{\ell=1}^d \|(\cP_{\hat{\bGamma}}^{\perp} - \cP_{{\bGamma}}^{*\perp})\bB_{\ell}^*\|_{1} \notag\\
        &\leq ds C\max_{j\in[d]}\|(\cP_{\hat{\bGamma}}^{\perp} - \cP_{{\bGamma}}^{*\perp})\be_{j}\|_{1}\notag\\
        &=\cO\left(\frac{ds}{\sqrt{p(n\wedge p)}}\right),\label{eq:lem:optimality-2-1}
    \end{align}
    under the same probability event above.

    Finally, combining \eqref{eq:lem:optimality-1}-\eqref{eq:lem:optimality-2-1}, we have
    \begin{align*}
        \cL(\hat{\bbeta},\hat{\bzeta}) + \lambda\|\hat{\bbeta}\|_1 &\leq\cL(\bTheta^*) + \lambda\|\bB^{*\top}\cP_{\hat{\bGamma}}^{\perp}\|_{1,1} + \Op\left( {\frac{1}{n\wedge p}} \right), \\
        &= \cL(\tilde{\bbeta}^*,\tilde{\bzeta}^*) + \lambda\|\tilde{\bbeta}^*\|_1 + \Op\left( {\frac{1}{n\wedge p}} \right)\\
        &= \cL({\bbeta}^*,{\bzeta}^*) + \lambda\|{\bbeta}^*\|_1 +  \Op\left( {\frac{1}{n\wedge p}} + \lambda \sqrt{\frac{(sd)^2}{n\wedge p^{1-k}}}\right),
    \end{align*}
    where the last inequality is from \eqref{eq:PGamma-B-l1-norm}.
    This completes the proof.
\end{proof}

\begin{lemma}[Infinity norm of the gradient]\label{lem:est-err-B-grad-inf-norm}
    Under the same conditions in \Cref{thm:est-err-B}, for any constant $c>0$ it holds that
    \begin{enumerate}[(1)]
        \item\label{lem:est-err-B-grad-inf-norm-item-1} $\|\nabla_{\bbeta}\cL(\bbeta^*,\bzeta^*)\|_{\infty} \leq  4\nu^2\sqrt{\frac{c\log^2(2nd)}{n}} $,
        
        \item\label{lem:est-err-B-grad-inf-norm-item-2}  $\|\nabla_{\bzeta}\cL(\bbeta^*,\bzeta^*)\|_{\infty} \leq  \frac{2c\alpha \log(2np)}{n}$,
    \end{enumerate}   
    with probability at least $1-(2nd)^{-c} - (2np)^{-c}$.
\end{lemma}
\begin{proof}[Proof of \Cref{lem:est-err-B-grad-inf-norm}]
    Recall that
    \begin{align}        
        \nabla_{\bbeta}\cL(\bbeta^*,\bzeta^*) &= \frac{1}{n}\sum_{\ell=1}^{np}[\tilde{y}_\ell-A'(\tilde{\bx}_\ell^{\top} \bbeta^* + \zeta_{\ell}^*)]\tilde{\bx}_{\ell} \\
        \nabla_{\zeta_{\ell}}\cL(\bbeta^*,\bzeta^*) & = \frac{1}{n}[\tilde{y}_\ell-A'(\tilde{\bx}_\ell^{\top} \bbeta^* + \zeta_{\ell}^*)],\qquad \ell\in[np] ,\label{eq:lem:est-err-B-grad-inf-norm-eq-3}
    \end{align}
    where $\tilde{\bx}_{\ell}$ is the $\ell$-th row of $\tilde{\bX}$.
    We split the proof into different parts.
    \paragraph{Part (1).} 
    Conditioned on $\bX$ and $\bZ^*$, the term $\tilde{y}_\ell-A'(\tilde{\bx}_\ell^{\top} \bbeta^* + \zeta_{\ell}^*)$ is a zero-mean $(\nu,\alpha)$-sub-exponential random variable, where $\nu = \sqrt{\kappa_2}$ and $\alpha = 1/C^2$, as shown in the proof of \Cref{thm:est-error-bTheta}.    
    Let $C'=\max_{\ell\in[np]}\|\tilde{\bx}_{\ell}^*\|_{\infty}$.
    Because $\tilde{\bx}_{\ell}^*$ are sparse vectors with $\|\tilde{\bx}_{\ell}^*\|_0=d$, we have that $[\tilde{y}_\ell-A'(\tilde{\bx}_\ell^{\top} \bbeta^* + \zeta_{\ell}^*)]\tilde{x}_{\ell j}$ is a zero-mean $(\nu,C'\alpha)$-sub-exponential random variable for $j=k,k+p,\ldots,k+p(d-1)$ and zero otherwise, where $k=\lfloor p/d\rfloor$.
    By Bernstein's inequality, it follows that 
    \[ \PP\left(\frac{1}{n}\sum_{\ell=1}^{np}[\tilde{y}_\ell-A'(\tilde{\bx}_\ell^{\top} \bbeta^* + \zeta_{\ell}^*)]\tilde{x}_{\ell j} \leq t\right) \geq 1-2\exp\left(- \frac{n}{2}\min\left\{\frac{t^2}{\nu^2}, \frac{t}{C'\alpha}\right\}\right).\]
    Applying union bound over $j=k,k+p,\ldots,k+p(d-1)$ yields that
    \begin{align*}
        \PP\left(\| \nabla_{\bbeta}\cL(\bbeta^*,\bzeta^*)\|_{\infty} \leq t\right) \geq 1-2d\exp\left(-n\min\left\{\frac{t^2}{2\nu^2}, \frac{t}{2C'\alpha}\right\}\right) .
    \end{align*}
    By setting $t = C'\sqrt{2\nu^2 c\log (2nd) / n}$ for some fixed constant $c>1$, we have
    \begin{align*}
        \PP\left(\| \nabla_{\bbeta}\cL(\bbeta^*,\bzeta^*)\|_{\infty} \leq C'\sqrt{\frac{2\nu^2c\log(2nd)}{n}} \right) \geq 1 - (2nd)^{1-c}, 
    \end{align*}
    when $n$ is large enough such that $t<\nu^2/(C'\alpha)$.
    By \Cref{lem:boundedness}, we also have $C'=\|\tilde{\bX}\|_{\max}\leq 2\sqrt{2}\nu\sqrt{c\log(nd)}$ with probability at least $1-(2nd)^{1-c}$.
    It follows that
    \begin{align}
        \PP\left(\| \nabla_{\bbeta}\cL(\bbeta^*,\bzeta^*)\|_{\infty} \leq 4\nu^2c\sqrt{\frac{\log^2(2nd)}{n}} \right) \geq 1 - 2(2nd)^{1-c}, \label{eq:lem:est-err-B-grad-inf-norm-eq-4}
    \end{align}
    which finishes the proof for Part (1).

    \paragraph{Part (2).} Because $\tilde{y}_\ell-A'(\tilde{\bx}_\ell^{\top} \bbeta^* + \zeta_{\ell}^*)$'s are zero-mean $(\nu,\alpha)$-sub-exponential random variables, by union bound, we have
    \[ \PP\left( n\|\nabla_{\bzeta}\cL(\bbeta^*,\bzeta^*)\|_{\infty} \leq t\right) \geq 1-2np\exp\left(- \frac{1}{2}\min\left\{\frac{t^2}{\nu^2}, \frac{t}{\alpha}\right\}\right),\]
    or equivalently,
    \begin{align*}
       \PP\left( \|\nabla_{\bzeta}\cL(\bbeta^*,\bzeta^*)\|_{\infty} \leq t\right) &\geq 1-2np\exp\left(- \frac{1}{2}\min\left\{\frac{(nt)^2}{\nu^2}, \frac{nt}{\alpha}\right\}\right)  \\
       &=  1-2np\exp\left(- \frac{nt}{2\alpha}\right)
    \end{align*}
    when $t\geq \nu^2/(\alpha n)$ is sufficiently large.
    Choosing $t = \max\{2c\alpha \log(2np), \nu^2/\alpha\}/n $ for any constant $c>1$ yields that
    \begin{align*}
       \PP\left( \|\nabla_{\bzeta}\cL(\bbeta^*,\bzeta^*)\|_{\infty} \leq \frac{2c\alpha \log(2np)}{n}\right)
       &\leq 1 - (2np)^{1-c}.
    \end{align*}

    Finally, we obtain the tail probability as the lemma states by taking the union bound over the above events.
\end{proof}

\begin{lemma}[Locally restricted strongly convexity]\label{lem:rsc}
    Under the same conditions in \Cref{thm:est-err-B}, define the augmented cone
    \begin{align}
        \begin{split}
            \cC_a(\xi,\cS,\eta_n) &:= \{(\bDelta_{\beta},\bDelta_{\zeta}) \in \RR^{pd}\times \RR^{np} \mid \bDelta_{\beta}\in\cC(\xi,\cS), \\
            & \max_{j\in[p]}\frac{1}{n}\|\bZ\bgamma_j-\bZ^*\bgamma^{*}_j\|_{2}^2\leq \frac{\log n}{n}\vee\eta_n ,\text{ such that }\bzeta=\vec(\bZ\bGamma^{\top})\},    
        \end{split}        
        \label{eq:lem:rsc:eq-aug-cone}
    \end{align}
    where the cone $\cC(\xi,\cS)$ defined in \eqref{eq:thm:est-err-B-eq-cone} and $\eta_n=o(1)$.  
    Then, it holds that
    \begin{align*}
        \inf_{(\bDelta_\beta,\bDelta_\zeta)\in \cC_a(\xi,\cS,\eta_n)}\cE(\bbeta^* + \bDelta_\beta, \bbeta^*;\bzeta^* + \bDelta_\zeta ,\bzeta^*) 
        & \geq \frac{\kappa_1C'}{2} \|\bDelta_{\beta}\|_2^2 -C''\sqrt{n^{-1} \vee \eta_n}  \|\bDelta_{\beta}\|_1,
    \end{align*}
    for some constant $C'>0$, with probability at least $1 -2 \exp(-\xi \log n)$.
\end{lemma}
\begin{proof}[Proof of \Cref{lem:rsc}]  
    For all $(\bDelta_{\beta},\bDelta_{\zeta})\in\cC_a(\xi,\cS,\eta_n)$, let $\bbeta=\bbeta^* + \bDelta_{\beta}$ and $\bzeta=\bzeta^* + \bDelta_{\zeta}$.
    Let $\bDelta_{\beta_j}\in\RR^d$ be the sub-vector containing the $(j-1)d$-th to $(jd-1)$-th entries of $\bDelta\in\RR^{pd}$.
    It is also equivalent to $\bb_j-\bb_j^*$, the $j$-th row of $\vec^{-1}(\bbeta-\bbeta^*)\in\RR^{p\times d}$.
    Let $\bE = \vec^{-1}(\bzeta)\in\RR^{n\times p}$ and $\bE^* = \vec^{-1}(\bzeta^*)\in\RR^{n\times p}$.
    
    We begin by decomposing the error into different terms:    
    \begin{align*}
        \cE(\bbeta, \bbeta^*;\bzeta,\bzeta^*) = \sum_{j=1}^p\cE_j(\bbeta, \bbeta^*;\bzeta,\bzeta^*),
    \end{align*}
    where
    \begin{align}
        \cE_j(\bbeta, \bbeta^*;\bzeta,\bzeta^*) &=  \frac{1}{n}\sum_{i=1}^{n}[A'(\bx_{i}^{\top} \bb_j + \bz_i^{\top}\bgamma_j) - A'(\bx_{i}^{\top} \bb_j^* + \bz_i^{*\top}\bgamma_j^*)]\bx_i^{\top}\bDelta_{\beta_j}\notag\\
        &= \frac{1}{n}\sum_{i=1}^{n}[A'(\bx_{i}^{\top} \bb_j + \bz_i^{\top}\bgamma_j) - A'(\bx_{i}^{\top} \bb_j^* + \bz_i^{\top}\bgamma_j)]\bx_i^{\top}\bDelta_{\beta_j} \notag\\
        &\qquad +  \frac{1}{n}\sum_{i=1}^{n}[A'(\bx_{i}^{\top} \bb_j^* + \bz_i^{\top}\bgamma_j) - A'(\bx_{i}^{\top} \bb_j^* + \bz_i^{*\top}\bgamma_j^*)]\bx_i^{\top}\bDelta_{\beta_j}\notag\\
        &=\frac{1}{n}\sum_{i=1}^{n}A''(\theta_{ij})(\bx_i^{\top}\bDelta_{\beta_j})^2 +  \frac{1}{n}\sum_{i=1}^{n}A''(\theta_{ij}')(\bz_i^{\top}\bgamma_j-\bz_i^{*\top}\bgamma_j^*)\bx_i^{\top}\bDelta_{\beta_j}\notag\\
        &= T_{1j} + T_{2j}, \label{eq:lem:rsc:eq-1} 
    \end{align}
    where $\theta_{ij}$ is between $\bx_{i}^{\top} \bb_j + \bz_i^{\top}\bgamma_j$ and $\bx_{i}^{\top} \bb_j^* + \bz_i^{\top}\bgamma_j$, and $\theta_{ij}'$ is between $\bx_{i}^{\top} \bb_j^* + \bz_i^{\top}\bgamma_j$ and $\bx_{i}^{\top} \bb_j^* + \bz_i^{*\top}\bgamma_j^*$.

    For the first term, because by \Cref{asm:model}, $\kappa_1:=\inf_{\theta\in\cR} A''(\theta)>0$, we have
    \begin{align*}
        T_{1j} &\geq \kappa_1 \frac{1}{n}\sum_{i=1}^{n}(\bx_i^{\top}\bDelta_{\beta_j})^2 \\
        &\geq \kappa_1\lambda_p\left(\frac{1}{n}\bX^{\top}\bX\right) \|\bDelta_{\beta_j}\|_2^2 \\
        &\geq \frac{\kappa_1}{C}\lambda_p\left(\frac{1}{n}\bSigma_x^{-\frac{1}{2}}\bX^{\top}\bX\bSigma_x^{-\frac{1}{2}}\right) \|\bDelta_{\beta_j}\|_2^2,
    \end{align*}
    where the last inequality is from \Cref{asm:covariate}.
    From \citet[Lemma 10.6.6]{vershynin2018high}, we further have that when $n\gtrsim \xi\log d$
    \begin{align}
        T_{1j} &\geq \frac{\kappa_1 C'}{2}\|\bDelta_{\beta_j}\|_2^2   \label{eq:lem:rsc:eq-T1}
    \end{align}
    over $\cC(\xi,\cS)$ (which contains $\cC_a(\xi,\cS, C_0)$), with probability at least $1 -2 \exp(-\xi \log n)$ for some absolute constant $C'>0$.

    For the second term, by Holder's inequality, we have
    \begin{align}
        |T_{2j}| &= \left|\frac{1}{n}(\bE_j - \bE_j^*)^{\top}\diag(A''(\bTheta_j')) \bX \bDelta_{\beta_j} \right| \notag\\
        &\leq \frac{1}{n}\|\bX^{\top}\diag(A''(\bTheta_j')) (\bE_j - \bE_j^*)\|_{\infty}\cdot \|\bDelta_{\beta_j}\|_1 \label{eq:lem:rsc:eq-T2-path-deriv}\\
        &=\max_{1\leq \ell\leq d} \frac{1}{n}|(\bE_j - \bE_j^*)^{\top}\diag(A''(\bTheta_j')) \bX_{\ell}|\cdot \|\bDelta_{\beta_j}\|_1, \notag
    \end{align}
    where $\bTheta_j'=(\theta_{1j}',\ldots,\theta_{nj}')^{\top}$.
    Recall that $\bE=\bZ\bGamma^{\top}$ and $\bE^*=\bZ^*\bGamma^{*\top}$ are such that $\max_{1\leq j\leq d}\frac{1}{\sqrt{n}}\|{\bE}_j - {\bE}_j^*\|_2 \lesssim \sqrt{(n^{-1}\log n )\vee\eta_n}$.
    From \Cref{cor:est-confound-col}, we have
    \begin{align}
        |T_{2j}| 
        &\leq \kappa_2 \max_{1\leq j\leq d}\frac{1}{\sqrt{n}}\|{\bE}_j - {\bE}_j^*\|_2\cdot\frac{1}{\sqrt{n}}\|\bX\|_{\oper}\cdot \|\bDelta_{\beta_j}\|_1\notag\\
        &\lesssim \sqrt{(n^{-1}\log n) \vee \eta_n} \|\bDelta_{\beta_j}\|_1, \label{eq:lem:rsc:eq-T2}
    \end{align}
    where the last inequality is because $\|\bX\|_{\oper}\lesssim \sqrt{n}$.

    By combining \eqref{eq:lem:rsc:eq-1}, \eqref{eq:lem:rsc:eq-T1} and \eqref{eq:lem:rsc:eq-T2}, we have that
    \[\cE_j(\bbeta, \bbeta^*;\bzeta,\bzeta^*) \geq \frac{\kappa_1 C'}{2} \|\bDelta_{\beta_j}\|_2^2 - C''\sqrt{n^{-1} \vee \eta_n}  \|\bDelta_{\beta_j}\|_1 ,\]
    for some constant $C''>0$. Thus,
    \[\cE(\bbeta, \bbeta^*;\bzeta,\bzeta^*) \geq \frac{\kappa_1 C'}{2} \|\bDelta_{\beta}\|_2^2  -C''\sqrt{n^{-1} \vee \eta_n}  \|\bDelta_{\beta}\|_1,\]
    over the cone $\cC_a(\xi,\cS,\eta_n)$.
\end{proof}

\begin{remark}[Neyman orthogonality]\label{remark:neyman-ortho}
    By coincidence, the proof above on $T_{2j}$ actually verifies the uniform Neyman orthogonality of the empirical loss in semiparametric models \citep{chernozhukov2018plug,foster2019orthogonal}. 
    This relies on the estimation error rates for the nuisance parameters $\bE_j^*$ by using the consistent estimation for the latent factors $\bZ^*$ with rate $\eta_n$ as assumed in the cone condition \eqref{eq:lem:rsc:eq-1}.
    To see this, recall that the pathwise derivative map of the gradient $\nabla_{\bb_j}\cL(\bbeta^*,\bzeta^*)$ evaluated at the true parameter $\bbeta^*$ and nuisance component value $\bzeta^*$ (when $t=0$) is given by
    \begin{align*}
        \frac{\partial}{\partial t} \nabla_{\bb_j}\cL(\bbeta^*,t(\bzeta-\bzeta^*)+\bzeta^*) \Big|_{t=0} &= \frac{1}{n}\sum_{i\in[n]}A''(\bx_i^{\top}\bb_j^* +e_{ij}^*)(e_{ij}-e_{ij}^*)\bx_i\\
        &= \frac{1}{n} \bX^{\top}\diag(A''(\bTheta_j^*))(\bE_j - \bE_j^*) .
    \end{align*}
    Compared to \eqref{eq:lem:rsc:eq-T2-path-deriv}, up to a constant factor, \eqref{eq:lem:rsc:eq-T2} also suggest that the pathwise derivative's infinity norm vanishes with a rate of $\sqrt{n^{-1} \vee \eta_n} $.
    In other words, at the true parameter value, local perturbations of the nuisance component around its true value have a negligible effect on the gradient of the loss with respect to the primary parameter, with high probability; see \citep{chernozhukov2018plug,foster2019orthogonal} for more detailed discussions about the Neyman orthogonality.
\end{remark}

\begin{lemma}[Bounds related to projection]\label{lem:proj-B-PGamma}
    Under assumptions in \Cref{thm:est-err-B}, it holds that 
    \begin{align*}
        \|\cP_{\hat{\bGamma}}^{\perp}\bGamma^*\|_{\fro}^2 =\cO\left(\frac{r}{n\wedge p}\right),
    \end{align*}
    with probability at least  $1-(n+p)^{-\delta}-(np)^{-\delta}-\exp(-n)$.
\end{lemma}
\begin{proof}[Proof of \Cref{lem:proj-B-PGamma}]
    From the proof of \Cref{thm:est-error-proj-bGamma}, the result follows by noting that
    \begin{align*}
        \|\cP_{\hat{\bGamma}}^{\perp}\bGamma^*\|_{\fro}^2 &= \|(\cP_{\hat{\bGamma}}^{\perp}-\cP_{{\bGamma}^*}^{\perp})\bGamma^*\|_{\fro}^2\\
        &= \sum_{\ell=1}^r\left\| \sum_{j=1}^p (\cP_{\hat{\bGamma}}^{\perp}-\cP_{{\bGamma}^*}^{\perp})\be_{j} \cdot \gamma_{j,\ell}^*\right\|_2^2 \\
        &\leq rpC^2 \max_{1\leq \ell\leq p}\|(\cP_{\hat{\bGamma}}^{\perp}-\cP_{{\bGamma}^*}^{\perp})\be_{j}\|_2^2 \\
        &=\cO(r(n\wedge p)^{-1}),
    \end{align*}
    with probability at least  $1-(n+p)^{-\delta}-(np)^{-\delta}-\exp(-n)$.
\end{proof}

\begin{lemma}[Sub-Gaussianity of $\bZ$]\label{lem:subgau-Z}
    Under \Crefrange{asm:model}{asm:latent}, $\bz_1,\ldots,\bz_n$ are independent and identically distributed sub-Gaussian random vectors.
\end{lemma}
\begin{proof}[Proof of \Cref{lem:subgau-Z}]
    Recall that $\bz_i^* = \bD^*\bx_i^* + \bw_i^*$ is the linear function of two independent sub-Gaussian random vectors.
    The mean is given by $\EE[\bz_i^*] = \bD^*\EE[\bx_i^*]$.
    It suffices to bound the operator norm of $\bD^*$.
    From $\bTheta^* = \bX\bB^{*\top} + \bZ^* \bGamma^{*\top}$, we have $\bX^{\top}\bTheta^* = \bX^{\top}\bX\bB^{*\top} + \bX^{\top}\bZ^* \bGamma^{*\top}$.
    Taking expectation over $\bX$ and $\bZ^*$ yield that $\EE[\bX^{\top}\bTheta^*/n] = \bSigma_x(\bB^{*\top} + \bD^{*\top} \bGamma^{*\top})$.
    Rearranging the formula yields that
    \begin{align*}
        \bGamma^*\bD^* &= \frac{1}{n}\EE[\bTheta^{*\top}\bX]\bSigma_x^{-1} -\bB^*
    \end{align*}
    and 
    \begin{align*}
        \bD^* &= (\bGamma^{*\top}\bGamma^*)^{-1}\bGamma^{*\top} \left(\frac{1}{n}\EE[\bTheta^{*\top}\bX^*]\bSigma_x^{-1} -\bB^*\right) .
    \end{align*}
    By the sub-multiplicative property of the operator norm, we have
    \begin{align*}
        \|\bD^*\|_{\oper} &\leq \|(\bGamma^{*\top}\bGamma^*)^{-1}\bGamma^{*\top}\|_{\oper} \left(\frac{1}{n}\|\EE[\bTheta^{*\top}\bX]\|_{\oper} \|\bSigma_x^{-1}\|_{\oper} + \|\bB^*\|_{\oper}\right)\\
        &\lesssim \frac{1}{\sqrt{p}} \left(\frac{1}{n}\EE[\|\bTheta^{*\top}\|_{\oper}\|\bX\|_{\oper}]  + \sqrt{p}\right) \\
        &\lesssim \frac{1}{\sqrt{p}} \left(\frac{\sqrt{np}\sqrt{n}}{n}  + \sqrt{p}\right)\\
        &\lesssim 1,
    \end{align*}
    where in the second inequality, we use Jensen's inequality and the norm inequality $\|\bB\|_{\oper}\leq \sqrt{p}\|\bB\|_{\max}$.
    This finishes the proof.    
\end{proof}

\begin{lemma}[Infinity norm of the covariates]\label{lem:boundedness}
    Under \Cref{asm:covariate}, it holds that $\|\tilde{\bX}\|_{\max} \leq 2\sqrt{2}\nu\sqrt{c\log (2nd)}$, with probability at least $1-(2nd)^{-c}$ for any fixed constant $c>0$.
\end{lemma}
\begin{proof}[Proof of \Cref{lem:boundedness}]
    Let $\vec^{-1}(\be_\ell) = \be_{n,i}\be_{p,j}^{\top}$ where $\be_{n,i}$ is the unit vector of dimension $n$ and $\be_{p,j}$ is defined analogously.
    Note that
    \begin{align*}
        \tilde{\bx}_{\ell} 
        &= (\bI_p \otimes \bX)^{\top} \be_\ell\\
        &= \vec(\bX^{\top}\vec^{-1}(\be_\ell)\bI_p) \\
        &= \vec(\bX^{\top}\be_{n,i}\be_{p,j}^{\top}) \\
        &= \vec(\bx_i\be_{p,j}^{\top}).
    \end{align*}   
    Because $\bx_i$'s are $\nu$-sub-Gaussian random vectors, we have 
    \begin{align*}
        \|\tilde{\bX}\|_{\max}&= \max_{\ell\in[np]}\|\tilde{\bx}_{\ell}\|_{\infty}\\
        &= \max_{i\in[n],j\in[p]}\|\bx_i\be_{p,j}^{\top}\|_{\infty}\\
        &\leq \max_{i\in[n]}\|\bx_i\|_{\infty}\\
        &\leq \max_{i\in[n],j\in[p]}|x_{ij}|\\
        &\leq 2\sqrt{2}\nu\sqrt{c\log (2nd)}
    \end{align*}
    with probability at least $1-(2nd)^{-c}$ for any fixed constant $c>0$.
\end{proof}

\section{Asymptotic normality of the debiased estimator}\label{app:sec:thm:normality}

\subsection{Proof of \Cref{thm:normality}}\label{app:subsec:thm:normality}

\begin{proof}[Proof of \Cref{thm:normality}]
    Condition on $\cD_2$, \Cref{thm:est-error-proj-bGamma} and \Cref{thm:est-err-B} imply that the following event holds:
    \begin{align*}
        \cE_1 &= \left\{\max_{1\leq j\leq p}\|(\cP_{\hat{\bGamma}}-\cP_{\bGamma^*})\be_j\|_{2}\lesssim (p(n\wedge p))^{-1/2}, \frac{1}{n}\|\hat{\bZ}\hat{\bgamma}_j-\bZ^*\bgamma_j^*\|_2^2\lesssim r_{n,p} ,\right. \\
        &\qquad\left. \|\hat{\bB}-\bB^*\|_{1,1} \lesssim \sqrt{sd}r_{n,p}, \|\hat{\bB}-\bB^*\|_{\fro} \lesssim r_{n,p}\right\}.
    \end{align*}
    where $r_{n,p}$ is defined in \Cref{thm:est-err-B}.

    Recall that 
    \begin{align}
        \hat{b}_{j1}^{\de} &= \hat{b}_{j1} + \hat{\bu}^{\top}\frac{1}{n}\sum_{i=1}^n \hat{\omega}_{i} \bx_{i} (\by_{i} - A'(\hat{\btheta}_{i}))^{\top} \bv_i\label{eq:proof-debias}
    \end{align}
    where $\bv_i=\diag(A''(\hat{\btheta}_i))^{-1}\cP_{\hat{\bGamma}}^{\perp}\be_j $.
    By Taylor expansion of $A'(\theta_{ij}^*)$ at $ \hat{\theta}_{ij} := \bx_i^{\top}\hat{\bb}_j + \hat{\bz}_i^{\top}\hat{\bgamma}_j$, we have
    \begin{align*}
        A'(\theta_{ij}^*) &= A'(\hat{\theta}_{ij}) + A''(\hat{\theta}_{ij}) ( \theta_{ij}^*-\hat{\theta}_{ij}) + \frac{1}{2}A'''(\psi_{ij}) (\theta_{ij}^* - \hat{\theta}_{ij})^2,
    \end{align*}
    for some $\psi_{ij}$ between $\hat{\theta}_{ij}$ and $\theta_{ij}^*$.
    Then, the residual of the $i$th sample can be decomposed into three sources of errors:
    \begin{align}
        \by_{i} - A'(\hat{\btheta}_{i})=&\underbrace{\bepsilon_{i}}_{\text{stochastic error}} + \underbrace{\bp_i}_{\text{remaining bias}}   + \underbrace{\bq_{i}}_{\text{ approximation error}} \label{eq:proof-error-decomp}
    \end{align}
    where the three terms of errors read that
    \begin{align*}
        \bepsilon_{i} &= \by_{i}- A'(\btheta_{i}^*)\\
        \bp_i &=A''(\hat{\btheta}_{i})\odot(\hat{\btheta}_{i}-\btheta_{i}^* )  \\
        \bq_i &= -\frac{1}{2}[A'''(\psi_{ij})(\theta_{ij}^* - \hat{\theta}_{ij})^2]_{j\in[p]}.
    \end{align*}
    Substituting \eqref{eq:proof-error-decomp} into \eqref{eq:proof-debias} yields that
    \begin{align}
        \hat{b}_{j1}^{\de} - b_{j1}^* &=    (\hat{b}_{j1} - b_{j1}^*) + \bu^{\top}\frac{1}{n}\sum_{i=1}^n \omega_i\bx_{i}\bepsilon_{i}^{\top}\bv_i + \bu^{\top}\frac{1}{n}\sum_{i=1}^n \omega_i\bx_{i}\bp_{i}^{\top}\bv_i  + \bu^{\top}\frac{1}{n}\sum_{i=1}^n \omega_i\bx_{i}\bq_{i}^{\top}\bv_i \notag\\
        &=  \bu^{\top}\frac{1}{n}\sum_{i=1}^n \omega_i\bx_{i}\bepsilon_{i}^{\top}\bv_i + \left(\bu^{\top}\frac{1}{n}\sum_{i=1}^n \omega_{i}\bx_{i} \bx_{i}^{\top} -\be_1^{\top}\right) (\bb_{j}^* - \hat{\bb}_{j}) \notag\\
        &\qquad  + 
        \bu^{\top}\frac{1}{n}\sum_{i=1}^n \bx_{i}  \bx_{i}^{\top}\bB^{*\top}((\cP_{\hat{\bGamma}}^{\perp}-\cP_{\bGamma^*}^{\perp}) \be_j - \cP_{\bGamma^*}\be_j) \notag\\
        &\qquad + \bu^{\top}\frac{1}{n}\sum_{i=1}^n  \bx_{i} \bq_i^{\top}\bv_{i} , \label{eq:proof-thm:normal-eq-decom}
    \end{align}
    In the second equality above, we use the properties that $\hat{\bB}$ and $\bB^*$ are in the column spaces of the orthogonal projections $\cP_{\hat{\bGamma}}^{\perp}$ and $\cP_{\bGamma^*}^{\perp}$, respectively. 
    Denote the four terms in the right-hand side of \eqref{eq:proof-thm:normal-eq-decom} by $T_{1j},T_{2j},T_{3j},T_{4j}$, respectively.
    We will analyze each of them separately, conditioning on $\cD_2$ and $\bX$.
    Then the randomness is from $\bepsilon_{i}$'s and $\bu$.
    We will show that the $T_1$ inherits $\sqrt{n}$-convergence rate and is asymptotically normally distributed, while the others have faster convergence rates.

    \underline{Part (1) $T_{1j}$.} From \Cref{lem:asym-normal}, it follows that 
    \begin{align*}
        \sqrt{n}\frac{T_{1j}}{\sigma_j} \dto \cN(0,1).
    \end{align*}
    
    \underline{Part (2) $T_{2j}$.} By Holder's inequality and the constraint of optimization problem \eqref{opt:min-var}, it follows that
    \begin{align*}
        |T_{2j}| &\leq \lambda_n \|\bb_j^*-\hat{\bb}_j\|_1 \lesssim \sqrt{\frac{\log(nd)}{n}} r_{n,p}',
    \end{align*}
    with probability tending to one.
    Thus, we have $\sqrt{n}|T_{2j}| \pto 0$ as $n,p\rightarrow\infty$.

    \underline{Part (3) $T_{3j}$.} From \Cref{thm:est-error-proj-bGamma} and \eqref{eq:proof-prop-norm} in the proof of \Cref{prop:iden-B}, we have that
    \begin{align*}
        |T_{3j}| &\leq \|\bu\|_2 \cdot \left\|\frac{1}{n}\sum_{i=1}^n \bx_i\bx_i^{\top}\right\|_{\oper}(\|\bB^*\|_{\oper} \|(\cP_{\hat{\bGamma}}^{\perp}-\cP_{\bGamma^*}^{\perp}) \be_j\|_2 + \|\bB^{*\top}\cP_{\bGamma^*} \be_j\|_2)\\        
        &\lesssim \sqrt{\frac{d}{p(n\wedge p)}} + \frac{sd}{p},
    \end{align*}
    with probability tending to one. Thus, if $\sqrt{n}/p\rightarrow 0$, we have $\sqrt{n}|T_{3j}| \pto 0$ as $n,p\rightarrow\infty$.

    \underline{Part (4) $T_{4j}$.} The higher-order term is bounded as below:
    \begin{align*}
        |T_{4j}| &\leq \frac{\kappa_2}{2\kappa_1}\max_{1\leq i\leq n}|\bu^{\top}\bx_i|\cdot \frac{1}{n}\sum_{i=1}^n \sum_{j=1}^n|\bx_i^{\top}(\hat{\bb}_j - \bb_j^*)|^2 \cdot \|\cP_{\hat{\bGamma}}^{\perp}\be_j\|_2\\
        &\lesssim \max_{1\leq i\leq n}|\bu^{\top}\bx_i|^3 \|\hat{\bB}-\bB^*\|_{\fro}^2 \\
        &\lesssim \tau_n^3 r_{n,p}^2,
    \end{align*}
    with probability tending to one.
    Thus, if $n/(\log(nd)p^{3/2})\rightarrow 0$ and $\sqrt{n}/p^{1-k}\rightarrow 0$, we have $\sqrt{n}|T_{4j}| \pto 0$ as $n,p\rightarrow\infty$.

    We are now combining the above four terms.
    Because when $n/\log(nd)=o(p^{3/2})$ and $n=o(p^{2(1-k)})$, $T_{2j},\ldots,T_{4j}=\op(1/\sqrt{n})$, we have 
    $\sqrt{n}(\hat{b}_{j1}^{\de} - b_{j1}^* ) /\sigma_j \dto\cN(0,1)$.
\end{proof}

\subsection{Proof of \Cref{prop:simul-inference}}\label{app:subsec:prop:simul-inference}
\begin{proof}[Proof of \Cref{prop:simul-inference}]
    By the definition of $t_j$ and \eqref{eq:proof-thm:normal-eq-decom} we have the decomposition
    \begin{align*}
        t_j &= \vartheta_j + \varsigma_j,
    \end{align*}
    where 
    \begin{align*}
        \vartheta_j &= \sqrt{n}\frac{\hat{\bu}^{\top}\frac{1}{n}\sum_{i=1}^n \hat{\omega}_i\bx_{i}\epsilon_{ij}A''(\hat{\theta}_{ij})^{-1}}{\hat{\sigma}_j} ,\\
        \varsigma_j &= \sqrt{n}\frac{\hat{\bu}^{\top}\frac{1}{n}\sum_{i=1}^n \hat{\omega}_i\bx_{i}\bepsilon_{i}^{\top}\diag(A''(\hat{\btheta}_i))^{-1}(\cP_{\hat{\bGamma}}^{\perp} -\bI_p)\be_j }{\hat{\sigma}_j} + \sqrt{n}\frac{T_{2j}+T_{3j}+T_{4j}}{\hat{\sigma}_j}.
    \end{align*}
    For the first component, note that $\vartheta_j$ for $j=1,\ldots,p$ are independent conditional on $\{(\bx_i,\bz_i^*)\}_{i=1}^n$ and $\cD_2$.
    Furthermore, $\vartheta_j\dto \cN(0,1)$ for $j\in\cN_p$ from \Cref{lem:asym-normal} by noting that $\hsigma_j$ is also consistent to the conditional variance of $\vartheta_j$.
    For the second component, from the proof of \Cref{thm:normality} and \Cref{lem:hsigma}, we know that
    \[\max_{1\leq j\leq p}\sqrt{n}\frac{T_{2j}+T_{3j}+T_{4j}}{\hat{\sigma}_j} =\op(1).\]
    On the other hand, from \Cref{thm:est-error-proj-bGamma} and \Cref{asm:latent}, we also have
    \begin{align*}
        &\max_{1\leq j\leq p}\sqrt{n}\frac{\hat{\bu}^{\top}\frac{1}{n}\sum_{i=1}^n \hat{\omega}_i\bx_{i}\bepsilon_{i}^{\top}\diag(A''(\hat{\btheta}_i))^{-1}(\cP_{\hat{\bGamma}}^{\perp} -\bI_p)\be_j }{\hat{\sigma}_j} \\
        =& \max_{1\leq j\leq p}\sqrt{n}\frac{\hat{\bu}^{\top}\frac{1}{n}\sum_{i=1}^n \hat{\omega}_i\bx_{i}\bepsilon_{i}^{\top}\diag(A''(\hat{\btheta}_i))^{-1}(\cP_{\hat{\bGamma}}^{\perp} -\cP_{{\bGamma}^*}^{\perp})\be_j }{\hat{\sigma}_j}  \\
        &\qquad + \max_{1\leq j\leq p}\sqrt{n}\frac{\hat{\bu}^{\top}\frac{1}{n}\sum_{i=1}^n \hat{\omega}_i\bx_{i}\bepsilon_{i}^{\top}\diag(A''(\hat{\btheta}_i))^{-1}\cP_{{\bGamma}^*}\be_j }{\hat{\sigma}_j} \\
        =& \Op\left(\sqrt{p}\max_{1\leq j\leq p}\|(\cP_{\hat{\bGamma}}^{\perp} -\cP_{{\bGamma}^*}^{\perp})\be_j\|_2\right)    \\
        &\qquad + \max_{1\leq j\leq p}\left\|\sqrt{n}\frac{\hat{\bu}^{\top}\frac{1}{n}\sum_{i=1}^n \hat{\omega}_i\bx_{i}\bepsilon_{i}^{\top}\diag(A''(\hat{\btheta}_i))^{-1}{\bGamma}^* }{\hat{\sigma}_j} \right\|_{2}\|({\bGamma}^{*\top}{\bGamma}^*)^{-1}\|_{\oper}\|{\bgamma}_j^*\|_2\\
        = &\Op(p^{-\frac{1}{2}})  + \Op(\sqrt{rp} \cdot p^{-1} \cdot 1)\\
        =& \op(1).
    \end{align*}
    where we use the subexponential concentration of $\epsilon_{ij}$ conditional on $\{(\bx_i,\bz_i^*)\}_{i=1}^n$ and $\cD_2$.
    Therefore, we have that $\max_{1\leq j\leq p}|\varsigma_j| =\op(1)$.

    The rest of the proof follows similarly to the proof of \citet[Theorem 3.4]{wang2017confounder}.
    We present here for completeness.

    \paragraph{Overall Type-I error control.} 
    Let $\varrho = |\mathcal{N}_p|^{-1} \sum_{j \in \mathcal{N}_p} \ind\left(\left|t_j\right|>z_{\frac{\alpha}{2}}\right)$.
    To prove the overall Type-I error control, we will show the expectation of $\varrho$ tends to $\alpha$ and its variance tends to zero.
    For the expectation, for any $\epsilon>0$, we have
    \begin{align*}
        \EE[\varrho] &=\frac{1}{\left|\mathcal{N}_p\right|} \sum_{j \in \mathcal{N}_p} \PP\left(\left|t_j\right|>z_{\frac{\alpha}{2}}\right)\\
        & \leq \frac{1}{\left|\mathcal{N}_p\right|} \sum_{j \in \mathcal{N}_p} [\PP\left(\left|\vartheta_j\right|>z_{\frac{\alpha}{2}}-\epsilon\right)+\PP\left(\left|\varsigma_j\right|>\epsilon\right)] \\
        & = \frac{1}{\left|\mathcal{N}_p\right|} \sum_{j \in \mathcal{N}_p} \PP\left(\left|\vartheta_j\right|>z_{\frac{\alpha}{2}} -\epsilon\right) +\frac{1}{\left|\mathcal{N}_p\right|} \sum_{j \in \mathcal{N}_p} \PP\left(\left|\varsigma_j\right|>\epsilon\right) \\
        & \leq \frac{1}{\left|\mathcal{N}_p\right|} \sum_{j \in \mathcal{N}_p} \PP\left(\left|\vartheta_j\right|>z_{\frac{\alpha}{2}} -\epsilon\right) +\PP\left(\max _{1 \leq j \leq p}\left|\varsigma_j\right|>\epsilon\right) \\
        & \rightarrow 2\left(1-\Phi\left(z_{\frac{\alpha}{2}}-\epsilon\right)\right) ,
    \end{align*}
    where the last convergence is because the Cesaro mean converges to the same limit as $\lim_{n,p} \PP\left(\left|\vartheta_j\right|>z_{\frac{\alpha}{2}} -\epsilon\right) = 2\left(1-\Phi\left(z_{\frac{\alpha}{2}}-\epsilon\right)\right) $ and the second term $\PP(\max _{1 \leq j \leq p}\left|\varsigma_j\right|>\epsilon)$ varnishes.
    Similarly, we can also show that $\liminf_{n,p\rightarrow\infty}\EE[\varrho]\geq 2\left(1-\Phi\left(z_{\frac{\alpha}{2}}-\epsilon\right)\right) $ for all $\epsilon>0$.
    This implies that $\EE[\varrho]\to \alpha$ as $n,p\rightarrow\infty$.

    Next, we analyze the second moment.
    For any $\epsilon>0$, the second moment can be upped bounded as
    \begin{align*}
        \EE[\varrho^2] &= \frac{1}{\left|\mathcal{N}_p\right|^2} \sum_{j, k \in \mathcal{N}_p} \PP\left(\left|t_j\right|>z_{\frac{\alpha}{2}},\left|t_k\right|>z_{\frac{\alpha}{2}}\right) \\
         &= \frac{1}{\left|\mathcal{N}_p\right|^2} \sum_{j \in \mathcal{N}_p} \PP\left(\left|t_j\right|>z_{\frac{\alpha}{2}}\right)+\frac{1}{\left|\mathcal{N}_p\right|^2} \sum_{j, k \in \mathcal{N}_p, j \neq k} \PP\left(\left|t_j\right|>z_{\frac{\alpha}{2}},\left|t_k\right|>z_{\frac{\alpha}{2}}\right) \\
        &\leq  \frac{1}{\left|\mathcal{N}_p\right|^2} \sum_{j \in \mathcal{N}_p} \PP\left(\left|t_j\right|>z_{\frac{\alpha}{2}}\right) \\
        &\qquad +\frac{1}{\left|\mathcal{N}_p\right|^2} \sum_{j, k \in \mathcal{N}_p, j \neq k} \PP\left(\left|\vartheta_j\right|>z_{\frac{\alpha}{2}}-\epsilon,\left|\vartheta_k\right|>z_{\frac{\alpha}{2}}-\epsilon\right)  +\PP\left(\left|\varsigma_j\right|>\epsilon\right)+\PP\left(\left|\varsigma_k\right|>\epsilon\right) \\
        &= \frac{1}{\left|\mathcal{N}_p\right|^2} \sum_{j, k \in \mathcal{N}_p, j \neq k} \PP\left(\left|\vartheta_j\right|>z_{\frac{\alpha}{2}}-\epsilon,\left|\vartheta_k\right|>z_{\frac{\alpha}{2}}-\epsilon\right)+o(1) \\
        &= \frac{1}{\left|\mathcal{N}_p\right|^2} \sum_{j, k \in \mathcal{N}_p, j \neq k} \EE[\PP(\left|\vartheta_j\right|>z_{\frac{\alpha}{2}}-\epsilon\mid \cC)\PP(\left|\vartheta_k\right|>z_{\frac{\alpha}{2}}-\epsilon\mid \cC)]+o(1) \\
        &\rightarrow \left[2\left(1-\Phi\left(z_{\frac{\alpha}{2}}-\epsilon\right)\right)\right]^2,
    \end{align*}
    where the last equality is from the independence of $\vartheta_j$ and $\vartheta_k$ condition on $\cC=\{(\bx_i,\bz_i^*)\}_{i=1}^n\cup\cD_2$.
    We can similarly obtain the lower bound.
    This implies that $\EE[\varrho^2]\rightarrow\alpha^2$ and $\Var(\varrho)\rightarrow0$ as $n,p\rightarrow\infty$.
    Combining the previous results yields that $\varrho\pto \alpha$.

    \paragraph{FWER control.} To prove the second statement, note that
    \begin{align*}
        \PP\left(|\cN_p|\varrho \geq 1\right) 
        = & \PP\left(\max _{j \in \mathcal{N}_p}\left|t_j\right|>\Phi^{-1}(1-\alpha /(2 p))\right) \\
        = & \PP\left(\max _{j \in \mathcal{N}_p}\left|\vartheta_j+\varsigma_j\right|>\Phi^{-1}(1-\alpha /(2 p))\right) \\
        \leq & \PP\left(\max _{j \in \mathcal{N}_p}\left|\vartheta_j\right|>\Phi^{-1}(1-\alpha /(2 p))-\max _{j \in \mathcal{N}_p}\left|\varsigma_j\right|\right) \\
        \leq & \PP\left(\max _{1 \leq j \leq p}\left|\vartheta_j\right|>\Phi^{-1}(1-\alpha /(2 p))-\max _{j \in \mathcal{N}_p}\left|\varsigma_j\right|\right),
    \end{align*}
    which is asymptotically upper bounded by $\alpha$, after applying Gaussian approximation \citep[Lemma 2.3]{chernozhukov2013gaussian} and the valid control of Bonferroni correction for i.i.d. normal random variables, by noting that $\Phi^{-1}(1-$ $\alpha /(2 p)) \rightarrow \infty$ as $p \rightarrow \infty$, and the result that $\max_{j \in \mathcal{N}_p}\left|\varsigma_j\right|=\op(1)$.
    
\end{proof}

\subsection{Technical lemmas}\label{app:subsec:inference-tech-lems}
\begin{lemma}\label{lem:sol-u}
    Under the same conditions as in \Cref{thm:normality}, suppose event $\cE_1$ holds, then the solution to optimization problem \eqref{opt:min-var} exists with probability at least $1-2(nd)^{-c}$.
\end{lemma}
\begin{proof}[Proof of \Cref{lem:sol-u}]
    Define the matrix $\bS = \EE[\hat{\omega}_{i}\bx_i\bx_i^{\top}]$.
    We next show that (1) $\bS$ is invertible and (2) the $j$-th column $\bu^*$ of $\bS^{-1}$ is feasible for the constraints of the optimization problem \eqref{opt:min-var} with high probability.
    We split the proof into two parts, as below.

    \paragraph{Part (1)}
    Because $C^{-1}\leq \hat{\omega}_i\leq C$, we have $C^{-1}\EE[\bx_1\bx_1^{\top}]\preceq\bS\preceq C\EE[\bx_1\bx_1^{\top}]$.
    On the other hand, note that $\EE[\bx_1\bx_1^{\top}]=\bSigma_x \succeq \lambda_{\min}(\bSigma_x)\bI_d$.
    Thus, for any unit vector $\ba\in\RR^d$, we have $\ba^{\top}\bS\ba \geq  C^{-1}\lambda_{\min}(\bSigma_x)>0$.
    This establishes claim (1).

    \paragraph{Part (2)}
    Let $\bu^*$ be the $j$-th column of $\bS^{-1}$.
    By definition, we have $\bS\bu^*=\be_j$.
    Conditional on $\cD_2$, we have that $\hat{\omega}_i\bu^{*\top}\bx_i\bx_i^{\top}\be_k$ for $i=1,\ldots,n$ are independent random variables with mean $\delta_{jk}$.
    Because $\hat{\omega}_i$ is bounded, we further have that $\hat{\omega}_i\bu^{*\top}\bx_i\bx_i^{\top}\be_k$'s are independent sub-exponential random variables.
    Applying Bernstein's inequality as in the proof of \Cref{lem:est-err-B-grad-inf-norm}, we have with probability at least $1-(nd)^{-c}$,
    \begin{align*}
        \left\|\frac{1}{n}\sum_{i=1}^n\hat{\omega}_i\bu^{*\top}\bx_i\bx_i^{\top} - \be_j\right\|_{\infty} &\leq \lambda_n,
    \end{align*}
    where $\lambda_n\asymp \sqrt{\log (nd)/ n}$.
    This also holds after taking account of the randomness of $\cD_2$.
    
    On the other hand, because $\omega_i$ is bounded away from zero and infinity, and $\|\bSigma_x\|_{\oper}=\cO(1)$, it follows that $\|\bS\|_{\oper}=\cO(1)$ and $\|\bu^*\|_2=\cO(1)$.
    By the sub-Gaussianity of $\bx_i$, we also have $\|\bX\bu^*\|_{\infty}=\max_{1\leq i\leq n}|\bx_i^{\top}\bu^*|\leq \tau_n$, with probability at least $1 - n^{-c}$.
    The above shows that $\bu^*$ is feasible for optimization problem \eqref{opt:min-var}, which establishes the claim (2).
    
    Finally, taking the union bound over the two probabilistic events finishes the proof.
\end{proof}

\begin{lemma}[Asymptotic normality]\label{lem:asym-normal}
    Under the conditions in \Cref{thm:normality}, it holds that
    \begin{align*}
        \sqrt{n}\sum_{i=1}^n \sigma_j^{-1}\hat{\bu}^{\top}\bx_i\bepsilon_{i}^{\top} \hat{\bv}_{i} \dto \cN(0,1),
    \end{align*}
    where $\sigma_j^2=n^{-1}\Var\left(\bu^{\top}\sum_{i=1}^n \bx_{i} \bepsilon_{i}^{\top}\bv_{i} \,\mid\, \{\bx_{i},\bz_{i}\}_{i=1}^n,\cD_{2}\right)$.
\end{lemma}
\begin{proof}[Proof of \Cref{lem:asym-normal}]
    Note that when conditioning on the natural parameters $\bTheta^*$, $\epsilon_{ij}$'s are independent ($\nu,\alpha$)-sub-exponential random variable as shown in the proof of \Cref{thm:est-error-bTheta}.
    Define $\xi_i := \sigma_j^{-1}\hat{\bu}^{\top}\bx_i\bepsilon_{i} \hat{\bv}_{i}$
    for $i\in[n]$.
    Then $\xi_i$'s are independent random variables with mean $\EE[\xi_i\mid \cD_2,\bX]=0$ and variance $\Var(\xi_i \mid \cD_2,\bX)=1$.    
    It suffices to check the bounded variance condition and Lindeberg's condition.
    
    \paragraph{Part (1) Boundedness of $\sigma_j$.}
    We first show the boundedness of the variance 
    \[\sigma_j^2 = \hat{\bu}^{\top}\frac{1}{n}\sum_{i=1}^n \hat{\omega}_{i}^2 (\be_j^{\top}\cP_{\hat{\bGamma}}^{\perp}\diag(A''(\hat{\btheta}_i))^{-1}\diag(A''(\btheta_{i}^*))\diag(A''(\hat{\btheta}_i))^{-1}\cP_{\hat{\bGamma}}^{\perp}\be_j)\bx_{i}\bx_{i}^{\top}\hat{\bu}.\]
    Because $A''(\theta)\geq \kappa_1>0$ for all $\theta\in\cR$, the quadratic term satisfies that 
    \[\be_j\cP_{\hat{\bGamma}}^{\perp}\diag(A''(\hat{\btheta}_i))^{-1}\diag(A''(\btheta_{i}^*))\diag(A''(\hat{\btheta}_i))^{-1}\cP_{\hat{\bGamma}}^{\perp}\be_j \geq 0,\] with equality holds if and only if $\cP_{\hat{\bGamma}}^{\perp}\be_j=\zero_p$.
    On the other hand, we have
    \begin{align*}
        \|\cP_{\hat{\bGamma}}^{\perp}\be_j\|_2 =\|\be_j-\cP_{\hat{\bGamma}}\be_j\|_2 \geq 1 - \|\cP_{\hat{\bGamma}}\be_j\|_2 \gtrsim \Omega(1 - p^{-1/2}),
    \end{align*}
    where the last inequality is because 
    \[\|\cP_{\hat{\bGamma}}\be_j\|_2 \leq \|\cP_{\bGamma^*}\be_j\|_2 + \|(\cP_{\bGamma^*}-\cP_{\hat{\bGamma}})\be_j\|_2 \lesssim \frac{1}{\sqrt{p}}\]
    from \Cref{asm:latent} and \Cref{thm:est-error-proj-bGamma}.
    This implies that, when $n,p$ are sufficiently large,
    \[c^{-1} \leq \hat{\omega}_{i} (\be_j^{\top}\cP_{\hat{\bGamma}}^{\perp}\diag(A''(\hat{\btheta}_i))^{-1}\diag(A''(\btheta_{i}^*))\diag(A''(\hat{\btheta}_i))^{-1}\cP_{\hat{\bGamma}}^{\perp}\be_j) \leq c\]
    for some constant $c>1$, under event $\cE_1$.
    Thus, it is equivalent to show the boundedness of $\hsigma_j^2 = \hat{\bu}^{\top}\hat{\bS}\hat{\bu}$, where $\hat{\bS}=n^{-1}\sum_{i=1}^n\hat{\omega}_i\bx_i\bx_i^{\top}$.
    From \Cref{lem:sol-u}, we know that $\bS=\EE[\hat{\bS}]$ has a bounded spectrum with high probability.
    The upper bound that $\hat{\sigma}_j^2\leq \bS_{jj}$ with high probability then follows by the sub-exponential concentration results as in the proof of \Cref{lem:sol-u}.
    
    Next, we proceed to show the lower bound.
    Because $\hat{\bu}$ satisfies the constraint $|\be_j^{\top}\hat{\bS}\hat{\bu} - 1|\leq \lambda_n$, we have that $\sigma_j^2 \geq \hat{\bu}^{\top}\hat{\bS}\hat{\bu} + t ((1-\lambda_n) - \be_j^{\top}\hat{\bS}\hat{\bu}) $ for any $t>0$.
    Note that $\min_{\bv\in\RR^d}\bv^{\top}\hat{\bS}\bv + t ((1-\lambda_n) - \be_j^{\top}\hat{\bS}\bv) = -t^2 \be_j^{\top}\hat{\bS}\be_j / 4 + t (1-\lambda_n)$ where the minimum is obtained when $\hat{\bS}\bv = t\hat{\bS}\be_j/2$. 
    We further have $\hat{\sigma}_j^2\geq \max_{t\geq 0}-t^2 \be_j^{\top}\hat{\bS}\be_j / 4 + t (1-\lambda_n)\geq (1-\lambda_n)^2/(\be_j^{\top}\hat{\bS}\be_j)$.
    By the sub-Gaussianity of $\bx_i$, $\be_j^{\top}\hat{\bS}\be_j \geq \bS_{jj} + \op(1)$.
    We then have $\hat{\sigma}_j^2\geq 0.5/\bS_{jj}$ when $n$ and $p$ are large enough.

    \paragraph{Part (2) Lindeberg's condition.}
    On the other hand, because 
    \[\max_{1\leq i\leq n}|\xi_i| \leq \max_{1\leq i\leq n}|\hat{\bu}^{\top}\bx_i| \|\bepsilon_i\|_2|\sigma_j^{-1}| \|\hat{\bv}_{i}\|_2\lesssim \sqrt{n},\]
    with probability at least $1-2(nd)^{-c}$, the Lindeberg's condition holds that
    \begin{align*}
        \lim_{n\rightarrow\infty}\frac{1}{n}\sum_{i=1}^n \EE[\xi_i^2 \ind\{|\xi_i|\geq \epsilon\sqrt{n}\}] =0
    \end{align*}
    for all $\epsilon>0$.
    Applying Lindeberg’s central limit theorem yields that
    \begin{align*}
        \sqrt{n}\sum_{i=1}^n \sigma_j^{-1}\hat{\bu}^{\top}\bx_i\bepsilon_{i} \hat{\bv}_{i} \dto \cN(0,1),
    \end{align*}
    which finishes the proof.
\end{proof}

\begin{lemma}[Consistent estimators of $\sigma_j$]\label{lem:hsigma}
    Under conditions in \Cref{thm:est-err-B} and condition (i) in \Cref{thm:normality}, $\hat{\omega}_j=A''(\hat{\theta}_{ij})$ satisfies condition (ii) of \Cref{thm:normality}.
    Furthermore, for the variance estimate defined in \eqref{eq:est-var} using sample splitting procedure \Cref{alg:split}, it holds that $\hsigma_j \pto \sigma_j.$
\end{lemma}
\begin{proof}[Proof of \Cref{lem:hsigma}]
    The boundedness of $\hat{\omega}_j$ follows from \eqref{eq:kappa}.

    By the sample splitting procedure \Cref{alg:split}, we know that for a given response $j\in[p]$, $\bY_j-A'(\bTheta_j^*)\in\RR^{n}$ is independent of $\hat{\bu}$, $\hat{\bB}$, $\hat{\bGamma}$, and $\hat{\bZ}$, when conditioning on $\bX$.
    However, noted that $\hat{\bB}$, $\hat{\bGamma}$, and $\hat{\bZ}$ may be specific to each $j\in I$.
    For the sake of simplicity, in the following proof, we will assume that $\bY-A'(\bTheta^*)$ is independent of $\hat{\bu}$, and a common set of estimators $\hat{\bB}$, $\hat{\bGamma}$, and $\hat{\bZ}$ when conditioning on $\bX$; or equivalently $I=[p]$.
    Note that, however, the proof still works for the cases in \Cref{alg:split}, except the constructed debiased estimators only use responses in the index set $I$; namely, $\hat{b}_{j1}^{\de} = \hat{b}_{j1} + \hat{\bu}^{\top}\frac{1}{n}\sum_{i=1}^n \hat{\omega}_{i} \bx_{i} ([\by_{i}]_I - [A'(\hat{\btheta}_{i})]_I)^{\top} \cP_{\hat{\bGamma}_I}^{\perp}\be_j$.

    Recall that 
    \[\sigma_j^2 = \hat{\bu}^{\top}\frac{1}{n}\sum_{i=1}^n \hat{\omega}_{i}^2 (\be_j^{\top}\cP_{\hat{\bGamma}}^{\perp}\diag(A''(\hat{\btheta}_i))^{-1}\diag(A''(\btheta_{i}^*))\diag(A''(\hat{\btheta}_i))^{-1}\cP_{\hat{\bGamma}}^{\perp}\be_j)\bx_{i}\bx_{i}^{\top}\hat{\bu},\]
    and
    \[\hat{\sigma}_j^2 = \hat{\bu}^{\top}\frac{1}{n}\sum_{i=1}^n \hat{\omega}_{i}\bx_{i}\bx_{i}^{\top}\hat{\bu}.\]
    Let $a_i = \hat{\omega}_{i}^2 (\be_j^{\top}\cP_{\hat{\bGamma}}^{\perp}\diag(A''(\hat{\btheta}_i))^{-1}\diag(A''(\btheta_{i}^*))\diag(A''(\hat{\btheta}_i))^{-1}\cP_{\hat{\bGamma}}^{\perp}\be_j)$ and $a_i' = \hat{\omega}_i$.
    We begin by bounding the difference between the two:
    \begin{align*}
        |\hsigma_j^2 - \sigma_j^2|  &= \left| \frac{1}{n}\sum_{i=1}^n (a_{i}' - a_{i})\cdot (\hat{\bu}\bx_i)^2\right|\\
        & \leq  \sqrt{\frac{1}{n}\sum_{i=1}^n(\hat{\bu}\bx_i)^4}\cdot \sqrt{\frac{1}{n}\sum_{i=1}^n (a_{i}'- a_{i})^2}\\
        &\lesssim \tau_n^2\cdot \sqrt{\frac{1}{n}\sum_{i=1}^n (a_{i}'- a_{i})^2} ,
    \end{align*}
    where the first inequality is from Holder's inequality and the second inequality is due to the second constraint of the optimization problem \eqref{opt:min-var}.
    For the second factor above, note that
    \begin{align}
        &\frac{1}{n}\sum_{i=1}^n (a_{i}'- a_{i})^2 \notag\\
        =& \max_{i\in[n]}\hat{\omega}_i^4\cdot  \frac{1}{n}\sum_{i=1}^n (\be_j^{\top}\cP_{\hat{\bGamma}}^{\perp}\diag(A''(\hat{\btheta}_i))^{-1}\diag(A''(\btheta_{i}^*) )\diag(A''(\hat{\btheta}_i))^{-1}\cP_{\hat{\bGamma}}^{\perp}\be_j - A''(\hat{\btheta}_i))^2. \label{eq:lem:sigma-eq-0}
    \end{align}
    Each term inside the square can be decomposed into
    \begin{align}
        &\be_j^{\top}\cP_{\hat{\bGamma}}^{\perp}\diag(A''(\hat{\btheta}_i))^{-1}\diag(A''(\btheta_{i}^*) )\diag(A''(\hat{\btheta}_i))^{-1}\cP_{\hat{\bGamma}}^{\perp}\be_j - A''(\hat{\theta}_{ij}) \notag\\
        =& \be_j^{\top}\cP_{\hat{\bGamma}}^{\perp}\diag(A''(\hat{\btheta}_i))^{-1}\diag(A''(\btheta_{i}^*) - A''(\hat{\btheta}_i))\diag(A''(\hat{\btheta}_i))^{-1}\cP_{\hat{\bGamma}}^{\perp}\be_j \notag\\
        &\qquad - [2\be_j^{\top} \cP_{\hat{\bGamma}}A''(\hat{\btheta}_{i}) \cP_{\hat{\bGamma}}^{\perp}\be_j + \be_j^{\top} \cP_{\hat{\bGamma}}A''(\hat{\btheta}_{i}) \cP_{\hat{\bGamma}}\be_j ] \notag\\
        =&: T_1 + T_2. \label{eq:lem:sigma-eq-1}
    \end{align}
    Now note that
    \begin{align}
        \frac{1}{n}\sum_{i=1}^nT_1^2
        \lesssim& \frac{1}{n}\sum_{i=1}^n(A''(\btheta_{ij}^*) - A''(\hat{\btheta}_{ij}))^2\notag \\
        = & \frac{1}{n}(\hat{\bTheta}_j - \bTheta_j^*)^{\top}\diag(A'''(\bTheta_j')) (\hat{\bTheta}_j - \bTheta_j^*) \notag\\
        \lesssim& \frac{1}{n}\|\hat{\bTheta}_j - \bTheta_j^*\|_2^2, \label{eq:lem:sigma-eq-2}
    \end{align}
    where the first inequality is due to the boundedness of $A{''}$ on $\cR$, and the bounded spectral of the projection matrix $\cP_{\bGamma}^{\perp}$, and noting that $\|\cP_{\hat{\bGamma}}\be_j\|_2\lesssim \Op(p^{-1/2})$;
    the second equality is from Taylor expansion with $\bTheta_j'=(\theta_{1j}',\ldots,\theta_{nj}')$ and $\theta_{ij}'$ being between $\hat{\theta}_{ij}$ and $\theta_{ij}^*$ for $i=1,\ldots,n$;
    and the second inequality is from the continuity and boundedness of $A{'''}$ on $\cR_C$.

    On the other hand,
    \begin{align}
        \frac{1}{n}\sum_{i=1}^nT_2^2 &\lesssim p^{-1} \label{eq:lem:sigma-eq-3}
    \end{align}
    by noting that $\|\cP_{\hat{\bGamma}}\be_j\|_2\lesssim \Op(p^{-1/2})$ again.

    By applying triangle inequality on \eqref{eq:lem:sigma-eq-0} and combining \eqref{eq:lem:sigma-eq-1}-\eqref{eq:lem:sigma-eq-3}, we further have
    \begin{align*}
        &\frac{1}{n}\sum_{i=1}^n (a_{i}'- a_{i})^2\\
        \lesssim& \max_{i\in[n]}\hat{\omega}_i^4\cdot\left( \frac{1}{n}\sum_{i=1}^nT_1^2 + \frac{1}{n}\sum_{i=1}^nT_2^2\right)\notag \\
        \lesssim& \frac{1}{n}\|\hat{\bTheta}_j - \bTheta_j^*\|_2^2, \notag
    \end{align*}
    where in the last inequality, we also use the boundedness of $\hat{\omega}_i$.
    This implies that
    \begin{align*}
        |\hsigma_j^2 - \sigma_j^2| 
        &\lesssim \tau_n^2 \sqrt{\frac{1}{n}\|\hat{\bTheta}_j - \bTheta_j^*\|_2^2}\\
        &\lesssim \tau_n^2 \sqrt{\frac{1}{n}\|\bX(\hat{\bb}_j - \bb_j^*)\|_2^2 + \frac{1}{n} \|\hat{\bE}_j - \bE_j^*\|_2^2}  \\
        &\lesssim \tau_n^2 r_{n,p}\\
        &= \op(1),
    \end{align*}
    where $r_{n,p}$ is defined in \Cref{thm:est-err-B}. 
    Here the concentration of $\|\bX(\hat{\bb}_j - \bb_j^*)\|_2^2$ is from \Cref{thm:est-err-B} and the one of $\|\hat{\bE}_j - \bE_j^*\|_2^2$ is from \eqref{eq:lem:rsc:eq-T2} as in the proof of \Cref{lem:rsc}.
    This implies that $\hsigma_j^2 \pto \sigma_j^2$.
    The conclusion then follows by applying the continuous mapping theorem.
\end{proof}

\begin{algorithm}[!t]\caption{Data splitting procedure.}\label{alg:split}
    \begin{algorithmic}[1]
        \REQUIRE Data $(\bx_i,\by_i)\in\RR^{d}\times\RR^{p}$ for $i=1,\ldots,2n$.
        
        \STATE Split the full data into two disjoint datasets $\cD_1=\{(\bx_i,\by_i)\mid i=1,\ldots,n\}$ and $\cD_2=\{(\bx_i,\by_i)\mid i=n+1,\ldots,2n\}$.

        \STATE Apply \Cref{alg:deconfounder} on $\cD_2$ to obtain the estimates $\hat{\bB}$ and $\hat{\bGamma}$.\label{algo:split-B-Gamma}

        \FOR{$j=1,2,\ldots,p$}
            \STATE Select a subset $I\subseteq [p]\cap\{j\}$ and set $I^c=[p]\setminus I$.

            \STATE Based on $\hat{\bB}$ and $\hat{\bGamma}$, use partial data $(\bX,\bY_{I^c})$ to estimate $\hat{\bZ}$, where $\bX=[\bx_1,\ldots \bx_n]^{\top}$, $\bY=[\by_1,\ldots \by_n]^{\top}$ and $\bY_{I^c} = [\bY_{\ell}]_{\ell\in I^c}$.\label{algo:split-Z}

            (Alternatively, Step \ref{algo:split-B-Gamma}-\ref{algo:split-Z} can be combined such that $\hat{\bB}$, $\hat{\bGamma}$, and $\hat{\bZ}$ are estimated jointly for gene $j$.)
            
            \STATE Based on $\hat{\bB}$, $\hat{\bGamma}$, and $\hat{\bZ}$, estimate $\hat{\omega}_i$'s and $\hat{\bu}$ on $(\bX,\bY_{I})$.
            
            \STATE Calculate the test statistics $z_j$ for gene $j$.
        \ENDFOR
        
        \ENSURE A set of test statistics $\{z_j\mid j=1,\ldots,p\}$.
    \end{algorithmic}
\end{algorithm}

\clearpage
\section{Computational aspects}\label{app:sec:comp}
    \subsection{Exponential family}\label{app:subsec:exp-family}
    Some commonly used exponential families, the exact formulas of the log-partition functions and other statistics, are summarized in \Cref{tab:exp-family}.
    \begin{table}[!ht]
        \footnotesize
        \everymath{\displaystyle}
        \centering
        \begin{tabularx}{0.99\textwidth}{*2{>{\centering\arraybackslash}m{1.7cm}} *2{>{\centering\arraybackslash}m{1.5cm}} *1{>{\centering\arraybackslash}m{1.cm}} *1{>{\centering\arraybackslash}m{2.1cm}} *2{>{\centering\arraybackslash}m{1.7cm}}}
            \toprule
            \textbf{Distribution}  &\textbf{Extra parameter} & \textbf{Base measure $h(y)$}  & \textbf{Sufficient statistics $T(y)$}  & \textbf{Domain $\mathrm{dom}(A(\theta))$}  & \textbf{Log-partition $A(\theta)$} &\textbf{Mean $\mu=A'(\theta)$} & \textbf{Variance $A''(\theta)$} \\[0.8em]
            \midrule\\\\[-4\medskipamount] 
            Gaussian   & variance $\sigma^2$ & ${\frac {e^{-{\frac {x^{2}}{2\sigma ^{2}}}}}{{\sqrt {2\pi }}\sigma }}$ & $\frac{y}{\sigma}$ & $\RR$ & $\frac{\theta^2}{2}$ & $\theta$ & $1$  \\[0.8em] \hline\\\\[-4\medskipamount] 
            Bernoulli & & $1$ & $y$ & $\RR$ & $\log(1+e^{\theta})$ & $\frac{1}{1+e^{-\theta}}$ & $\mu(1-\mu)$           \\[0.8em] \hline\\\\[-4\medskipamount] 
            Binomial & number of trials $m$ & $\binom{m}{y}$ & $y$ & $\RR$ & $m\log(1+e^{\theta})$ & $\frac{m}{1+e^{-\theta}}$ & $\mu\left(1-\frac{\mu}{m}\right)$ \\[0.8em] \hline\\\\[-4\medskipamount] 
            Poisson  & & $\frac{1}{y!}$  & $y$  & $\RR$ & $e^{\theta}$ & $e^{\theta}$ & $e^{\theta}$ \\[0.8em] \hline\\\\[-4\medskipamount] 
            Negative Binomial  & number of failures $\phi$ & $\binom{y+\phi-1}{y}$ &  $y$  & $\RR_-$ & $-\phi\log(1-e^\theta)$ & $\phi \frac{e^{\theta}}{1 - e^{\theta}}$ & $\phi \frac{e^{\theta}}{(1 - e^{\theta})^2}$ \\
             [0.8em]
            \bottomrule
        \end{tabularx}
      \caption{Summary of exponential family in canonical form.}\label{tab:exp-family}
    \end{table}
    
    \subsection{Optimization details}\label{app:subsec:opt}
    \paragraph{Initialization.} Our initialization procedure for optimalization problem \eqref{opt:1} is inspired by \citet{lin2023esvd}.

    \begin{itemize}
        \item Initialize the marginal effects $\bF$ by solving a generalized linear model without considering the latent variables.
        When the fitting of GLM is numerically unstable, one can also add a small ridge penalty $\lambda=10^{-5}$.

        \item Initialize $\bW$ and $\bGamma$ using the SVD of the matrix $\log(\bY+1) = \bU_Y\bSigma_Y\bV_Y^{\top}$ for Poisson likelihood or Negative Binomial likelihood with log link.
        Let $\bW=(\cP_{\bX}^{\perp}\bU_Y\bSigma_Y^{1/2})_{1:r}$ and $\bGamma=(\bV_Y\bSigma_Y^{1/2})_{1:r}$ be the first $r$ columns of the corresponding matrices.
        Here the projection $\cP_{\bX}^{\perp}$ ensures that $\bW$ is uncorrelated with $\bX$.
        In particular, when the intercept is included in the covariates, the initial value of $\bW$ also has zero means per column. 
    \end{itemize}

    To initialize variables for optimization problem \eqref{opt:2}:
    \begin{itemize}
        \item Initialize the direct effects $\bB$ as $\cP_{\hat{\bGamma}}^{\perp}\hat{\bF}$.
        \item Initialize $\bZ$ and $\bGamma$ using the SVD of the matrix $\bX\hat{\bF}^{\top}\cP_{\hat{\bGamma}}+\hat{\bW}\hat{\bGamma}^{\top} = \bU'\bSigma'\bV'^{\top}$.
        Let $\bZ = (\bU'\bSigma'^{1/2})_{1:r}$ and $\bGamma=(\bV'\bSigma'^{1/2})_{1:r}$.
        Because the latter has the same column space as $\hat{\bGamma}$, we simply treat the latter as $\hat{\bGamma}$ in optimization problem \eqref{opt:2} with a light abuse of notation.
    \end{itemize}

    \paragraph{Alternative maximization}
    The alternative maximization \Cref{alg:jmle} is used to perform nonconvex matrix factorization.
        In our setup where the objective function is convex in the natural parameter, each iteration of \Cref{alg:jmle} is simply solving two convex optimization subproblems.

    \begin{algorithm}\caption{Joint maximum likelihood estimation by alternative maximization}\label{alg:jmle}
        \begin{algorithmic}[1]
            \REQUIRE Data $\bY\in\RR^{n\times p}$ from exponential family with log-partition function $A$, the regularization parameter $\lambda$, and initial value $\bl_i^{(0)}\in\cD_{l_i}$, $\br_j^{(0)}\in\cD_{r_j}$ for $i\in[n]$ and $j\in[p]$.
            
            \STATE Initialize the iteration number $t=0$.

            \WHILE{not converged}
                \STATE $t\gets t+1$.
                
                \FOR{$i=1,2,\ldots,n$}
                    \STATE $$\bl_i^{(t)} \in \argmax_{\bl\in\cD_{l_i}} \frac{1}{p}\sum_{j=1}^p \left(y_{ij} \bl^{\top}\br_j^{(t-1)} - A(\bl^{\top}\br_j^{(t-1)})\right)$$
                \ENDFOR

                \FOR{$j=1,2,\ldots,p$}
                    \STATE $$\br_j^{(t)} \in \argmax_{\br\in\cD_{r_j}} \frac{1}{n}\sum_{i=1}^n \left(y_{ij} \bl_i^{(t)\top}\br - A(\bl_i^{(t)\top}\br)\right) - \lambda \|\br\|_1$$
                \ENDFOR

            \ENDWHILE
            
            \ENSURE $\bL=[\bl_i^{(t)}]_{i\in[n]}^{\top}$ and $\bR=[\br_j^{(t)}]_{j\in[p]}^{\top}$ with $\bL\bR^{\top}$ being the estimated natural parameters.
        \end{algorithmic}
    \end{algorithm}

    By default, we use the inexact line search algorithm with an initial step size of 0.1 and a shrinkage factor of 0.5 for each iteration.
    The search is stopped if the Armijo rule is satisfied with tolerance $10^{-4}$ or the number of iterations reaches 20.
    We early stop the alternative maximization if the objective value does not increase more than a tolerance of $10^{-4}$ for 20 iterations.

    \paragraph{Estimation of dispersion parameters}
        To estimate the dispersion parameter, we first fit GLMs on the data and obtain the estimated mean expression of gene $j$, denoted as $\hat{\mu}_j$ for $j=1,\ldots,p$.
        Note that when $y_{ij}$ comes from a Negative Binomial distribution, its variance is given by
        \[ \Var(y_{ij}\mid \theta_{ij}) = \mu\left( 1 + \alpha_j\mu\right)\]
        where $\mu=\EE [y_{ij}\mid \theta_{ij}]$ is the conditional mean while $\alpha_j$ is the dispersion parameter of the NB1 form.
        In the form of exponential family in \Cref{tab:exp-family} parameterized by the parameter $\phi_j$, $\alpha_j$ is the reciprocal of $\phi_j$, namely, $\alpha_j = 1/\phi_j$.
        By methods of moments, we can solve the following equation to obtain an estimator $\hat{\phi}_j$ for $\phi_j$:
        \[\frac{1}{n}\sum_{i=1}^n(y_{ij}-\hat{\mu}_j)^2 = \hat{\mu}_j\left( 1 + \alpha\hat{\mu}_j\right).\]
        Finally, we clip $\hat{\alpha}_j$ to be in $[10^{-2},10^2]$ and set $\hat{\phi}_j=1/\hat{\alpha}_j$.

    \subsection{Choice of hyperparameters in practice}\label{app:subsec:hyperparam}

    The main text provides theoretical results for the proposed method under certain assumptions.
    The proposed algorithm \Cref{alg:deconfounder} requires the choice of hyperparameters, such as the rank $r$, the regularization parameters $\lambda$ in optimization problem \eqref{opt:2}, and $(\tau_n,\lambda_n)$ in optimization problem \eqref{opt:min-var}.
    Although the theoretical orders of some parameters are provided for consistency and asymptotic normality, the choice of hyperparameters in practice is crucial for the performance of the proposed method. 
    Below, we discuss the choice of hyperparameters in practice.

    \paragraph{Boundedness constant $C$.}
    The boundedness constant $C$ is a reasonably large constant that ensures a finite solution to optimization problems exists.
    In our simulations, estimating the model parameters is not sensitive to the choice of boundedness constant as long as it is set to be sufficiently large; see also \citet[Appendix D]{chen2022determining} for detailed discussions.
    Therefore, in our implementation, instead of restricting the parameters to be bounded, we project the gradient at each step of the alternative maximization onto the $L_2$-norm ball with radius $2C'$ for some constant $C'$.
    A smaller value of $C'$ equals decreasing the learning rates while improving the numerical stability. 
    We set $C'$ to be $10^5$ and $10^3$ for experiments with Poisson and Negative Binomial likelihoods, respectively.

    \paragraph{Lasso penalty $\lambda$.} For the lasso penalty $\lambda = c_1 \sqrt{\log p / n}$, one can use cross-validation to tune the lasso penalty for optimal log-likelihood.
        However, because the estimation results are insensitive to the choice of this penalty, we simply set $c_1$ to be 0.02 and 0.01 for experiments with Poisson and Negative Binomial likelihoods, respectively.

    \paragraph{The number of factors $r$.}
    For the number of factors $r$, the joint-likelihood-based information criterion (JIC) proposed by \citet{chen2022determining} can be utilized to select a proper number of latent factors.
    The JIC value is the sum of deviance and a penalty on model complexity:
    \begin{align*}
        \JIC(\hat{\bTheta}^{(r)}) &= \mathrm{deviance} + \nu(n,p,d+r) \\
        &= - 2 \sum_{i\in[n],j\in[p]}\log p(y_{ij}\mid \hat{\theta}_{ij}^{(r)}) + c_{\JIC}\cdot\frac{(d+r)\log(n\wedge p)}{n\wedge p},
    \end{align*}
    where $\hat{\bTheta}^{(r)}$ is the estimated natural parameter matrix with $r$ latent factors and $d$ observed covariates, and $c_{\JIC}>0$ is a universal constant set to be one in all our simulations.
    As shown by \citet{chen2022determining}, minimizing the empirical JIC yields a consistent estimate for the number of factors in generalized linear factor models.
    As an illustration, we compute the values of JIC at different numbers of factors on simulated datasets and visualize them in \Cref{fig:JIC-r}.
    When the unmeasured confounding effects are strong, the default choice of $c_{\JIC}=1$ gives reasonable estimates for the number of factors under both Poisson and Negative Binomial likelihoods.
    Because the complexity term is a linear function in $r$, one can also inspect the increment of log-likelihood compared to the increment of the complexity term as a function of $r$, as shown in the right panel of \Cref{fig:JIC-r}.
    For real-world datasets, this can help to adjust the penalty level $c_{\JIC}$ to select a suitable value of $r$ to achieve a sufficient reduction of negative log-likelihood while avoiding overfitting.
    
    \begin{figure}[!ht]
        \centering
        \includegraphics[width=0.9\textwidth]{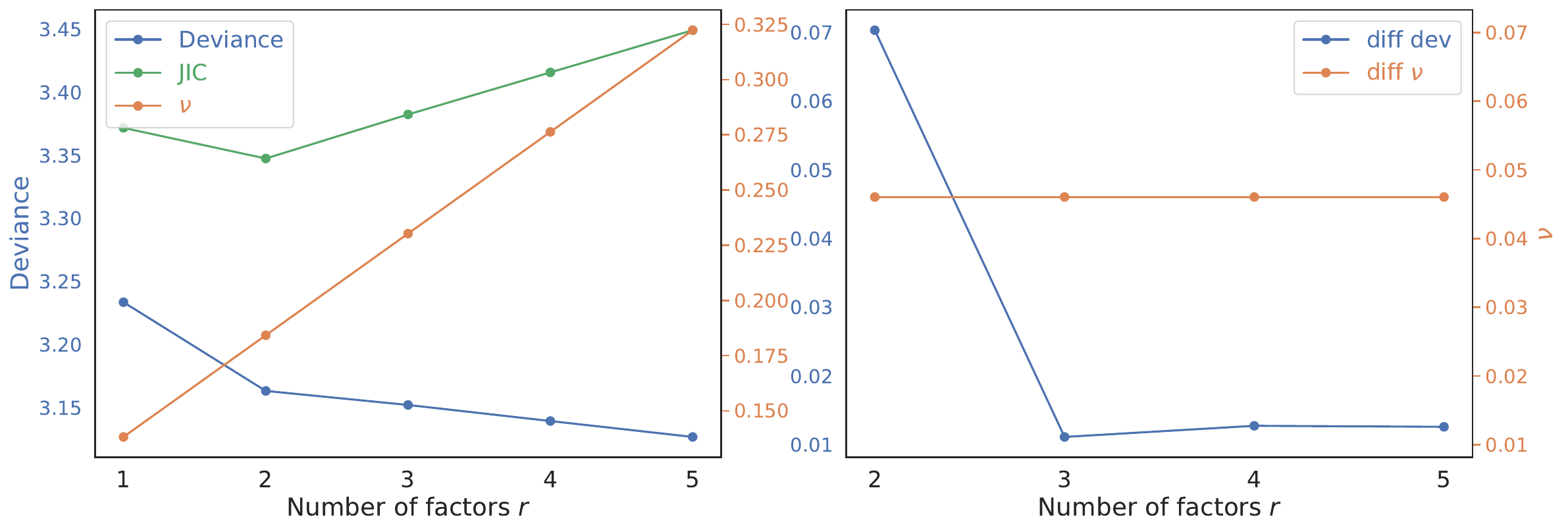}
        \caption{
        The left panel shows the deviance and the complexity penalty $\nu$ at different numbers of factors $r$.
        The JIC is the sum of deviance and $\nu$.
        The right panel shows the decrement of the deviance and the complexity at different numbers of factors $r$.
        The values are computed from one simulation in \Cref{sec:simulation} with $n=100$ and $r^*=2$ underlying factors.}
        \label{fig:JIC-r}
    \end{figure}

    \paragraph{Debiasing parameters $(\tau_n,\lambda_n)$.}
    For inference, two parameters $(\tau_n,\lambda_n)$ are to be specified.
        However, the parameter $\tau_n$ is less important because as long as both the covariate $\bx_i$ and the projection direction $\bu$ are bounded in $L_{\infty}$-norm, the second constraint in Equation (4.4) will always be satisfied.
        For this reason, we can ignore the second constraint and solve the relaxed optimization problem, similar to the implementation of \citet{cai2021statistical}.
        Therefore, one only needs to determine the parameter $\lambda_n = c_2 \sqrt{\log (n) / n}$.

        To address this, \citet{cai2021statistical} only implemented a single value for $c_2$. However, we propose a more effective heuristic method to guide the selection of $c_2$.
        Specifically, we enumerate different values of $c_2\in\{0.001,0.002,\ldots, 0.01,0.02,\ldots,0.1,0.2,\ldots,1\}$ and compute the median and median absolute deviation (MAD) of the corresponding empirical $z$-statistics.
        We then generate the scree plot of the two summarized statistics.
        As shown in \Cref{fig:median-lamn}, as $\lambda_n$ increases, both the median and the MAD of the empirical null distribution change.
        Specifically, as $\lambda_n$ increases, the median decreases, while the MAD increases and then decreases.
        Therefore, when $\lambda_n$ is too small, the empirical null distribution concentrates around 0, and the resulting tests will be conservative.
        On the other hand, when $\lambda_n$ is too large, the tests will be anti-conservative.
        Therefore, a reasonable choice for $\lambda_n$ is such that the absolute value of the median is not too large while the MAD of the corresponding test statistics is near one.
    
    For simulations with Poisson likelihood, according to the scree plot, the adaptive choice of the value $c_2$ would be the largest value that makes the median deviate from 0 by no more than a threshold of 0.1.
    Analogously, we set the median deviation threshold to be 0.025 for the Negative Binomial simulations.
    Note that any value below the selected $\lambda_n$ also provides valid inference results but with lower power.

    \begin{figure}[!ht]
        \centering
        \includegraphics[width=0.5\textwidth]{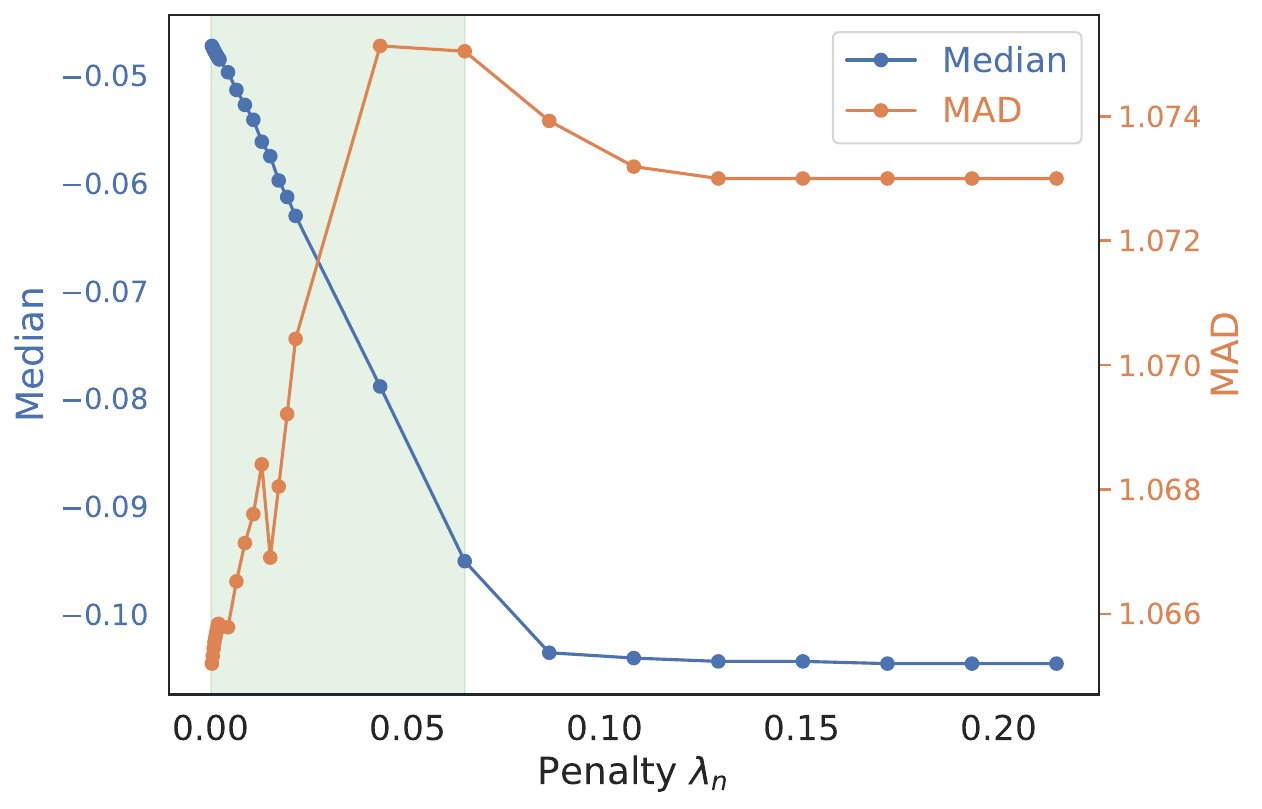}
        \caption{The median and MAD of the $z$-statistics as a function of the regularization parameter $\lambda_n$ computed from one simulation in \Cref{subsec:simu-poisson} with $n=100$ and $r=2$.
        The shaded region indicates feasible values of $\lambda_n$, for which the absolute values of the medians of the corresponding test statistics are less than 0.1.}
        \label{fig:median-lamn}
    \end{figure}

    \subsection{Negative binomial likelihood with non-canonical link}\label{app:subsec-nb-log-link}
    While theoretically nice, the canonical link function for Negative Binomial distributions (NB-C) is not recommended in general because its natural parameter value is always negative, but linear predictors ought to be unbounded in general.
    Numerical instability may occur in the boundary of the natural parameters.
    Furthermore, the NB-C model is sensitive to the initial values and may converge to a local solution.

    The common choice of a link function for generalized linear models with Negative Binomial likelihood is the log link \citep{agresti2015foundations}.
    Below, we show how to incorporate non-canonical ink into our framework.

    For Negative Binomial distribution, recall that $\phi$ is a parameter that represents the number of failures as in \Cref{tab:exp-family}.
    Define the Negative Binomial canonical link $A'^{-1}$ and log link $L^{-1}$, such that $A'(\theta) = \phi e^\theta/(1-e^\theta)$ and $L(\xi)=e^\xi$.

    Let $\theta$ and $\xi$ be the natural parameter and its representation under the log link; namely, the mean $\mu$ can be obtained from them through the corresponding link functions:
    \[\mu = A'(\theta) = L(\xi).\]
    This gives rise to the transformation equations:
    \begin{align*}
        e^{\xi} &= \phi \frac{e^{\theta}}{1 - e^{\theta}},\qquad
        e^{\theta} =  \frac{e^{\xi}}{\phi + e^{\xi}} ,
    \end{align*}
    and
    \begin{align}
        \theta &= A'^{-1}(L(\xi)) = \log \frac{e^{\xi}}{\phi + e^{\xi}}.\label{eq:nb-theta-xi}
    \end{align}
    Note that the negative log-likelihood is given by
    \begin{align*}
        l (\xi) &:= - y\cdot (A')^{-1}(L(\xi)) + A((A')^{-1}(L(\xi))) \\
        &= -y \log \frac{e^{\xi}}{\phi+e^{\xi}} - \phi \log \frac{\phi}{\phi + e^{\xi}},
    \end{align*}
    which has gradient and hessian:
    \begin{align}
        \frac{\partial l}{\partial \xi} &= \frac{\partial l}{\partial \theta} \frac{\partial \theta}{\partial \xi} = -(y - A'(\theta)) \frac{\phi}{\phi + e^{\xi} } \label{eq:nb-grad}\\
        \frac{\partial^2 l}{\partial \xi^2} &= \frac{\partial^2 l}{\partial \theta^2} \frac{\partial \theta}{\partial \xi} + \frac{\partial l}{\partial \theta} \frac{\partial^2 \theta}{\partial \xi^2} \notag\\
        &= A''(\theta) \left(\frac{\partial \theta}{\partial \xi}\right)^2+ \frac{\partial l}{\partial \theta} \frac{\partial^2 \theta}{\partial \xi^2} \notag\\
        &\approx \frac{\phi e^{\xi}}{\phi + e^{\xi}} \label{eq:nb-hess}
    \end{align}
    where the last line is because the conditional expectation on $y-A'(\theta)$ given $\theta$ is zero, $\EE[{\partial l}/{\partial \theta}\mid\theta] =0$, so that the second term is ignorable.  The latter approximation approach is also used in classic GLM to derive the asymptotic variance of the estimates.

    For $n$ i.i.d. samples $(\bx_i,\bz_i,\by_i)$, the linear predictor reads that $\bXi = \bX\bB^{\top}+ \bZ\bGamma^{\top}$ when using the log link.
    Based on the relationship \eqref{eq:nb-theta-xi}, we can perform estimation and inference for the log link function, as described below.
    
    For estimation, the objective function \eqref{eq:loglikelihood} now becomes:
    \begin{align*}
        l(\bXi) =l(\bGamma,\bZ,\bB) &= -\frac{1}{n}\sum_{i\in[n]}\sum_{j\in[p]} (y_{ij} A'^{-1}(L(\xi_{ij})) -  L(\xi_{ij}) ).
    \end{align*}
    Even though the new objective is now nonconvex in the parameter $\bXi$, the alternative maximization algorithm \Cref{alg:jmle} is still applicable to it, because the gradient can be computed based on \eqref{eq:nb-grad}.
    If we initialize $\hat{\bF}$ from GLM estimates and treat it as fixed, then solving optimization problem \eqref{opt:3} reduces to a nonconvex matrix factorization problem.
    Under this setting, there is a rich literature on establishing the estimation error for $\bW^*$ and $\bGamma^*$ given that the initial value is close to the truth; see \citet{wang2017unified,lin2023esvd} among the others.
    In other words, we may also obtain error bounds on $\|\bXi^* - \hat{\bXi}\|_{\fro}^2$ by imposing additional conditions.

    For inference, we simply apply the chain rule and \eqref{eq:nb-grad}-\eqref{eq:nb-hess} to rewrite \eqref{eq:debias} as:
    \begin{align*}
        \hat{b}_{j1}^{\de} &= \hat{b}_{j1} + {\bu}^{\top}\frac{1}{n}\sum_{i=1}^n  \bx_{i} (\by_{i} - A'(\hat{\btheta}_{i}))^{\top} \diag\left(\left\{\frac{\partial \theta}{\partial \xi}\Big|_{\xi=\hat{\xi}_{ij}}\right\}_{j\in[p]}\right) \bv_i,
    \end{align*}
    with
    \begin{align*}
        \bv_i &= {\omega}_{i}\diag\left(\left\{\frac{\partial \theta}{\partial \xi}\Big|_{\xi=\hat{\xi}_{ij}}\right\}_{j\in[p]}\right) ^{-2}  \diag(A''(\hat{\btheta}_i))^{-1}\cP_{\hat{\bGamma}}^{\perp}\be_j,\\
        \omega_i &= \EE\left[\frac{\partial^2 l}{\partial \xi^2} \,\Big|\, \xi=\hat{\xi}_{ij}\right] = A''(\hat{\theta}_{ij}) \left(\frac{\partial \theta}{\partial \xi}\Big|_{\xi=\hat{\xi}_{ij}}\right)^2 = \frac{\phi e^{\xi}}{\phi + e^{\xi}}.
    \end{align*}
    Because when $\cR_C$ is bounded, the derivative function $\partial\theta/\partial\xi$ is Lipschitz continuous, the estimation error of it:
    \[\sum_{j=1}^p(\partial\theta/\partial\xi|_{\xi=\hat{\xi_{ij}^*}} - \partial\theta/\partial\xi|_{\xi=\hat{{\xi}_{ij}}} )^2 \lesssim \|\bxi_i^* - \hat{\bxi}_i\|_2^2\]
    can be bounded if $\bxi_i^*$ can be well estimated.
    Similarly, the estimation error of $\bTheta^*$ can be controlled because $\theta^*_{ij}$ is a Lipschitz continuous function of $\xi^*_{ij}$.
    Thus, \Cref{thm:normality} also applies if
    \begin{align}
        \|\bXi^* - \hat{\bXi}\|_{\fro}^2 &\lesssim \sqrt{n\vee p} \label{eq:nb-cond-1}\\
        \max_{1\leq j\leq p}\frac{1}{\sqrt{n}}\|\bZ^*\bGamma^{*\top}_j - \hat{\bZ}\hat{\bGamma}^{\top}_j\|_{\fro} &\lesssim \frac{1}{\sqrt{n\wedge p}}   \label{eq:nb-cond-2}\\
        \|\hat{\bB}-\bB^*\|_{1,1} &\lesssim \frac{\sqrt{sd}}{\sqrt{n\wedge p}} \label{eq:nb-cond-3}\\
        \|\hat{\bB}-\bB^*\|_{\fro} &\lesssim \frac{1}{\sqrt{n\wedge p}} , \label{eq:nb-cond-4}
    \end{align}
    where \eqref{eq:nb-cond-1} requires an analysis tool from nonconvex matrix factorization, \eqref{eq:nb-cond-2} is a direct consequence from \eqref{eq:nb-cond-1} similar to \Cref{cor:est-confound-col}, and \eqref{eq:nb-cond-3}-\eqref{eq:nb-cond-4} requires non-asymptotic analysis as in the proof of \Cref{thm:est-err-B} but for nonconvex objectives instead.

\clearpage
\section{Extra experiment results}\label{app:sec:extra-ex-results}

    \subsection{Efficiency loss of sample splitting}\label{app:subsec:simu-sample-split}

    To evaluate the efficiency loss caused by sample splitting described in \Cref{alg:split}, we conduct the experiments with different splitting proportions and compare their results.
    To apply \Cref{alg:split}, we split the $p$ genes into 2 groups with equal sizes, so that $I_1=\{1,\ldots,p/2\}$ and $I_2=\{p/2+1,\ldots,p\}$.
    For each of the groups $I$, the optimization is jointly conducted based on $\bX$, $\bY_{I^c}$ and $\cD_2$, and the inference is conducted for genes in $I$.  
    As summarized in \Cref{tab:split}, the performance on Type-I error and FDP control is similar across different splitting ratios.
    However, the power and precision are affected when the ratio of observations reserved for inference is too small.
    This suggests that one should leave more observations to conduct the debias step.
    Lastly, we see similar performance even without sample splitting, suggesting that the validity of the inferential procedure could be true even without sample splitting.

    \begin{table}[!ht]
        \centering\footnotesize
        \begin{tabular}{lrrrr}
        \toprule
        \textbf{ratio split} &  \textbf{type-I error} &  \textbf{FDP} &  \textbf{power} &  \textbf{precision} \\\midrule
        0.2         &      0.050 & 0.200 &  0.454 &      0.610 \\
        0.4         &      0.049 & 0.193 &  0.755 &      0.920 \\
        0.6         &      0.050 & 0.191 &  0.901 &      1.000 \\
        0.8         &      0.051 & 0.195 &  0.963 &      1.000 \\
        no splitting           &      0.051 & 0.219 &  0.987 &      1.000 \\
        \bottomrule
        \end{tabular}
        \caption{
        Performance with varying ratios of observations reserved for inference, under the same data setup in \Cref{subsec:simu-poisson} with $n=250$ and $r=2$.
        The values are medians over 100 simulated datasets.}
        \label{tab:split}
    \end{table}

    \subsection{The blessing of dimensionality}\label{subsec:est-err}
    
    \begin{figure}[!ht]
        \centering
        \includegraphics[width=0.7\linewidth]{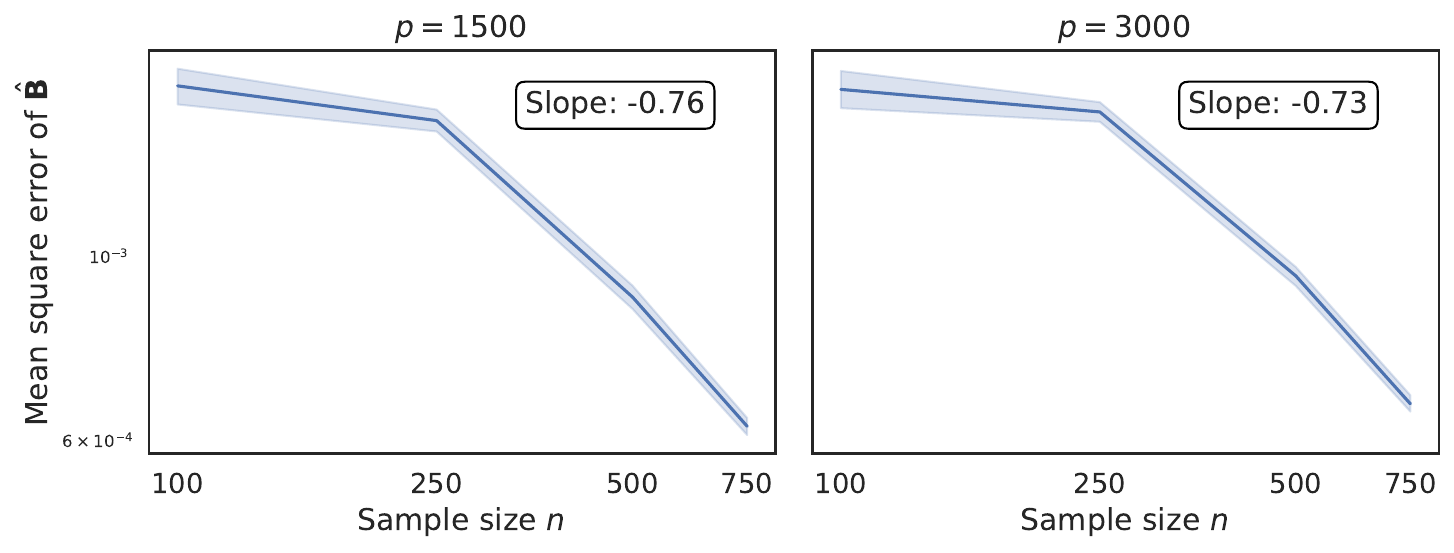}
        \caption{The mean square error of $\hat{\bB}$ with varying outcome dimension $p$ and sample size $n$, displayed on the log-log scale.
        When the outcome dimension $p$ is sufficiently large (not growing exponentially in $n$), the estimation error of $\bB$ is mainly driven by the sample size $n$.
        The slope is estimated using sample sizes larger than 100.
        The data generating process is given in \Cref{subsec:simu-poisson}.}
        \label{fig:est-err}
    \end{figure}

    \clearpage
    \subsection{Information about lupus data}\label{app:subsec-lupus}
    \begin{table}[!ht]\footnotesize
            \centering
            \begin{tabular}{cccc}
            \toprule
            Cell type & Number of samples $n$ & Number of genes $p$ & Proportion of non-zeros \\
            \midrule
            T4 & 256 & 1255 & 0.398 \\
            cM & 256 & 1208 & 0.434 \\
            B & 254 & 1269 & 0.417 \\
            T8 & 256 & 1281 & 0.471 \\
            NK & 256 & 1178 & 0.385 \\
            \bottomrule
            \end{tabular}
            \caption{Summary statistics of the preprocessed lupus datasets in each cell type. The last column represents the proportion of non-zero count in the gene expression matrix.}
            \label{tab:my_label}
        \end{table}

        \begin{figure}[!ht]
            \centering
            \includegraphics[width=0.8\textwidth]{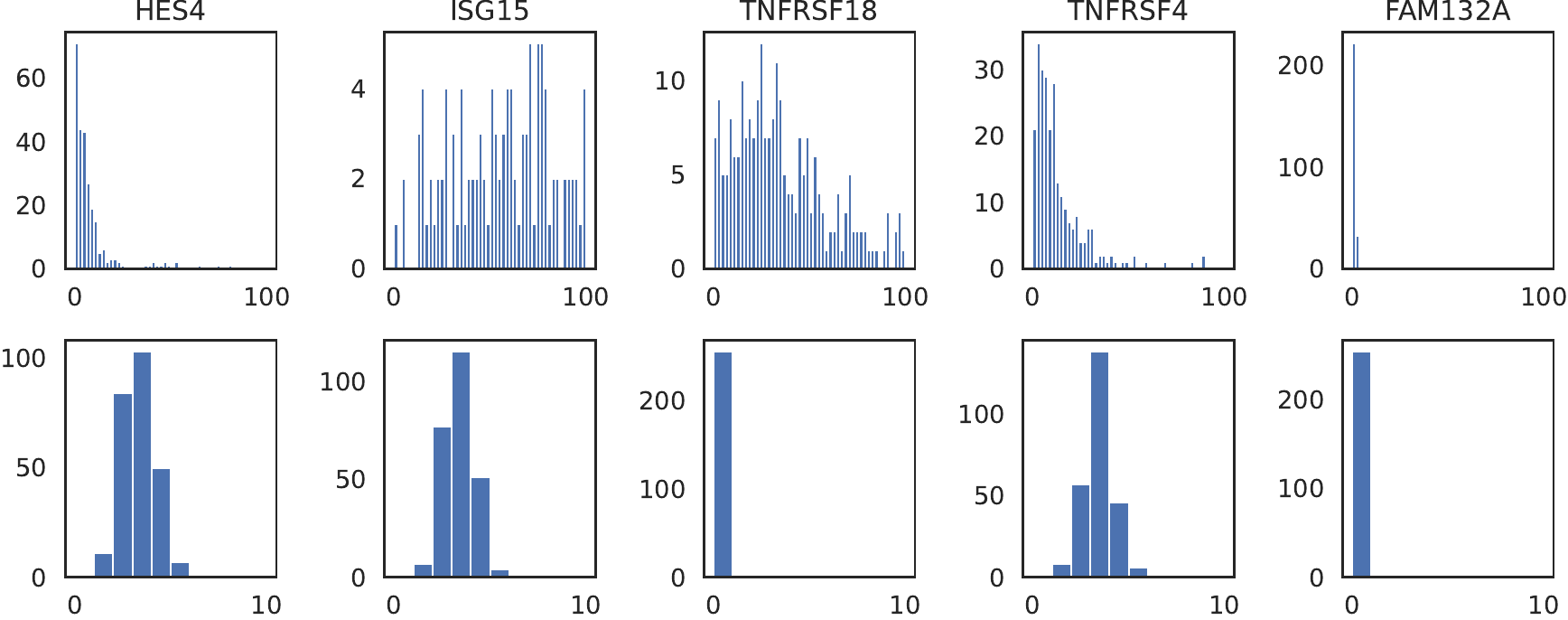}
            \caption{Histograms of expressions of 5 genes on the T4 cell type.
            The first row shows the raw pseudo-bulk counts and the second row shows the counts after library size normalization and log1p transformation, which is used for \textsc{cate}.
            Due to the sparsity of the gene expressions, some genes are not distributed like normal after transformation.}\label{fig:lupus-hist-count}
    \end{figure}

    \subsection{Extra results on lupus datasets}\label{app:subsec-extra}

    \subsubsection{Sensitivity analysis for the number of latent factors}\label{subsubsec:sensitivity}
    We inspect the sensitivity of \textsc{gcate}-subset to the number of latent factors $r$.
    By utilizing JIC \eqref{eq:JIC}, we have selected $r=7$ for the T4 cell type, which is close to the number of major covariates we drop.
    In \Cref{tab:summary-zscores} and \Cref{tab:summary-zscores-full}, we examine the performance of \textsc{gcate}-subset and \textsc{gcate}-full for different values of $r$.
    Remarkably, the resulting distributions of $z$-statistics generated by \textsc{gcate}, across varying numbers of factors $r$, are similar to the standard normal distribution when $r\geq3$ because the MAD is close to one.
    Thus, JIC can serve as a valuable criterion for determining the appropriate number of latent factors for \textsc{gcate}.
    Furthermore, it is noteworthy that the number of discoveries remains consistent when $r$ falls within a reasonable range.
    These observations collectively suggest the stability of \textsc{gcate}'s inferential outcomes within this range of reasonable factor selections.

    \begin{table}[!ht]
        \centering\footnotesize
        \begin{tabular}{rrrrrrr}
        \toprule
         \textbf{$r$} &   \textbf{mean} &  \textbf{median} &  \textbf{mad} &  \textbf{num\_sig} &  \textbf{deviance} &     \textbf{JIC} \\
            \midrule
            1 & 0.348 & -0.018 & 1.551 & 302 & 3.578 & 3.600 \\
            2 & 0.361 & -0.040 & 1.539 & 313 & 3.565 & 3.592 \\
            3 & 0.141 & -0.001 & 1.145 & 57 & 3.553 & 3.586 \\
            4 & 0.140 & 0.063 & 1.087 & 37 & 3.546 & 3.584 \\
            5 & 0.130 & 0.043 & 1.067 & 33 & 3.539 & 3.582 \\
            6 & 0.140 & 0.051 & 1.084 & 39 & 3.532 & 3.581 \\
            7 & 0.128 & 0.057 & 1.033 & 22 & 3.526 & 3.580 \\
            8 & 0.139 & 0.069 & 1.044 & 18 & 3.521 & 3.580 \\
            9 & 0.143 & 0.073 & 1.032 & 18 & 3.516 & 3.581 \\
            10 & 0.155 & 0.095 & 1.038 & 20 & 3.513 & 3.583 \\
        \bottomrule
        \end{tabular}
        \caption{The summary of the $z$-statistics and model fitness for a varying number of latent factors $r$ for \textsc{gcate}-subset analysis.
        The metrics include the mean, median, median absolute deviation (mad), and the total number of significant genes of $q$-value less than 0.2.
        The last two columns show the deviance (2 times the negative log-likelihood) and the JIC model selection criteria \eqref{eq:JIC} with $c_{\JIC}=0.25$.}
        \label{tab:summary-zscores}
    \end{table}

    \begin{table}[!ht]
        \centering\footnotesize
        \begin{tabular}{rrrrrrr}
        \toprule
         \textbf{$r$} &   \textbf{mean} &  \textbf{median} &  \textbf{mad} &  \textbf{num\_sig} &  \textbf{deviance} &     \textbf{JIC} \\
            \midrule
            1 & 0.242 & -0.014 & 1.350 & 200 & 3.552 & 3.650 \\
            2 & 0.163 & 0.013 & 1.192 & 93 & 3.542 & 3.650 \\
            3 & 0.108 & 0.010 & 1.111 & 21 & 3.534 & 3.653 \\
            4 & 0.143 & 0.052 & 1.119 & 20 & 3.528 & 3.658 \\
            5 & 0.151 & 0.066 & 1.071 & 23 & 3.523 & 3.664 \\
            6 & 0.165 & 0.071 & 1.119 & 24 & 3.518 & 3.670 \\
            7 & 0.174 & 0.077 & 1.104 & 29 & 3.514 & 3.676 \\
            8 & 0.170 & 0.092 & 1.067 & 25 & 3.510 & 3.683 \\
            9 & 0.170 & 0.100 & 1.070 & 38 & 3.507 & 3.691 \\
            10 & 0.178 & 0.103 & 1.075 & 38 & 3.502 & 3.697 \\
        \bottomrule
        \end{tabular}
        \caption{The summary of the $z$-statistics and model fitness for a varying number of latent factors $r$ for \textsc{gcate}-subset analysis.
        The metrics include the mean, median, median absolute deviation (mad), and the total number of significant genes of $q$-value less than 0.2.
        The last two columns show the deviance (2 times the negative log-likelihood) and the JIC model selection criteria \eqref{eq:JIC} with $c_{\JIC}=0.5$.}
        \label{tab:summary-zscores-full}
    \end{table}

    \clearpage
    \subsubsection{Selection of hyperparameters}
        \begin{figure}[!ht]
            \centering
            \includegraphics[width=0.9\textwidth]{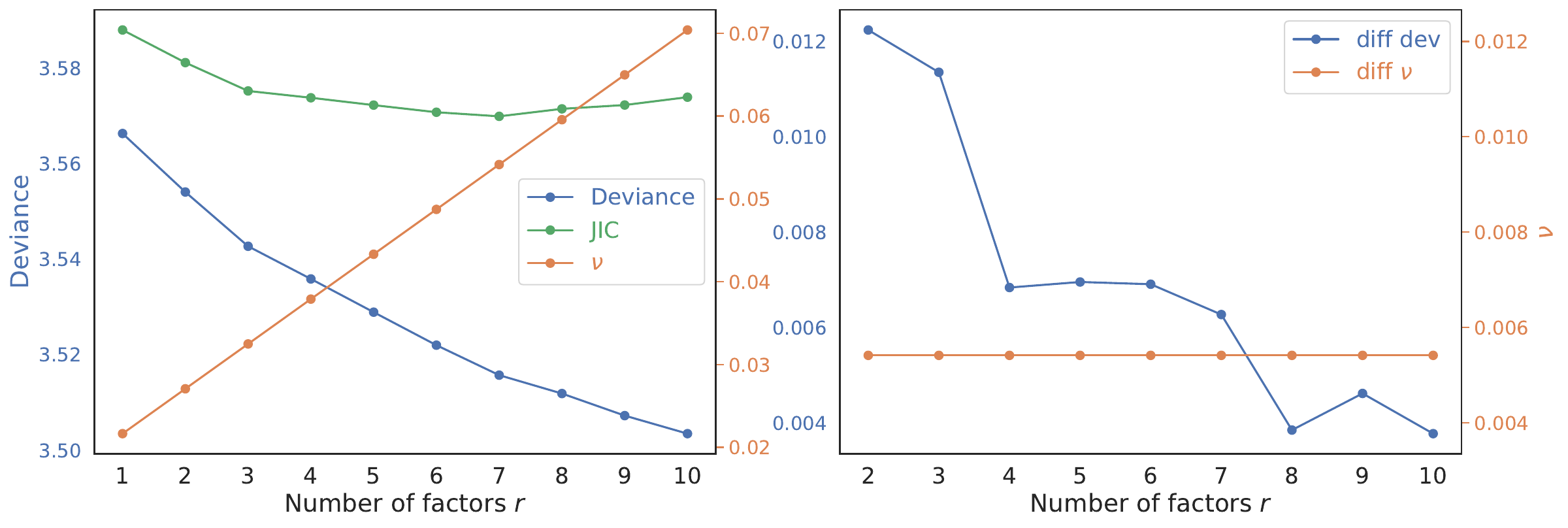}
            \includegraphics[width=0.9\textwidth]{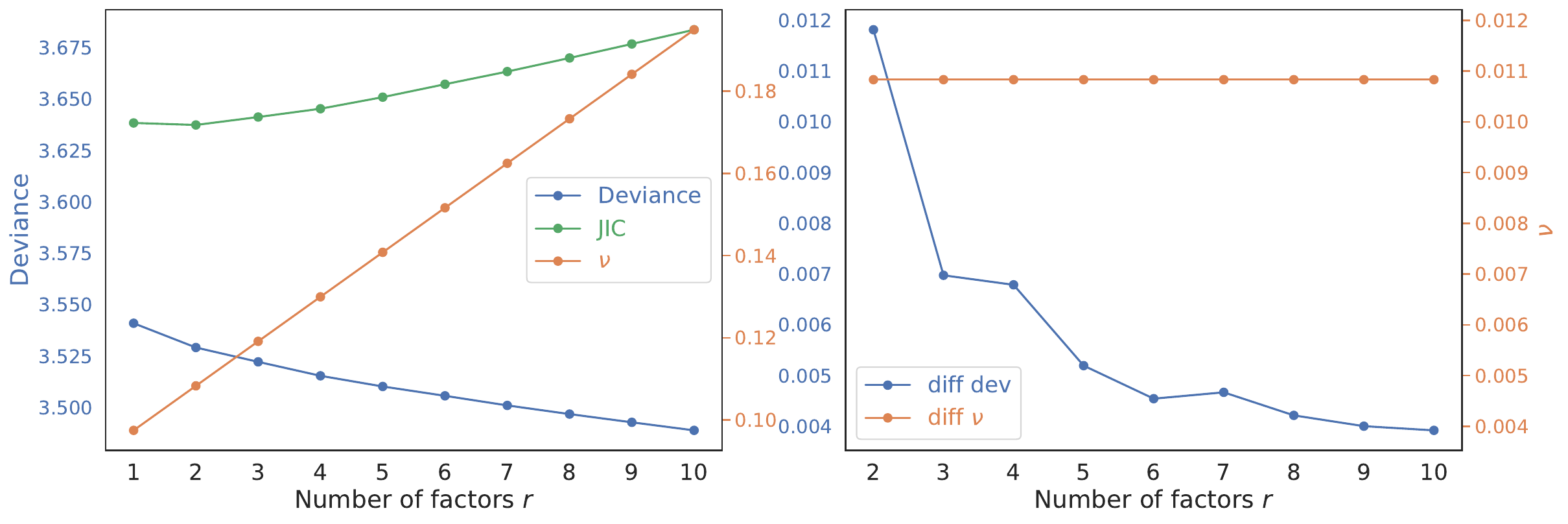}
            \caption{
            The first and second rows show the results for \textsc{gcate}-subset and \textsc{gcate}-full, respectively.
            The right panel shows the deviance and the complexity penalty $\nu$ at different numbers of factors $r$, computed on the T4 cell type of the Lupus dataset.
            The JIC is computed with $c_{\JIC}=0.25$ and 0.5, respectively.
            The right panel shows the decrement of the deviance and the complexity at different numbers of factors $r$.
            }
            \label{fig:JIC-r-lupus-T4}
        \end{figure}

        \begin{figure}[!ht]
            \centering
            \includegraphics[width=0.49\textwidth]{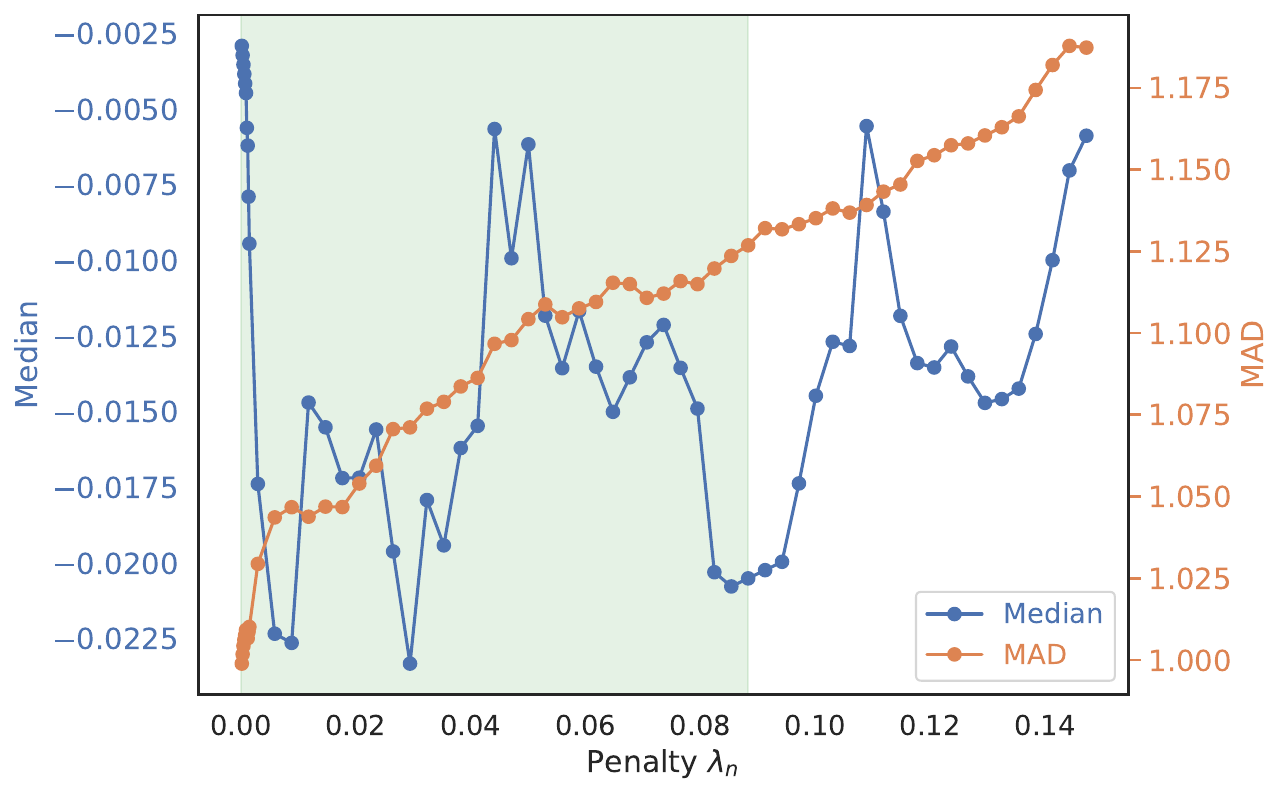}
            \includegraphics[width=0.49\textwidth]{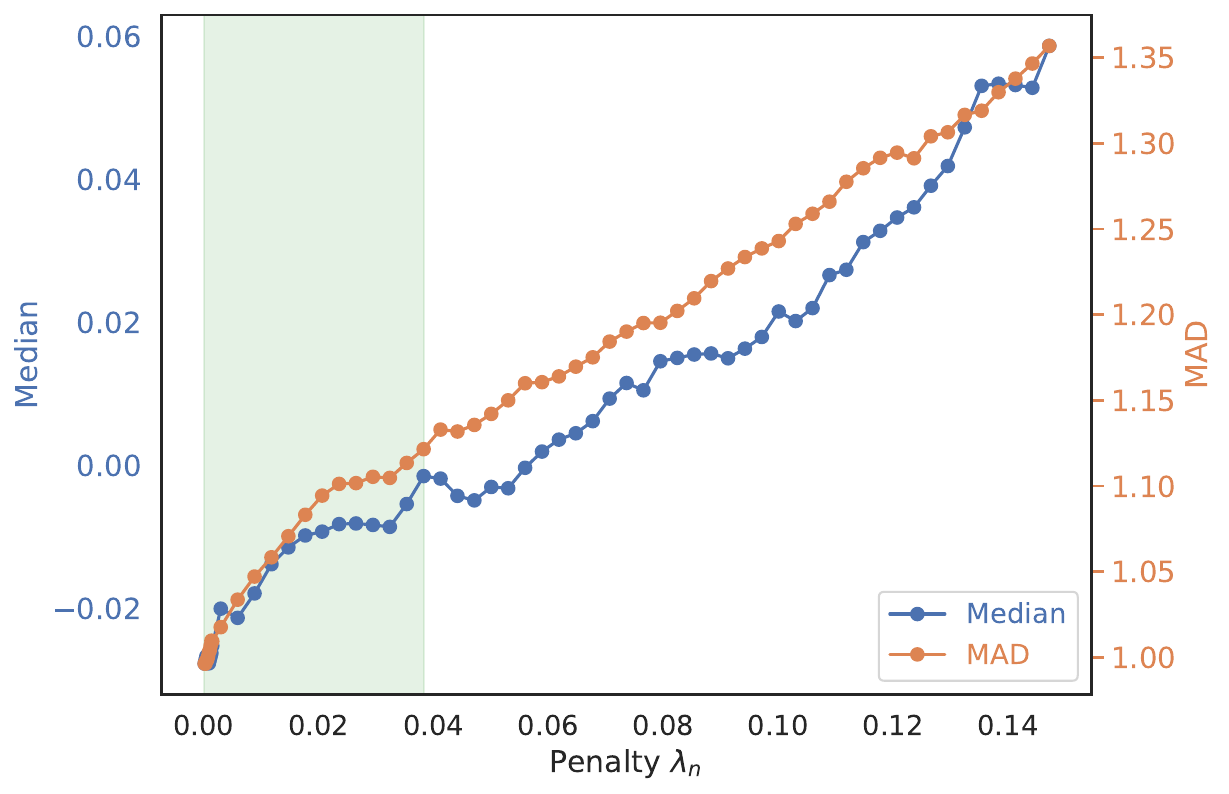}
            \caption{
            The median and MAD of the $z$-statistics as a function of the regularization parameter $\lambda_n$ computed from the T4 cell type of the Lupus dataset for \textsc{gcate}-subset and \textsc{gcate}-full analyses, respectively.
            The shaded region indicates feasible values of $\lambda_n$, for which the MADs of the corresponding test statistics are less than 1.13.}
            \label{fig:lam-lupus-T4}
        \end{figure}

    \clearpage
    \subsubsection{Results on all cell types}
    \begin{figure}[!ht]
            \centering
            \includegraphics[width=0.8\textwidth]{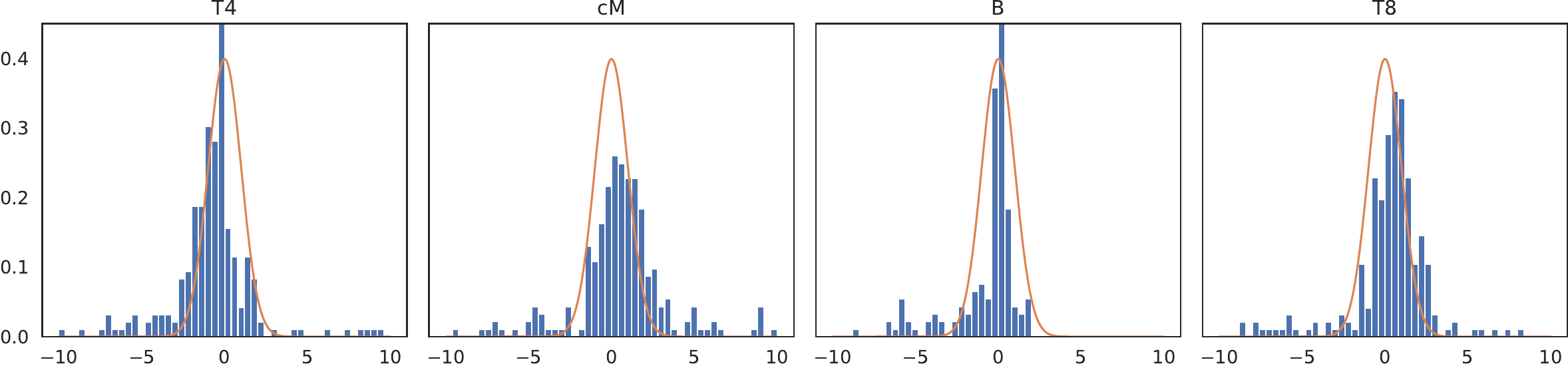}            
            \caption{Histograms of lupus $z$-statistics of \textsc{cate} on T4, cM, B, and T8 cell types, when restricted to the top 250 highly variable genes.
            The preprocessing procedure is as described in \Cref{sec:case-studeis}, but with genes expressed less than 5 subjects excluded.
            The result on the NK cell type is not included because the fitting of \textsc{cate} fails due to sparsity of the gene expressions. }\label{fig:lupus-cate-250}
    \end{figure}

    \begin{figure}[!ht]
            \centering
            \includegraphics[width=0.7\textwidth]{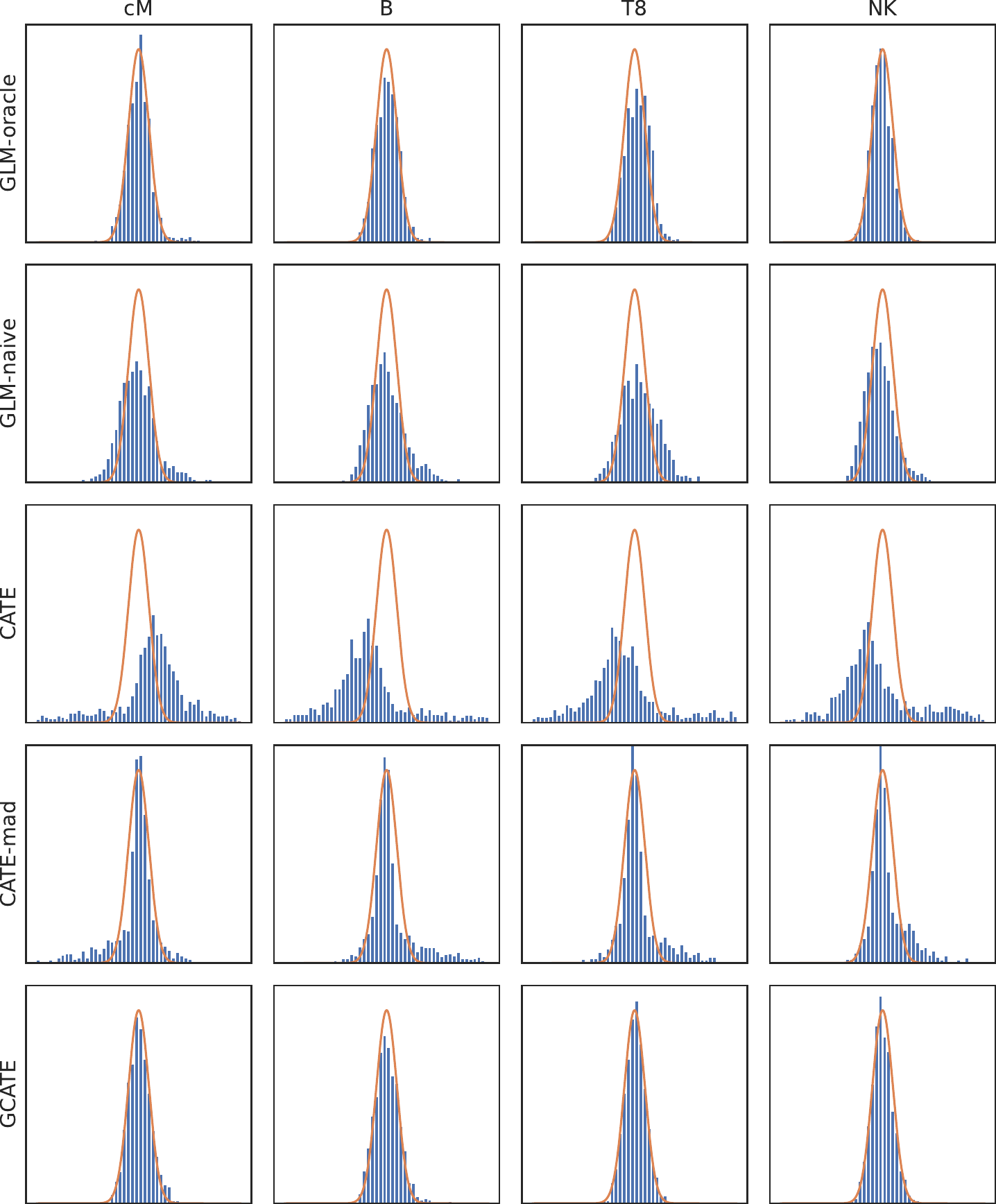}
            \caption{Histograms of lupus $z$-statistics of different methods on cM, B, T8 and NK cell types.}\label{fig:lupus-zscores-others}
    \end{figure}

    \begin{figure}[!ht]
        \centering
        \includegraphics[width=0.4\linewidth]{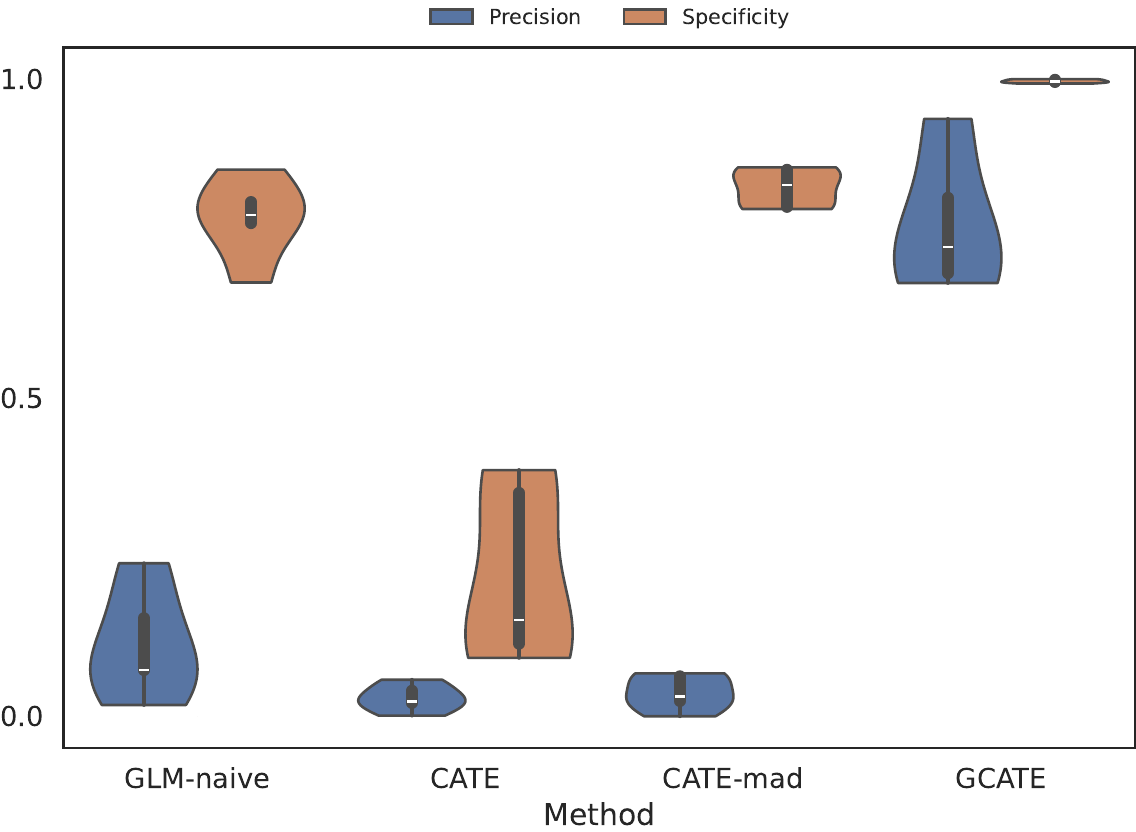}
        \caption{
            The precision and specificity for four methods computed across 5 major cell types on the lupus datasets.}\label{fig:lupus-violin}
    \end{figure}

    \begin{figure}[!ht]
        \centering
        \includegraphics[width=0.6\textwidth]{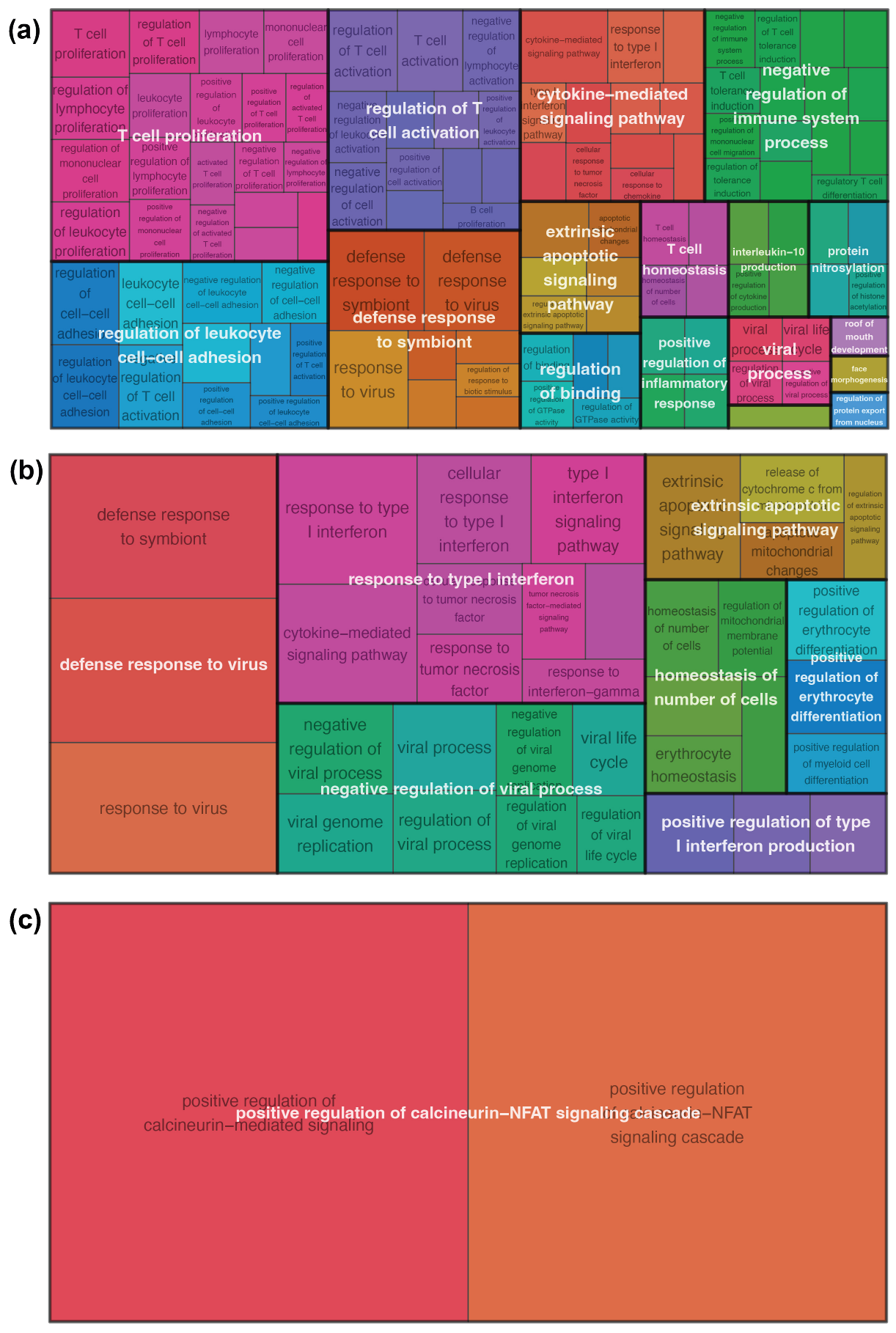}
      \caption{The treemap plot produced by \texttt{rrvgo} \citep{sayols2023rrvgo} of GO enrichment analysis results on \textbf{(a)} significant genes by the \textsc{glm}-oracle method; \textbf{(b)} significant genes by both the \textsc{glm}-oracle and \textsc{gcate} methods; and \textbf{(c)} significant genes by the \textsc{cate}-mad method but not the \textsc{glm}-oracle method.
      }\label{fig:treemap}
    \end{figure}

    \begin{figure}
        \centering
        \includegraphics[width=0.8\linewidth]{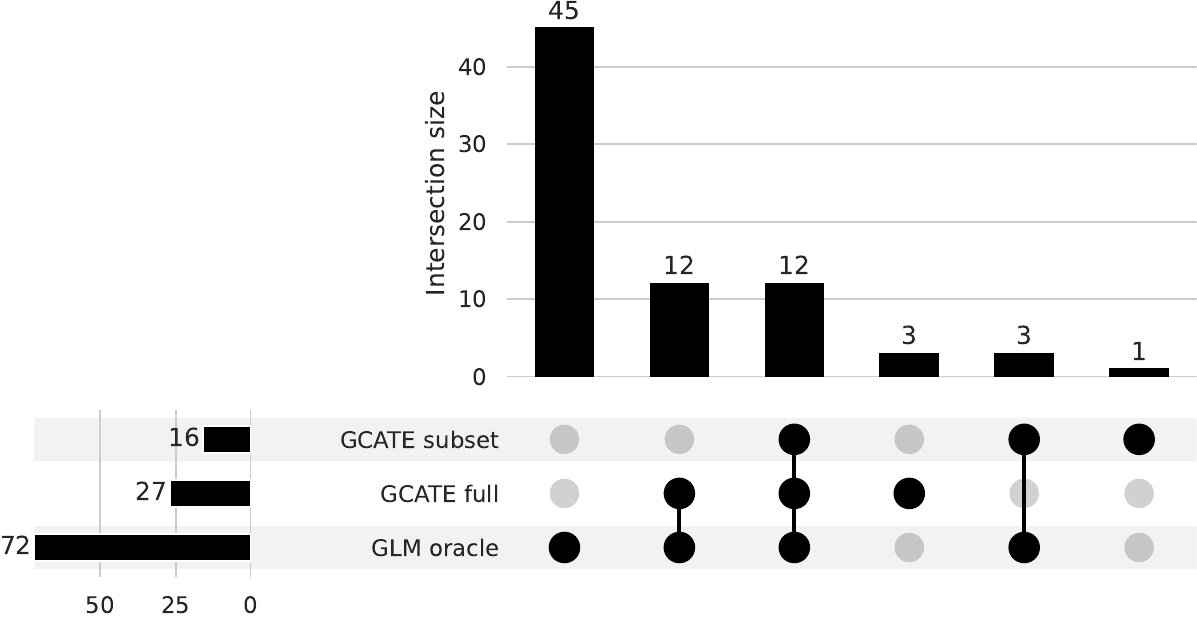}
        \caption{Upset plot of the number of discoveries of \textsc{gcate} (subset), \textsc{gcate} (full) and \textsc{glm}-oracle, with q-value cutoff $0.2$.
        Here, ``subset'' and ``full'' indicate whether all of the measured covariates are used by the corresponding methods.
        }
        \label{fig:lupus-T4-upset}
    \end{figure}

    \begin{figure}[!ht]
        \centering
        \includegraphics[width=0.7\linewidth]{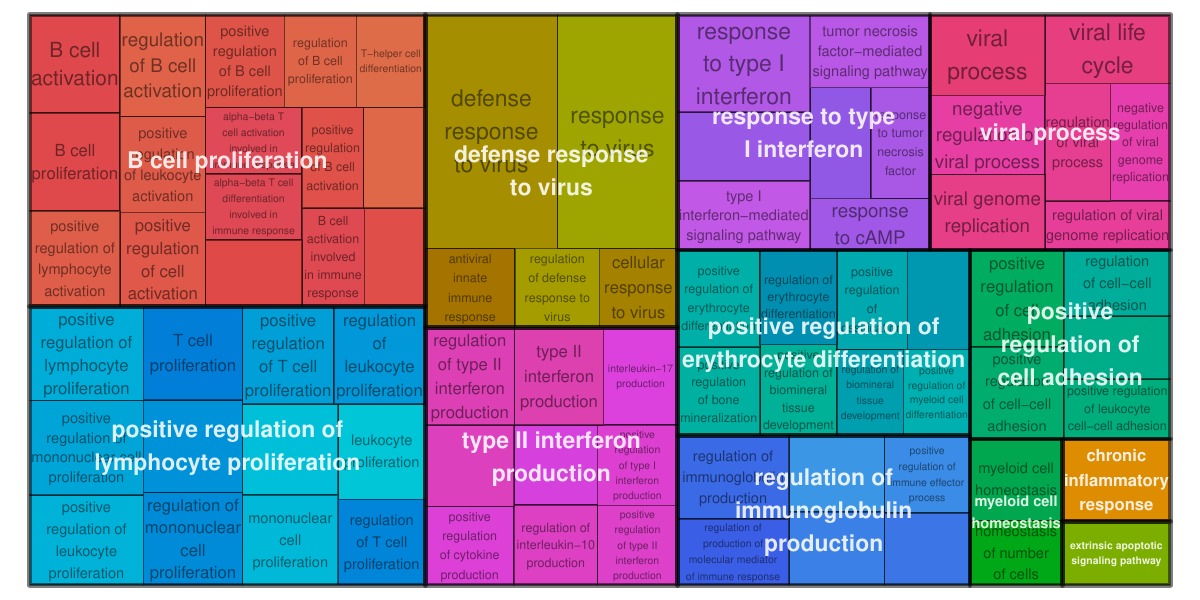}
        \caption{The treemap plot produced by \texttt{rrvgo} \citep{sayols2023rrvgo} of GO enrichment analysis results on 24 significant genes by both the \textsc{glm} and \textsc{gcate} methods with all covariates included.}
        \label{fig:treemap-full}
    \end{figure}

\clearpage
% \putbib[ref]
% \end{bibunit}


\begin{thebibliography}{}

    \bibitem[Agresti, 2015]{agresti2015foundations}
    Agresti, A. (2015).
    \newblock {\em Foundations of linear and generalized linear models}.
    \newblock John Wiley \& Sons.
    
    \bibitem[Bai, 2003]{bai2003inferential}
    Bai, J. (2003).
    \newblock Inferential theory for factor models of large dimensions.
    \newblock {\em Econometrica}, 71(1):135--171.
    
    \bibitem[Bai and Li, 2012]{bai2012statistical}
    Bai, J. and Li, K. (2012).
    \newblock Statistical analysis of factor models of high dimension.
    \newblock {\em The Annals of Statistics}, pages 436--465.
    
    \bibitem[Bandeira and van Handel, 2016]{bandeira2016sharp}
    Bandeira, A.~S. and van Handel, R. (2016).
    \newblock Sharp nonasymptotic bounds on the norm of random matrices with
      independent entries.
    \newblock {\em The Annals of Probability}, 44(4):2479--2506.
    
    \bibitem[Barry et~al., 2024]{barry2023exponential}
    Barry, T., Mason, K., Roeder, K., and Katsevich, E. (2024).
    \newblock Robust differential expression testing for single-cell crispr screens
      at low multiplicity of infection.
    \newblock {\em Genome biology}, 25(1):124.
    
    \bibitem[Bing et~al., 2023]{bing2022inference}
    Bing, X., Cheng, W., Feng, H., and Ning, Y. (2023).
    \newblock Inference in high-dimensional multivariate response regression with
      hidden variables.
    \newblock {\em Journal of the American Statistical Association}, pages 1--12.
    
    \bibitem[Bing et~al., 2022]{bing2022adaptive}
    Bing, X., Ning, Y., and Xu, Y. (2022).
    \newblock Adaptive estimation in multivariate response regression with hidden
      variables.
    \newblock {\em The Annals of Statistics}, 50(2):640--672.
    
    \bibitem[Cai et~al., 2021]{cai2021statistical}
    Cai, T.~T., Guo, Z., and Ma, R. (2021).
    \newblock Statistical inference for high-dimensional generalized linear models
      with binary outcomes.
    \newblock {\em Journal of the American Statistical Association}, pages 1--14.
    
    \bibitem[{\'C}evid et~al., 2020]{cevid2020spectral}
    {\'C}evid, D., B{\"u}hlmann, P., and Meinshausen, N. (2020).
    \newblock Spectral deconfounding via perturbed sparse linear models.
    \newblock {\em The Journal of Machine Learning Research}, 21(1):9442--9482.
    
    \bibitem[Chen and Li, 2022]{chen2022determining}
    Chen, Y. and Li, X. (2022).
    \newblock Determining the number of factors in high-dimensional generalized
      latent factor models.
    \newblock {\em Biometrika}, 109(3):769--782.
    
    \bibitem[Chen et~al., 2019]{chen2019joint}
    Chen, Y., Li, X., and Zhang, S. (2019).
    \newblock Joint maximum likelihood estimation for high-dimensional exploratory
      item factor analysis.
    \newblock {\em Psychometrika}, 84:124--146.
    
    \bibitem[Chen et~al., 2020]{chen2020structured}
    Chen, Y., Li, X., and Zhang, S. (2020).
    \newblock Structured latent factor analysis for large-scale data:
      Identifiability, estimability, and their implications.
    \newblock {\em Journal of the American Statistical Association},
      115(532):1756--1770.
    
    \bibitem[Cheng et~al., 2023]{Cheng:2023}
    Cheng, J., Lin, G., Wang, T., Wang, Y., Guo, W., Liao, J., Yang, P., Chen, J.,
      Shao, X., Lu, X., Zhu, L., Wang, Y., and Fan, X. (2023).
    \newblock Massively parallel {CRISPR}-based genetic perturbation screening at
      single-cell resolution.
    \newblock {\em Adv Sci (Weinh)}, 10(4):e2204484.
    
    \bibitem[Chernozhukov et~al., 2018a]{chernozhukov2018double}
    Chernozhukov, V., Chetverikov, D., Demirer, M., Duflo, E., Hansen, C., Newey,
      W., and Robins, J. (2018a).
    \newblock Double/debiased machine learning for treatment and structural
      parameters.
    \newblock {\em The Econometrics Journal}, 21(1):C1--C68.
    
    \bibitem[Chernozhukov et~al., 2013]{chernozhukov2013gaussian}
    Chernozhukov, V., Chetverikov, D., and Kato, K. (2013).
    \newblock Gaussian approximations and multiplier bootstrap for maxima of sums
      of high-dimensional random vectors.
    \newblock {\em The Annals of Statistics}, 41(6):2786--2819.
    
    \bibitem[Chernozhukov et~al., 2017]{chernozhukov2017lava}
    Chernozhukov, V., Hansen, C., and Liao, Y. (2017).
    \newblock A lava attack on the recovery of sums of dense and sparse signals.
    \newblock {\em The Annals of Statistics}, 45(1):39--76.
    
    \bibitem[Chernozhukov et~al., 2018b]{chernozhukov2018plug}
    Chernozhukov, V., Nekipelov, D., Semenova, V., and Syrgkanis, V. (2018b).
    \newblock Plug-in regularized estimation of high-dimensional parameters in
      nonlinear semiparametric models.
    \newblock {\em arXiv preprint arXiv:1806.04823}.
    
    \bibitem[Dai et~al., 2023]{dai2023scale}
    Dai, C., Lin, B., Xing, X., and Liu, J.~S. (2023).
    \newblock A scale-free approach for false discovery rate control in generalized
      linear models.
    \newblock {\em Journal of the American Statistical Association}, pages 1--15.
    
    \bibitem[Dixit et~al., 2016]{Dixit2016}
    Dixit, A., Parnas, O., Li, B., Chen, J., Fulco, C.~P., Jerby-Arnon, L.,
      Marjanovic, N.~D., Dionne, D., Burks, T., Raychowdhury, R., et~al. (2016).
    \newblock {Perturb-Seq: dissecting molecular circuits with scalable single-cell
      RNA profiling of pooled genetic screens}.
    \newblock {\em cell}, 167(7):1853--1866.
    
    \bibitem[Editorial, 2023]{review:2023}
    Editorial (2023).
    \newblock A focus on single-cell omics.
    \newblock {\em Nat Rev Genet}, 24(8):485.
    
    \bibitem[Feng, 2020]{feng2020causal}
    Feng, Y. (2020).
    \newblock Causal inference in possibly nonlinear factor models.
    \newblock {\em arXiv preprint arXiv:2008.13651}.
    
    \bibitem[Foster and Syrgkanis, 2023]{foster2019orthogonal}
    Foster, D.~J. and Syrgkanis, V. (2023).
    \newblock Orthogonal statistical learning.
    \newblock {\em The Annals of Statistics}, 51(3):879--908.
    
    \bibitem[Gagnon-Bartsch and Speed, 2012]{Gagnon-Bartsch:2012}
    Gagnon-Bartsch, J.~A. and Speed, T.~P. (2012).
    \newblock Using control genes to correct for unwanted variation in microarray
      data.
    \newblock {\em Biostatistics}, 13(3):539--52.
    
    \bibitem[Gerard and Stephens, 2020]{gerard2020empirical}
    Gerard, D. and Stephens, M. (2020).
    \newblock {Empirical Bayes shrinkage and false discovery rate estimation,
      allowing for unwanted variation}.
    \newblock {\em Biostatistics}, 21(1):15--32.
    
    \bibitem[Guo et~al., 2022]{guo2022doubly}
    Guo, Z., {\'C}evid, D., and B{\"u}hlmann, P. (2022).
    \newblock Doubly debiased lasso: High-dimensional inference under hidden
      confounding.
    \newblock {\em The Annals of Statistics}, 50(3):1320.
    
    \bibitem[Hong and Iakoucheva, 2023]{Hong:2023}
    Hong, D. and Iakoucheva, L.~M. (2023).
    \newblock Therapeutic strategies for autism: targeting three levels of the
      central dogma of molecular biology.
    \newblock {\em Transl Psychiatry}, 13(1):58.
    
    \bibitem[Javanmard and Montanari, 2014]{javanmard2014confidence}
    Javanmard, A. and Montanari, A. (2014).
    \newblock Confidence intervals and hypothesis testing for high-dimensional
      regression.
    \newblock {\em The Journal of Machine Learning Research}, 15(1):2869--2909.
    
    \bibitem[Jiang and Ning, 2022]{jiang2022treatment}
    Jiang, K. and Ning, Y. (2022).
    \newblock Treatment effect estimation with unobserved and heterogeneous
      confounding variables.
    \newblock {\em arXiv preprint arXiv:2207.14439}.
    
    \bibitem[Kampmann, 2020]{Kampmann:2020}
    Kampmann, M. (2020).
    \newblock Crispr-based functional genomics for neurological disease.
    \newblock {\em Nat Rev Neurol}, 16(9):465--480.
    
    \bibitem[Lee et~al., 2017]{lee2017improved}
    Lee, S., Sun, W., Wright, F.~A., and Zou, F. (2017).
    \newblock An improved and explicit surrogate variable analysis procedure by
      coefficient adjustment.
    \newblock {\em Biometrika}, 104(2):303--316.
    
    \bibitem[Leek et~al., 2010]{Leek:2010}
    Leek, J.~T., Scharpf, R.~B., Bravo, H.~C., Simcha, D., Langmead, B., Johnson,
      W.~E., Geman, D., Baggerly, K., and Irizarry, R.~A. (2010).
    \newblock Tackling the widespread and critical impact of batch effects in
      high-throughput data.
    \newblock {\em Nat Rev Genet}, 11(10):733--9.
    
    \bibitem[Leek and Storey, 2007]{Leek:2007}
    Leek, J.~T. and Storey, J.~D. (2007).
    \newblock Capturing heterogeneity in gene expression studies by surrogate
      variable analysis.
    \newblock {\em PLoS Genet}, 3(9):1724--35.
    
    \bibitem[Leek and Storey, 2008]{leek2008general}
    Leek, J.~T. and Storey, J.~D. (2008).
    \newblock A general framework for multiple testing dependence.
    \newblock {\em Proceedings of the National Academy of Sciences},
      105(48):18718--18723.
    
    \bibitem[Lin et~al., 2021]{lin2021exponential}
    Lin, K.~Z., Lei, J., and Roeder, K. (2021).
    \newblock {Exponential-family embedding with application to cell developmental
      trajectories for single-cell RNA-seq data}.
    \newblock {\em Journal of the American Statistical Association},
      116(534):457--470.
    
    \bibitem[Lin et~al., 2024]{lin2023esvd}
    Lin, K.~Z., Qiu, Y., and Roeder, K. (2024).
    \newblock esvd-de: Cohort-wide differential expression in single-cell rna-seq
      data using exponential-family embeddings.
    \newblock {\em BMC bioinformatics}, 25(1):113.
    
    \bibitem[Love et~al., 2014]{love2014moderated}
    Love, M.~I., Huber, W., and Anders, S. (2014).
    \newblock Moderated estimation of fold change and dispersion for {RNA}-seq data
      with {DESeq2}.
    \newblock {\em Genome Biology}, 15(12):550.
    
    \bibitem[McKennan and Nicolae, 2019]{mckennan2019accounting}
    McKennan, C. and Nicolae, D. (2019).
    \newblock Accounting for unobserved covariates with varying degrees of
      estimability in high-dimensional biological data.
    \newblock {\em Biometrika}, 106(4):823--840.
    
    \bibitem[McKennan and Nicolae, 2022]{mckennan2022estimating}
    McKennan, C. and Nicolae, D. (2022).
    \newblock Estimating and accounting for unobserved covariates in
      high-dimensional correlated data.
    \newblock {\em Journal of the American Statistical Association},
      117(537):225--236.
    
    \bibitem[Ouyang et~al., 2023]{ouyang2022high}
    Ouyang, J., Tan, K.~M., and Xu, G. (2023).
    \newblock High-dimensional inference for generalized linear models with hidden
      confounding.
    \newblock {\em Journal of Machine Learning Research}, 24(296):1--61.
    
    \bibitem[Owen and Wang, 2016]{owen2016bi}
    Owen, A.~B. and Wang, J. (2016).
    \newblock Bi-cross-validation for factor analysis.
    \newblock {\em Statistical Science}, pages 119--139.
    
    \bibitem[Perez et~al., 2022]{perez2022single}
    Perez, R.~K., Gordon, M.~G., Subramaniam, M., Kim, M.~C., Hartoularos, G.~C.,
      Targ, S., Sun, Y., Ogorodnikov, A., Bueno, R., Lu, A., et~al. (2022).
    \newblock {Single-cell RNA-seq reveals cell type-specific molecular and genetic
      associations to lupus}.
    \newblock {\em Science}, 376(6589):eabf1970.
    
    \bibitem[Salim et~al., 2022]{salim2022ruv}
    Salim, A., Molania, R., Wang, J., De~Livera, A., Thijssen, R., and Speed, T.~P.
      (2022).
    \newblock Ruv-iii-nb: Normalization of single cell rna-seq data.
    \newblock {\em Nucleic Acids Research}, 50(16):e96--e96.
    
    \bibitem[Sarkar and Stephens, 2021]{sarkar2021separating}
    Sarkar, A. and Stephens, M. (2021).
    \newblock {Separating measurement and expression models clarifies confusion in
      single-cell RNA sequencing analysis}.
    \newblock {\em Nature genetics}, 53(6):770--777.
    
    \bibitem[Sayols, 2023]{sayols2023rrvgo}
    Sayols, S. (2023).
    \newblock rrvgo: a bioconductor package for interpreting lists of gene ontology
      terms.
    \newblock {\em Micropublication Biology}, 2023.
    
    \bibitem[Squair et~al., 2021]{Squair:2021}
    Squair, J.~W., Gautier, M., Kathe, C., Anderson, M.~A., James, N.~D., Hutson,
      T.~H., Hudelle, R., Qaiser, T., Matson, K. J.~E., Barraud, Q., Levine, A.~J.,
      La~Manno, G., Skinnider, M.~A., and Courtine, G. (2021).
    \newblock Confronting false discoveries in single-cell differential expression.
    \newblock {\em Nat Commun}, 12(1):5692.
    
    \bibitem[Sun et~al., 2023]{sun2022decorrelating}
    Sun, Y., Ma, L., and Xia, Y. (2023).
    \newblock A decorrelating and debiasing approach to simultaneous inference for
      high-dimensional confounded models.
    \newblock {\em Journal of the American Statistical Association}, pages 1--12.
    
    \bibitem[Sun et~al., 2012]{sun2012multiple}
    Sun, Y., Zhang, N.~R., and Owen, A.~B. (2012).
    \newblock Multiple hypothesis testing adjusted for latent variables, with an
      application to the agemap gene expression data.
    \newblock {\em The Annals of Applied Statistics}, pages 1664--1688.
    
    \bibitem[van~de Geer et~al., 2014]{van2014asymptotically}
    van~de Geer, S., B{\"u}hlmann, P., Ritov, Y., and Dezeure, R. (2014).
    \newblock On asymptotically optimal confidence regions and tests for
      high-dimensional models.
    \newblock {\em The Annals of Statistics}, 42(3):1166--1202.
    
    \bibitem[Vershynin, 2018]{vershynin2018high}
    Vershynin, R. (2018).
    \newblock {\em High-dimensional probability: An introduction with applications
      in data science}, volume~47.
    \newblock Cambridge university press.
    
    \bibitem[Wainwright, 2019]{wainwright2019high}
    Wainwright, M.~J. (2019).
    \newblock {\em High-dimensional statistics: A non-asymptotic viewpoint},
      volume~48.
    \newblock Cambridge university press.
    
    \bibitem[Wang et~al., 2017a]{wang2017confounder}
    Wang, J., Zhao, Q., Hastie, T., and Owen, A.~B. (2017a).
    \newblock Confounder adjustment in multiple hypothesis testing.
    \newblock {\em The Annals of Statistics}, 45(5):1863.
    
    \bibitem[Wang et~al., 2017b]{wang2017unified}
    Wang, L., Zhang, X., and Gu, Q. (2017b).
    \newblock A unified computational and statistical framework for nonconvex
      low-rank matrix estimation.
    \newblock In {\em Artificial Intelligence and Statistics}, pages 981--990.
      PMLR.
    
    \bibitem[Wolf et~al., 2018]{wolf2018scanpy}
    Wolf, F.~A., Angerer, P., and Theis, F.~J. (2018).
    \newblock {SCANPY: large-scale single-cell gene expression data analysis}.
    \newblock {\em Genome biology}, 19:1--5.
    
    \bibitem[Zappia et~al., 2017]{zappia2017splatter}
    Zappia, L., Phipson, B., and Oshlack, A. (2017).
    \newblock {Splatter: simulation of single-cell RNA sequencing data}.
    \newblock {\em Genome biology}, 18(1):174.
    
\end{thebibliography}
\end{document}